\documentclass{article}%
\usepackage{amsmath}
\usepackage{amssymb}
\usepackage{amsfonts}
\usepackage{graphicx}%
\providecommand{\U}[1]{\protect\rule{.1in}{.1in}}
\newtheorem{theorem}{Theorem}

\newtheorem{corollary}[theorem]{Corollary}

\newtheorem{definition}[theorem]{Definition}

\newtheorem{proposition}[theorem]{Proposition}
\newtheorem{remark}[theorem]{Remark}

\newenvironment{proof}[1][Proof]{\noindent\textbf{#1.} }{\ \rule{0.5em}{0.5em}}
\begin{document}

\title{Differential Structure of the Hyperbolic Clifford Algebra }
\author{Eduardo A. Notte-Cuello$^{(1)}$ and Waldyr A. Rodrigues Jr.$^{(2)}$\\$^{(1)}$~{\small Departamento de} {\small Matem\'{a}ticas,}\\{\small Universidad de La Serena}\\{\small Av. Cisternas 1200, La Serena-Chile}\\$^{(2)}~${\footnotesize Institute of Mathematics, Statistics and Scientific
Computation}\\{\footnotesize IMECC-UNICAMP CP 6065}\\{\footnotesize 13083-859 Campinas, SP, Brazil}\\{\footnotesize \texttt{walrod@ime.unicamp.br; {\small enotte@userena.cl;
}roldao.rocha@ufabc.edu.br }}}
\maketitle

\begin{abstract}
This paper presents a thoughful review of: (a) the Clifford algebra
$\mathcal{C}l\left(  H_{V}\right)  $ of \emph{multivecfors }which is naturally
associated with a hyperbolic space $H_{V}$; (b) the study of the properties of
the duality product of \emph{multivectors} and \emph{multiforms;} (c) the
theory of $k$ multivector and $l$ multiform variables multivector extensors
over $V$ and (d) the use of the above mentioned structures to present a theory
of the parallelism structure on an arbitrary smooth manifold introducing the
concepts of covariant derivarives, deformed covariant derivatives and relative
covariant derivatives of multivector, multiform fields and extensors fields.

\end{abstract}


\section{Introduction}

In this paper we first present a review of the Hyperbolic Clifford algebra of
a real $n$-dimensinal space $V$. We emphazise that the hyperbolic Clifford
algebra is a very important structure for modern theories of Physics, in
particular superfields can be seen as ideal sections of an hyperbolic bundel
over a spacetime $M$, see, e.g.,
\cite{quinrod,qr07,Crumeyrolle,Geroch1,Geroch2,Milnor}. Besides given a
thoughtful exposition of the Clifford algebra $\mathcal{C}l\left(
H_{V}\right)  $ of \emph{multivecfors} which is naturally associated with a
hyperbolic space $H_{V}$, we review also the properties of the \emph{duality
product} of multivectors and multiforms and the theory of $k$ multivector and
$l$ multiform variables multivector extensors over $V.$ The algebraic theory
is then used to review with enough details the theory of parallelism
structures on an arbitrary smooth manifold $M$, and the concepts of covariant
derivatives, deformed derivatives and relative covariant derivatives of
multivector, multiform field and extensors fields (all of crucial importance
in several physical theories, see, e.g., \cite{fr2010}).

\section{Hyperbolic Spaces}

\subsection{Definition and Basic Properties}

Let $V,V^{\ast}$ be a pair of dual $n$-dimensional vector spaces over the real
field $\mathbb{R}$ and $V\oplus V^{\ast}$ the exterior direct sum of the
vector spaces $V\,\ $and $V^{\ast}$. We call hyperbolic structure over $V$ to
the pair
\[
H_{V}=\left(  V\oplus V^{\ast},\left\langle ,\right\rangle \right)
\]
where $\left\langle ,\right\rangle $ is the non-degenerate symmetric bilinear
form of index $n$ defined by%
\begin{equation}
\left\langle \mathbf{x},\mathbf{y}\right\rangle =x^{\ast}\left(  y_{\ast
}\right)  +y^{\ast}\left(  x_{\ast}\right)  \label{a1}%
\end{equation}
for all $\mathbf{x}=x_{\ast}\oplus x^{\ast},\mathbf{y}=y_{\ast}\oplus y^{\ast
}\in H_{V}.$

\begin{definition}
The elements of $H_{V}$ will be referred as \textit{vecfors. }A vecfor
$\mathbf{x}=x_{\ast}\oplus x^{\ast}\in H_{V}$ will be said to be positive if
$x^{\ast}\left(  x_{\ast}\right)  >0,$ null if $x^{\ast}\left(  x_{\ast
}\right)  =0$ and negative if $x^{\ast}\left(  x_{\ast}\right)  <0.$ If
$x^{\ast}\left(  x_{\ast}\right)  =1$ we say that $\mathbf{x}=x_{\ast}\oplus
x^{\ast}$ is a unit vecfor.
\end{definition}

Let $H_{V}^{\ast}$ denote the dual space of $H_{V}$, i.e., $H_{V}^{\ast
}=\left(  \left(  V\oplus V^{\ast}\right)  ^{\ast},\left\langle ,\right\rangle
^{-1}\right)  ,$ where $\left\langle ,\right\rangle ^{-1}$ is the reciprocal
of the bilinear form $\left\langle ,\right\rangle .$ We have the following
isomorphisms%
\[
H_{V}^{\ast}\simeq H_{V^{\ast}}\simeq H_{V}%
\]
where $H_{V^{\ast}}$ is the hyperbolic space of $V^{\ast}.$ Therefore, $H_{V}$
is an "auto-dual" space.

The spaces $V$ and $V^{\ast}$ are naturally identified with their images in
$H_{V}$ under the inclusions $i_{\ast}:V\rightarrow V\oplus V^{\ast},x_{\ast
}\mapsto i_{\ast}x_{\ast}=x_{\ast}\oplus0\equiv x_{\ast}$ and $i^{\ast
}:V^{\ast}\rightarrow V\oplus V^{\ast},x^{\ast}\mapsto i^{\ast}x^{\ast
}=0\oplus x^{\ast}\equiv x^{\ast},$ then from (\ref{a1}) we have%
\[
\left\langle i_{\ast}x_{\ast},i_{\ast}y_{\ast}\right\rangle =\left\langle
x_{\ast},y_{\ast}\right\rangle =0,\quad\left\langle i^{\ast}x^{\ast},i^{\ast
}y^{\ast}\right\rangle =\left\langle x^{\ast},y^{\ast}\right\rangle
=0,\quad\left\langle i^{\ast}x^{\ast},i_{\ast}y_{\ast}\right\rangle =x^{\ast
}\left(  y_{\ast}\right)
\]
for all $x_{\ast},y_{\ast}\in V\subset V\oplus V^{\ast}$ and $x^{\ast}%
,y^{\ast}\in V^{\ast}\subset V\oplus V^{\ast}.$ This means that $V$ and
$V^{\ast}$ are maximal totally isotropic subspace of $H_{V}$ and any pair of
dual basis satisfying
\[
\left\langle e_{i},e_{j}\right\rangle =0,\quad\left\langle \theta^{i}%
,\theta^{j}\right\rangle =0\quad\text{and}\quad\left\langle \theta^{i}%
,e_{j}\right\rangle =\delta_{j}^{i}.
\]

\begin{remark}
More generally, to each subspace $S\subset V$ we can associated a maximal
totally isotropic subspace $I\left(  S\right)  \subset H_{V}$, see \cite{qr07}.
\end{remark}

To each Witt basis $\left\{  e_{1},\cdots,e_{n},\theta^{1},\cdots,\theta
^{n}\right\}  $ of $H_{V},$ we can associated an orthonormal basis $\left\{
\mathbf{\sigma}_{1},\cdots,\mathbf{\sigma}_{2n}\right\}  $ of $H_{V}$ by
letting%
\begin{equation}
\mathbf{\sigma}_{k}=\frac{1}{\sqrt{2}}\left(  e_{k}\oplus\theta^{k}\right)
\quad\text{and\quad}\mathbf{\sigma}_{n+k}=\frac{1}{\sqrt{2}}\left(
\overline{e}_{k}\oplus\theta^{k}\right)  \label{a2}%
\end{equation}
where $\overline{e}_{k}=-e_{k},k=1,...,n.$ We have for $k,l=1,...,n,$%
\[
\left\langle \mathbf{\sigma}_{k},\mathbf{\sigma}_{l}\right\rangle =\delta
_{kl},\quad\left\langle \mathbf{\sigma}_{n+k},\mathbf{\sigma}_{n+l}%
\right\rangle =-\delta_{kl}\quad\text{and}\quad\left\langle \mathbf{\sigma
}_{k},\mathbf{\sigma}_{n+l}\right\rangle =0.
\]
If $\mathbf{x}=x_{\ast}\oplus x^{\ast},$ with $x_{\ast}=x_{\ast}^{k}e_{k}$ and
$x^{\ast}=x_{k}^{\ast}\theta^{k},$ then the components $\left(  x^{k}%
,x^{n+k}\right)  $ of $\mathbf{x}$ with respect to $\mathbf{\sigma}_{k}$ are
\begin{equation}
x^{k}=\frac{1}{\sqrt{2}}\left(  x_{k}^{\ast}+x_{\ast}^{k}\right)
\quad\text{and}\quad x^{n+k}=\frac{1}{\sqrt{2}}\left(  x_{k}^{\ast}-x_{\ast
}^{k}\right)  . \label{a3}%
\end{equation}
Then, base vectors $\mathbf{\sigma}_{n+l}$ are obtained from the
$\mathbf{\sigma}_{k}$ through the involution $\mathbf{x}=x_{\ast}\oplus
x^{\ast}\mapsto\overline{\mathbf{x}}=\left(  -x_{\ast}\right)  \oplus x^{\ast
},$ where $\overline{\mathbf{x}}$ is the (hyperbolic) conjugate of the vector
$\mathbf{x.}$

Due to the auto-duality of $H_{V}$ stated by the isomorphisms $H_{V}^{\ast
}\simeq H_{V^{\ast}}\simeq H_{V}$, we can use the notation $\left\{
\mathbf{\sigma}^{1},\cdots,\mathbf{\sigma}^{2n}\right\}  $ to indicate the
dual basis $\left\{  \mathbf{\sigma}_{1},\cdots,\mathbf{\sigma}_{2n}\right\}
$ as well the reciprocal basis of this same basis. As elements of $H_{V}%
^{\ast}$ the expressions of the $\mathbf{\sigma}^{k\prime}$ are
\[
\mathbf{\sigma}^{k}=\frac{1}{\sqrt{2}}\left(  \theta^{k}\oplus e_{k}\right)
\quad\text{and\quad}\mathbf{\sigma}_{k}=\frac{1}{\sqrt{2}}\left(
\overline{\theta}^{k}\oplus e_{k}\right)  .
\]
Another remarkable result on the theory of hyperbolic space is the following
(see \cite{Knus}).

\begin{proposition}
Given an arbitrary non-degenerate symmetric bilinear form $b$ on $V,$ there is
an isomorphism $H_{V}\simeq\left(  V,b\right)  \oplus\left(  V,-b\right)
=H_{bV}$.
\end{proposition}

\subsection{Exterior Algebra of a Hyperbolic Space}

The Grassman algebra $\bigwedge H_{V}$ of the hyperbolic structure $H_{V}$ is
the pair
\[
\bigwedge H_{V}=\left(  \bigwedge\left(  V\oplus V^{\ast}\right)
,\left\langle ,\right\rangle \right)
\]
where%
\[
\bigwedge\left(  V\oplus V^{\ast}\right)  =\sum\limits_{r=0}^{2n}%
\bigwedge\nolimits^{r}\left(  V\oplus V^{\ast}\right)
\]
is the exterior algebra of $V\oplus V^{\ast}$ and $\left\langle ,\right\rangle
$ is the canonical bilinear form on $\bigwedge\left(  V\oplus V^{\ast}\right)
$ induced by the bilinear form $\left\langle ,\right\rangle $ of $H_{V}$,
e.g., for simple elements $\mathbf{x}_{1}\wedge...\wedge\mathbf{x}%
_{r},\mathbf{y}_{1}\wedge...\wedge\mathbf{y}_{r}\in\bigwedge^{r}%
H_{V},\left\langle ,\right\rangle $ is given by%
\begin{equation}
\left\langle \mathbf{x}_{1}\wedge...\wedge\mathbf{x}_{r},\mathbf{y}_{1}%
\wedge...\wedge\mathbf{y}_{r}\right\rangle =\det\left(
\begin{array}
[c]{ccc}%
\left\langle \mathbf{x}_{1},\mathbf{y}_{1}\right\rangle  & \cdots &
\left\langle \mathbf{x}_{1},\mathbf{y}_{r}\right\rangle \\
. & \cdots & .\\
. & \cdots & .\\
\left\langle \mathbf{x}_{r},\mathbf{y}_{1}\right\rangle  & \cdots &
\left\langle \mathbf{x}_{r},\mathbf{y}_{r}\right\rangle
\end{array}
\right)  \label{a4}%
\end{equation}
and it is extended by linearity and orthogonality to all of the algebra
$\bigwedge H_{V}.$ Of course, due to the isomorphisms $H_{V}^{\ast}\simeq
H_{V^{\ast}}\simeq H_{V},$ we have%
\[
\bigwedge H_{V}^{\ast}\simeq\bigwedge H_{V^{\ast}}\simeq\bigwedge H_{V},
\]
and it follows that%
\[
\bigwedge H_{V}\simeq\left(  \bigwedge H_{V}\right)  ^{\ast},
\]
i.e., $\bigwedge H_{V}$ is itself an auto-dual space. The elements of
$\bigwedge H_{V}$ will be called \textit{multivecfors. }

Grade involution, reversion and conjugation in the algebra $\bigwedge H_{V}$
are defined as usual, see \cite{rodoliv2006}. For a homogeneous multivecfors
$\mathbf{u\in}\bigwedge^{r}H_{V},$%
\[
\widehat{\mathbf{u}}=\left(  -1\right)  ^{r}\mathbf{u,\quad}\widetilde
{\mathbf{u}}=\left(  -1\right)  ^{\frac{1}{2}r\left(  r-1\right)
}\mathbf{u,\quad}\overline{\mathbf{u}}=\left(  -1\right)  ^{\frac{1}%
{2}r\left(  r+1\right)  }\mathbf{u}%
\]
and we call that every element $\mathbf{u\in}\bigwedge H_{V}$ is uniquely
decomposed into a sum the even and odd part of $\mathbf{u}$, i.e.,
$\mathbf{u=u}_{+}+\mathbf{u}_{-}$, where
\[
\mathbf{u}_{+}=\frac{1}{2}\left(  \mathbf{u}+\widehat{\mathbf{u}}\right)
\quad\text{and}\quad\mathbf{u}_{-}=\frac{1}{2}\left(  \mathbf{u}%
-\widehat{\mathbf{u}}\right)  .
\]
The spaces $\bigwedge V$ and $\bigwedge V^{\ast}$ are identified with their
images in $\bigwedge H_{V}$ under the homomorphisms $i^{w}:\bigwedge
V\rightarrow\bigwedge H_{V}$ defined by%
\[
i^{w}\left(  x_{1\ast}\wedge...\wedge x_{r\ast}\right)  =\left(  x_{1\ast
}\oplus0\right)  \wedge...\wedge\left(  x_{r\ast}\oplus0\right)  \equiv
x_{1\ast}\wedge...\wedge x_{r\ast}%
\]
and $i_{w}:\bigwedge V^{\ast}\rightarrow\bigwedge H_{V}$ defined by
\[
i_{w}\left(  x_{1}^{\ast}\wedge...\wedge x_{r}^{\ast}\right)  =\left(  0\oplus
x_{1}^{\ast}\right)  \wedge...\wedge\left(  0\oplus x_{r}^{\ast}\right)
\equiv x_{1}^{\ast}\wedge...\wedge x_{r}^{\ast}.
\]
Then, for $u_{\ast},v_{\ast}\in\bigwedge V\subset\bigwedge H_{V}$ and
$u^{\ast},v^{\ast}\in\bigwedge V^{\ast}\subset\bigwedge H_{V},$ we have%
\begin{equation}
\left\langle u_{\ast},v_{\ast}\right\rangle =0,\quad\left\langle u^{\ast
},v^{\ast}\right\rangle =0\quad\text{and}\quad\left\langle u^{\ast},v_{\ast
}\right\rangle =u^{\ast}\left(  v_{\ast}\right)  . \label{a5}%
\end{equation}
Thus, $\bigwedge V$ and $\bigwedge V^{\ast}$ are totally isotropic subspace of
$\bigwedge H_{V}.$ But they are no longer maximal, $\left\langle
,\right\rangle $ being neutral, the dimension of a maximal totally isotropic
subspace is $2^{2n-1},$ whereas $\dim\bigwedge V=\dim\bigwedge V^{\ast}%
=2^{n}.$ For elements $\mathbf{u}=u_{\ast}\wedge u^{\ast},\mathbf{v}=v_{\ast
}\wedge v^{\ast}\in\bigwedge H_{V}$ with $u_{\ast}\in\bigwedge^{r}V,u^{\ast
}\in\bigwedge^{s}V^{\ast},v_{\ast}\in\bigwedge^{s}V$ and $v^{\ast}\in
\bigwedge^{r}V^{\ast}$ it holds%
\[
\left\langle \mathbf{u},\mathbf{v}\right\rangle =\left(  -1\right)
^{rs}u^{\ast}\left(  v_{\ast}\right)  v^{\ast}\left(  u_{\ast}\right)  .
\]

\begin{proposition}
There is the following natural isomorphism%
\[%
{\textstyle\bigwedge}
H_{V}\simeq%
{\textstyle\bigwedge}
V\widehat{\otimes}%
{\textstyle\bigwedge}
V^{\ast}%
\]
where $\widehat{\otimes}$ denotes the graded tensor product. Moreover, being
$b$ a non-degenerate bilinear form on $V$, it holds also%
\[%
{\textstyle\bigwedge}
H_{V}\simeq%
{\textstyle\bigwedge}
H_{bV}\widehat{\otimes}%
{\textstyle\bigwedge}
H_{-bV}.
\]
For a proof see, e.g., Greub \cite{Greub}. The first of the above isomorphisms
is given by the mapping $\bigwedge V\widehat{\otimes}\bigwedge V^{\ast
}\rightarrow\bigwedge H_{V}$ by%
\[
u_{\ast}\widehat{\otimes}u^{\ast}\mapsto i^{w}u_{\ast}\wedge i_{w}u^{\ast},
\]
for all $u_{\ast}\in\bigwedge V$ and $u^{\ast}\in\bigwedge V^{\ast}$. Under
this mapping, we can make the identification%
\[%
{\textstyle\bigwedge\nolimits^{r}}
H_{V}=%
{\textstyle\sum\limits_{p+q=r}}
{\textstyle\bigwedge\nolimits^{p}}
V\widehat{\otimes}%
{\textstyle\bigwedge\nolimits^{q}}
V^{\ast}.
\]

\end{proposition}

\subsection{Orientation}

Besides a canonical metric structure, a hyperbolic space is also provided with
a canonical sense of orientation, induced by the $2n$-vector $\mathbf{\sigma
\in}\bigwedge\nolimits^{2n}H_{V}$ given by%
\[
\mathbf{\sigma=\sigma}_{1}\wedge...\wedge\mathbf{\sigma}_{2n}%
\]
where $\left\{  \mathbf{\sigma}_{1},\cdots,\mathbf{\sigma}_{2n}\right\}  $ is
the orthonormal basis of $H_{V}$ naturally associated with the dual basis
$\left\{  e_{1},\cdots,e_{n}\right\}  $ of $V$ and $\left\{  \theta^{1}%
,\cdots,\theta^{n}\right\}  $ of $V^{\ast}.$ Note that from Eq.(\ref{a4}) it
follows immediately that%
\[
\left\langle \mathbf{\sigma,\sigma}\right\rangle =\left(  -1\right)  ^{n}.
\]
\ 

The reason for saying that $\mathbf{\sigma}$ define a canonical sense of
orientation is that it is independent on the choice of the bases for $V.$\ To
see this it is enough to verify that%
\[
\mathbf{\sigma}=e_{\ast}\wedge\theta^{\ast}%
\]
where
\[
e_{\ast}=e_{1}\wedge\cdots\wedge e_{n}\quad\text{and\quad}\theta^{\ast}%
=\theta^{1}\wedge\cdots\wedge\theta^{n}.
\]
Thus, under a change of basis in $V,$ $e_{\ast}$ transforms as $\lambda
e_{\ast},$ with $\lambda\neq0,$ whereas $\theta^{\ast}$ transforms as
$\lambda^{-1}\theta^{\ast},$ so that $\mathbf{\sigma}$ remains unchanged.

\subsection{Contractions}

A left contraction $\lrcorner:%
{\textstyle\bigwedge}
H_{V}\times%
{\textstyle\bigwedge}
H_{V}\rightarrow%
{\textstyle\bigwedge}
H_{V}$ and a right contraction $\llcorner:%
{\textstyle\bigwedge}
H_{V}\times%
{\textstyle\bigwedge}
H_{V}\rightarrow%
{\textstyle\bigwedge}
H_{V}$ are introduced in the algebra $%
{\textstyle\bigwedge}
H_{V}$ in the usual way (see, e.g., \cite{rodoliv2006}), i.e., by%
\[
\left\langle \mathbf{u}\lrcorner\mathbf{v},\mathbf{w}\right\rangle
=\left\langle \mathbf{v},\widetilde{\mathbf{u}}\wedge\mathbf{w}\right\rangle
\quad and\quad\left\langle \mathbf{v}\llcorner\mathbf{u},\mathbf{w}%
\right\rangle =\left\langle \mathbf{v},\mathbf{w}\wedge\widetilde{\mathbf{u}%
}\right\rangle ,
\]
for all $\mathbf{u,v,w\in}%
{\textstyle\bigwedge}
H_{V}.$ These operations have the general properties, for all $\mathbf{x,y\in
}H_{V}$, for all $\mathbf{u,v,w\in}%
{\textstyle\bigwedge}
H_{V}$ and $\mathbf{\sigma\in}\bigwedge\nolimits^{2n}H_{V},$%
\[%
\begin{array}
[c]{l}%
\mathbf{x}\lrcorner\mathbf{y}=\mathbf{x}\llcorner\mathbf{y=}\left\langle
\mathbf{x},\mathbf{y}\right\rangle \quad\text{;}\quad1\lrcorner\mathbf{u=u}%
\llcorner=\mathbf{u},\quad\mathbf{x}\lrcorner=1\llcorner\mathbf{x}=0\bigskip\\
\left(  \mathbf{u}\lrcorner\mathbf{v}\right)  ^{\wedge}=\widehat{\mathbf{u}%
}\lrcorner\widehat{\mathbf{v}},\quad\left(  \mathbf{u}\llcorner\mathbf{v}%
\right)  ^{\wedge}=\widehat{\mathbf{u}}\llcorner\widehat{\mathbf{v}}%
\quad\text{;}\quad\left(  \mathbf{u}\lrcorner\mathbf{v}\right)  ^{\sim
}=\widetilde{\mathbf{u}}\lrcorner\widetilde{\mathbf{v}},\quad\left(
\mathbf{u}\llcorner\mathbf{v}\right)  ^{\sim}=\widetilde{\mathbf{u}}%
\llcorner\widetilde{\mathbf{v}}\bigskip\\
\mathbf{u}\lrcorner\left(  \mathbf{v}\lrcorner\mathbf{w}\right)  =\left(
\mathbf{u}\wedge\mathbf{v}\right)  \lrcorner\mathbf{w;}\text{ }\left(
\mathbf{u}\llcorner\mathbf{v}\right)  \llcorner\mathbf{w}=\mathbf{u}%
\llcorner\left(  \mathbf{v}\wedge\mathbf{w}\right)  \mathbf{,}\text{ }\left(
\mathbf{u}\lrcorner\mathbf{v}\right)  \llcorner\mathbf{w}=\mathbf{u}%
\lrcorner\left(  \mathbf{v}\llcorner\mathbf{w}\right)  \bigskip\\
\mathbf{x}\lrcorner\left(  \mathbf{v}\wedge\mathbf{w}\right)  =\left(
\mathbf{x}\lrcorner\mathbf{v}\right)  \wedge\mathbf{w+}\widehat{\mathbf{v}%
}\wedge\left(  \mathbf{x\lrcorner w}\right)  \quad\text{;}\quad\left(
\mathbf{u}\wedge\mathbf{v}\right)  \llcorner\mathbf{x}=\mathbf{u\wedge}\left(
\mathbf{\mathbf{v}}\llcorner\mathbf{x}\right)  \mathbf{+}\left(
\mathbf{u\llcorner x}\right)  \wedge\widehat{\mathbf{v}}\bigskip\\
\mathbf{x}\wedge\left(  \mathbf{v}\lrcorner\mathbf{w}\right)  =\widehat
{\mathbf{v}}\lrcorner\left(  \mathbf{x\wedge w}\right)  -\left(
\widehat{\mathbf{v}}\llcorner\mathbf{x}\right)  \lrcorner\mathbf{w}%
\quad\text{;}\quad\left(  \mathbf{u\llcorner v}\right)  \wedge\mathbf{x}%
=\left(  \mathbf{u\wedge x}\right)  \llcorner\widehat{\mathbf{v}%
}\mathbf{-u\llcorner}\left(  \mathbf{x}\lrcorner\widehat{\mathbf{v}}\right)
\bigskip\\
\mathbf{u}_{+}\lrcorner\mathbf{v=v\llcorner u}_{+}\quad\text{;}\quad
\mathbf{u}_{-}\lrcorner\mathbf{v=}\widehat{\mathbf{v}}\mathbf{\llcorner
}\widehat{\mathbf{u}}_{-}\bigskip\\
\mathbf{u}\wedge\left(  \mathbf{v}\lrcorner\mathbf{\sigma}\right)  =\left(
\mathbf{u}\lrcorner\mathbf{v}\right)  \lrcorner\mathbf{\sigma}\quad
\text{;}\quad\left(  \mathbf{\sigma\llcorner v}\right)  \wedge\mathbf{w=\sigma
\llcorner}\left(  \mathbf{v}\llcorner\mathbf{w}\right)
\end{array}
\]
Moreover, from Eqs.(\ref{a5}) we get
\begin{equation}
u_{\ast}\lrcorner v_{\ast}=u_{\ast}\llcorner v_{\ast}=0,\quad u^{\ast
}\lrcorner v^{\ast}=u^{\ast}\llcorner v^{\ast}=0\label{a6}%
\end{equation}
for all $u_{\ast},v_{\ast}\mathbf{\in}%
{\textstyle\bigwedge}
V$ and $u^{\ast},v^{\ast}\mathbf{\in}%
{\textstyle\bigwedge}
V^{\ast},$ so that for elements of the form $\mathbf{u}=u_{\ast}\wedge
u^{\ast}$ and $\mathbf{x}=x_{\ast}\oplus x^{\ast},$ it holds%
\[%
\begin{array}
[c]{c}%
\mathbf{x}\lrcorner\mathbf{u}=\left(  x^{\ast}\lrcorner u_{\ast}\right)
\wedge u^{\ast}+\widehat{u}_{\ast}\wedge\left(  x_{\ast}\lrcorner u^{\ast
}\right)  \medskip\\
\mathbf{u}\llcorner\mathbf{x}=u_{\ast}\wedge\left(  u^{\ast}\llcorner x_{\ast
}\right)  +\left(  x^{\ast}\lrcorner u_{\ast}\right)  \wedge\widehat{u}^{\ast}%
\end{array}
\]
For more details about the properties of the left and right contractions, see,
e.g., \cite{Lounesto,rodoliv2006}.

\subsection{Poincar\'{e} Automorphism (Hodge Dual)}

Define now the Poincar\'{e} automorphism or Hodge dual $\star:%
{\textstyle\bigwedge}
H_{V}\rightarrow%
{\textstyle\bigwedge}
H_{V}$ by%
\begin{equation}
\star\mathbf{u=}\widetilde{\mathbf{u}}\lrcorner\mathbf{\sigma,} \label{a7}%
\end{equation}
for all $\mathbf{u\in}%
{\textstyle\bigwedge}
H_{V}.$ The inverse $\star^{-1}$ of this operation is given by%
\begin{equation}
\star^{-1}\mathbf{u=}\widetilde{\mathbf{\sigma}}\mathbf{\llcorner
\widetilde{\mathbf{u}}.} \label{a8}%
\end{equation}
The following general properties of the Hodge duality holds true%
\[%
\begin{array}
[c]{l}%
\star\mathbf{\sigma=}\left(  -1\right)  ^{n},\quad\star^{-1}\mathbf{\sigma
=}1,\quad\left\langle \star\mathbf{u,}\star\mathbf{v}\right\rangle =\left(
-1\right)  ^{n}\left\langle \mathbf{u,v}\right\rangle ,\medskip\\
\star\left(  \mathbf{u\wedge v}\right)  =\widetilde{\mathbf{v}}\lrcorner
\star\mathbf{u},\quad\star^{-1}\left(  \mathbf{u\wedge v}\right)  =\left(
\star^{-1}\mathbf{v}\right)  \llcorner\widetilde{\mathbf{u,}}\medskip\\
\star\left(  \mathbf{u\llcorner v}\right)  =\widetilde{\mathbf{v}}\wedge
\star\mathbf{u},\quad\star^{-1}\left(  \mathbf{u\lrcorner v}\right)  =\left(
\star^{-1}\mathbf{v}\right)  \wedge\widetilde{\mathbf{u}}.
\end{array}
\]

For $\mathbf{x}=x_{\ast}\oplus x^{\ast}\in H_{V}\subset%
{\textstyle\bigwedge}
H_{V}$ we have, since $x_{\ast}\lrcorner e_{\ast}=0$ and $x^{\ast}%
\lrcorner\theta^{\ast}=0$ that%
\[%
\begin{array}
[c]{l}%
\star\mathbf{x=x\lrcorner\sigma=x\lrcorner}\left(  e_{\ast}\wedge\theta^{\ast
}\right)  =\left(  \mathbf{x\lrcorner}e_{\ast}\right)  \wedge\theta^{\ast
}+\widehat{e}_{\ast}\wedge\left(  \mathbf{x\lrcorner}\theta^{\ast}\right)
\medskip\\
=\left(  x^{\ast}\mathbf{\lrcorner}e_{\ast}\right)  \wedge\theta^{\ast
}-e_{\ast}\wedge\left(  \theta^{\ast}\llcorner x_{\ast}\right)  ,
\end{array}
\]
and it follows that, for $u_{\ast}\in%
{\textstyle\bigwedge}
V\subset%
{\textstyle\bigwedge}
H_{V}$ and $u^{\ast}\in%
{\textstyle\bigwedge}
V^{\ast}\subset%
{\textstyle\bigwedge}
H_{V},$%
\[
\star u^{\ast}\mathbf{=}\left(  \widetilde{u}^{\ast}\mathbf{\lrcorner}e_{\ast
}\right)  \wedge\theta^{\ast}=D_{\nshortparallel}u^{\ast}\wedge\theta^{\ast
}\quad\text{and}\quad\star u_{\ast}\mathbf{=}e_{\ast}\wedge\left(
\theta^{\ast}\llcorner\overline{u}_{\ast}\right)  =e_{\ast}\wedge
D^{\nshortparallel}u_{\ast}%
\]
where we introduced the Poincar\'{e} isomorphisms $D_{\nshortparallel}:%
{\textstyle\bigwedge}
V^{\ast}\rightarrow%
{\textstyle\bigwedge}
V$ and $D^{\nshortparallel}:%
{\textstyle\bigwedge}
V\rightarrow%
{\textstyle\bigwedge}
V^{\ast}$ by (see, e.g., \cite{Greub})
\[
D_{\nshortparallel}u^{\ast}=\widetilde{u}^{\ast}\mathbf{\lrcorner}e_{\ast
}\quad\text{and}\quad D^{\nshortparallel}u_{\ast}=\theta^{\ast}\llcorner
\overline{u}_{\ast}.
\]
For an element of the form $\mathbf{u}=u_{\ast}\wedge u^{\ast}$ with $u_{\ast
}\in%
{\textstyle\bigwedge}
V$ and $u^{\ast}\in%
{\textstyle\bigwedge}
V^{\ast},$ we have%
\[
\star\mathbf{u=}D_{\nshortparallel}u^{\ast}\wedge D^{\nshortparallel}u_{\ast
}.
\]

\subsection{Clifford Algebra of a Hyperbolic Space}

Introduce in $%
{\displaystyle\bigwedge}
H_{V}$ the Clifford product of a vector $\mathbf{x}\in H_{V}$ by an element
$\mathbf{u}\in%
{\textstyle\bigwedge}
H_{V}$ by
\[
\mathbf{xu}=\mathbf{x}\lrcorner\mathbf{u}+\mathbf{x}\wedge\mathbf{u}%
\]
and extend this product by linearity and associativity to all of the space $%
{\displaystyle\bigwedge}
H_{V}$. The resulting algebra is isomorphic to the Clifford algebra
$\mathcal{C\ell}(H_{V})$ of the hyperbolic structure $H_{V}$ and will thereby
be identified with it.

We call $\mathcal{C\ell}(H_{V})$ the \emph{mother} algebra\/(or the hyperbolic
Clifford algebra\/) of the vector space $V$. The even and odd subspaces of
$\mathcal{C\ell}(H_{V})$ will be denoted respectively by $\mathcal{C\ell
}^{(0)}(H_{V})$ and $\mathcal{C\ell}^{(1)}(H_{V})$, so that
\[
\mathcal{C\ell}(H_{V})=\mathcal{C\ell}^{(0)}(H_{V})\oplus\mathcal{C\ell}%
^{(1)}(H_{V})
\]
and the same notation of the exterior algebra is used for grade involution,
reversion, and conjugation in $\mathcal{C\ell}(H_{V})$, which obviously
satisfy
\[
(\mathbf{uv})\hat{\ }=\hat{\mathbf{u}}\hat{\mathbf{v}},\qquad(\mathbf{uv}%
)\tilde{\ }=\tilde{\mathbf{v}}\tilde{\mathbf{u}},\qquad(\mathbf{uv})\bar
{\ }=\bar{\mathbf{v}}\bar{\mathbf{u}}.
\]

For vecfors $\mathbf{x},\mathbf{y}\in H_{V}$, we have the relation
\[
\mathbf{xy}+\mathbf{yx}=2\boldsymbol{<}\mathbf{x},\mathbf{y}>.
\]

For the basis elements $\{\mathbf{\sigma}_{k}\}$ it holds%
\[%
\begin{array}
[c]{l}%
\mathbf{\sigma}_{k}\mathbf{\sigma}_{l}+\mathbf{\sigma}_{l}\mathbf{\sigma}%
_{k}=2\delta_{kl},\\
\mathbf{\sigma}_{n+k}\mathbf{\sigma}_{n+l}+\mathbf{\sigma}_{n+l}%
\mathbf{\sigma}_{n+k}=-2\delta_{kl},\\
\mathbf{\sigma}_{k}\mathbf{\sigma}_{n+l}=-\mathbf{\sigma}_{n+l}\mathbf{\sigma
}_{k}.
\end{array}
\]
In turn, for the Witt basis $\{e_{k},\theta^{k}\}$, we have instead
\begin{align*}
&  e_{k}e_{l}+e_{l}e_{k}=0,\\
&  \theta^{k}\theta^{l}+\theta^{l}\theta^{k}=0,\\
&  \theta^{k}e_{l}+e_{l}\theta^{k}=2\delta_{l}^{k}.
\end{align*}

The Clifford product has the following general properties:%
\[%
\begin{array}
[c]{l}%
\mathbf{u\lrcorner\sigma=u\sigma,\qquad\sigma\llcorner u=\sigma u,\medskip}\\
\left\langle \mathbf{u,vw}\right\rangle =\left\langle \widetilde{\mathbf{v}%
}\mathbf{u,w}\right\rangle =\left\langle \mathbf{u}\widetilde{\mathbf{w}%
}\mathbf{,v}\right\rangle ,\mathbf{\medskip}\\
\mathbf{x\wedge u=}\frac{1}{2}\left(  \mathbf{xu+}\widehat{\mathbf{u}%
}\mathbf{x}\right)  \mathbf{,\qquad u\wedge x=}\frac{1}{2}\left(
\mathbf{ux+x}\widehat{\mathbf{u}}\right)  ,\mathbf{\medskip}\\
\mathbf{x\lrcorner u=}\frac{1}{2}\left(  \mathbf{xu-}\widehat{\mathbf{u}%
}\mathbf{x}\right)  \mathbf{,\qquad u\llcorner x=}\frac{1}{2}\left(
\mathbf{ux-x}\widehat{\mathbf{u}}\right)  ,\mathbf{\medskip}\\
\mathbf{x\lrcorner}\left(  \mathbf{uv}\right)  \mathbf{=}\left(
\mathbf{x\lrcorner u}\right)  \mathbf{v+}\widehat{\mathbf{u}}\left(
\mathbf{x\lrcorner v}\right)  \mathbf{,\qquad}\left(  \mathbf{uv}\right)
\mathbf{\llcorner x=u}\left(  \mathbf{v\llcorner x}\right)  \mathbf{-}\left(
\mathbf{u\llcorner x}\right)  \widehat{\mathbf{v}},\mathbf{\medskip}\\
\mathbf{x\wedge}\left(  \mathbf{uv}\right)  \mathbf{=}\left(
\mathbf{x\lrcorner u}\right)  \mathbf{v+}\widehat{\mathbf{u}}\left(
\mathbf{x\wedge v}\right)  =\left(  \mathbf{x\wedge u}\right)  \mathbf{v-}%
\widehat{\mathbf{u}}\left(  \mathbf{x\lrcorner v}\right)  ,\mathbf{\medskip}\\
\left(  \mathbf{uv}\right)  \wedge\mathbf{x=u}\left(  \mathbf{v\wedge
x}\right)  -\left(  \mathbf{u\llcorner x}\right)  \widehat{\mathbf{v}%
}=\mathbf{u}\left(  \mathbf{v\llcorner x}\right)  \mathbf{+}\left(
\mathbf{u\wedge x}\right)  \widehat{\mathbf{v}},\mathbf{\medskip}\\
\star\mathbf{u}=\widetilde{\mathbf{u}}\mathbf{\sigma,\qquad}\star
^{-1}\mathbf{u}=\widetilde{\mathbf{\sigma}}\widetilde{\mathbf{u}%
},\mathbf{\medskip}\\
\star\left(  \mathbf{uv}\right)  =\widetilde{\mathbf{v}}\left(  \star
\mathbf{u}\right)  \mathbf{,\qquad}\star^{-1}\left(  \mathbf{uv}\right)
=\left(  \star^{-1}\mathbf{v}\right)  \mathbf{u.}\\
\end{array}
\]

Moreover, for an element of the form $\mathbf{u}=u_{\ast}\wedge u^{\ast}$,
with $u_{\ast}\in%
{\textstyle\bigwedge}
V$ and $u^{\ast}\in%
{\textstyle\bigwedge}
V^{\ast}$ and $\mathbf{x\in}%
{\textstyle\bigwedge}
H_{V}$ it is:%
\begin{align*}
\mathbf{xu}  &  \mathbf{=}\left(  x_{\ast}\oplus x^{\ast}\right)
\lrcorner\left(  u_{\ast}\wedge u^{\ast}\right)  +\left(  x_{\ast}\oplus
x^{\ast}\right)  \wedge\left(  u_{\ast}\wedge u^{\ast}\right)  \medskip\\
&  =\left(  x^{\ast}\lrcorner u_{\ast}\right)  \wedge u^{\ast}+\widehat
{u}_{\ast}\wedge\left(  x_{\ast}\lrcorner u^{\ast}\right)  +\left(  x_{\ast
}\oplus x^{\ast}\right)  \wedge\left(  u_{\ast}\wedge u^{\ast}\right)  .
\end{align*}

On the other hand,%
\[
\left(  x_{\ast}\oplus x^{\ast}\right)  \wedge\left(  u_{\ast}\wedge u^{\ast
}\right)  =\widehat{u}_{\ast}\wedge\left(  x_{\ast}\wedge u^{\ast}\right)
+\left(  x^{\ast}\wedge u_{\ast}\right)  \wedge u^{\ast}%
\]
and then,
\begin{align*}
\mathbf{xu}  &  \mathbf{=}\left(  x^{\ast}\lrcorner u_{\ast}\right)  \wedge
u^{\ast}+\widehat{u}_{\ast}\wedge\left(  x_{\ast}\lrcorner u^{\ast}\right)
+\widehat{u}_{\ast}\wedge\left(  x_{\ast}\wedge u^{\ast}\right)  +\left(
x^{\ast}\wedge u_{\ast}\right)  \wedge u^{\ast}\medskip\\
&  =\left(  \left(  x^{\ast}\lrcorner u_{\ast}\right)  +\left(  x^{\ast}\wedge
u_{\ast}\right)  \right)  \wedge u^{\ast}+\widehat{u}_{\ast}\wedge\left(
\left(  x_{\ast}\lrcorner u^{\ast}\right)  +\left(  x_{\ast}\wedge u^{\ast
}\right)  \right)  \medskip\\
&  =(x^{\ast}u_{\ast})\wedge u^{\ast}+\hat{u}_{\ast}\wedge(x_{\ast}u^{\ast}).
\end{align*}
Also note that the square of the volume $2n$-vector $\mathbf{\sigma}$ satisfy
\[
\mathbf{\sigma}^{2}=1.
\]

\begin{proposition}
\label{prop2} There is the following natural isomorphism
\[
\mathcal{C\ell}(H_{V})\simeq\mathop{\mathrm{End}}(%
{\displaystyle\bigwedge}
V).
\]
In addition, being $b$ a non-degenerate symmetric bilinear form on $V$, it
holds also
\[
\mathcal{C\ell}(H_{V})\simeq\mathcal{C\ell}(H_{bV})\simeq\mathcal{C\ell
}(V,b)\text{ }\widehat{\otimes}\text{ }\mathcal{C\ell}(V,-b).
\]

\end{proposition}

\noindent\textbf{Proof. }The first isomorphism is given by the extension to
$\mathcal{C\ell}(H_{V})$ of the Clifford map\/ $\varphi:H_{V}\rightarrow
\mathop{\mathrm{End}}(%
{\displaystyle\bigwedge}
V$ $)$ by $\mathbf{x}\mapsto\varphi_{x}$, with
\[
\varphi_{x}(u_{\ast})={\frac{1}{\sqrt{2}}}\left(  x^{\ast}\mathbin\lrcorner
u_{\ast}+x_{\ast}\wedge u_{\ast}\right)  ,
\]
for all $u_{\ast}\in%
{\displaystyle\bigwedge}
V$. The second isomorphism, in turn, is induced from the Clifford map
\[
x_{\ast}\oplus x^{\ast}\mapsto x_{+}\widehat{\otimes}1+1\widehat{\otimes}%
x_{-},
\]
with
\[
x_{\pm}={{\frac{1}{\sqrt{2}}}}(b^{\ast}x^{\ast}\pm x_{\ast}).
\]

\begin{corollary}
The even and odd subspaces of the hyperbolic Clifford algebra are
\[
\mathcal{C\ell}^{(0)}(H_{V})\simeq\mathop{\mathrm{End}}(%
{\displaystyle\bigwedge\nolimits^{(0)}}
V)\oplus\mathop{\mathrm{End}}(%
{\displaystyle\bigwedge\nolimits^{(1)}}
V)
\]
and
\[
\mathcal{C\ell}^{(1)}(H_{V})\simeq L(%
{\displaystyle\bigwedge\nolimits^{(0)}}
V,%
{\displaystyle\bigwedge\nolimits^{(1)}}
V)\oplus L(%
{\displaystyle\bigwedge\nolimits^{(1)}}
{\displaystyle\bigwedge}
,%
{\displaystyle\bigwedge\nolimits^{(0)}}
V)
\]
where $L(V,W)$ denotes the space of the linear mappings from $V$ to $W$ and $%
{\textstyle\bigwedge\nolimits^{\left(  0\right)  }}
V$ and $%
{\textstyle\bigwedge\nolimits^{\left(  1\right)  }}
V$ denote respectively the spaces of the even and of odd elements of $%
{\textstyle\bigwedge}
V.$
\end{corollary}

\begin{remark}
Recalling the definition of the second order hyperbolic structure, $H_{V}^{2}%
$, it follows from propositon \ref{prop2} that
\[
\mathcal{C\ell}(H_{V}^{2})\simeq\mathop{\mathrm{End}}(\mathcal{C\ell}(H_{V})
\]
The algebra $\mathcal{C\ell}(H_{V}^{2})$ may be called \textit{grandmother}
algebra\/of the vector space $V$.\hfill\ 
\end{remark}

\section{Duality Products of Multivectors and Multiforms, and Extensors}

In this section we study the \emph{duality product} of multivectors by
multiforms used in the definition of the hyperbolic algebra $\mathcal{C\ell
}(V\oplus V^{\ast},\left\langle ,\right\rangle )$ of \emph{multivecfors}. We
detail some important properties of the left and right contracted products
among the elements of $%
{\textstyle\bigwedge}
V$ and $%
{\textstyle\bigwedge}
V^{\ast},$by\ introducing a very useful notation for these products. Next, we
give a theory of the $k$\emph{\ multivector and }$l$\emph{\ multiform
variables multivector }(\emph{or multiform})\emph{\ extensors} over $V$
(defining the spaces $\left.  \overset{\left.  {}\right.  }{ext}\right.
_{k}^{l}(V)$ and $\left.  \overset{\left.  \ast\right.  }{ext}\right.
_{k}^{l}(V)$) introducing the concept of exterior product of extensors, and of
several operators acting on these objects as, e.g., the \emph{adjoint
operator}, the \emph{exterior power extension operator} and the
\emph{contracted extension operator}. We analyze the properties of these
operators with considerable detail.

\subsection{Duality Scalar Product of Multivectors and Multiforms}

In the Appendix A we briefly recall the exterior algebras of multivectors
(elements of $%
{\textstyle\bigwedge}
V$) and multiforms (elements of $%
{\textstyle\bigwedge}
V^{\ast}$) associated with a real vector space $V$ of finite dimension which
is need for the following.

\begin{definition}
The duality scalar product of a multiform $\Phi$ with a multivector $X$ is the
scalar $\left\langle \Phi,X\right\rangle $ defined by the following axioms:

For all $\alpha,\beta\in\mathbb{R}:$%
\begin{equation}
\left\langle \alpha,\beta\right\rangle =\left\langle \beta,\alpha\right\rangle
=\alpha\beta. \label{DSP1}%
\end{equation}

For all $\Phi_{p}\in%
{\textstyle\bigwedge\nolimits^{p}}
V^{\ast}$ and $X^{p}\in%
{\textstyle\bigwedge\nolimits^{p}}
V$ (with $1\leq p\leq n$)$:$
\begin{equation}
\left\langle \Phi_{p},X^{p}\right\rangle =\left\langle X^{p},\Phi
_{p}\right\rangle =\frac{1}{p!}\Phi_{p}(e_{j_{1}},\ldots,e_{j_{p}}%
)X^{p}(\varepsilon^{j_{1}},\ldots,\varepsilon^{j_{p}}), \label{DSP2}%
\end{equation}
where $\left\{  e_{j},\varepsilon^{j}\right\}  $ is any pair of dual bases
over $V.$

For all $\Phi\in%
{\textstyle\bigwedge}
V^{\ast}$ and $X\in%
{\textstyle\bigwedge}
V:$ if $\Phi=\Phi_{0}+\Phi_{1}+\cdots+\Phi_{n}$ and $X=X^{0}+X^{1}%
+\cdots+X^{n},$ then%
\begin{equation}
\left\langle \Phi,X\right\rangle =\left\langle X,\Phi\right\rangle
=\overset{n}{\underset{p=0}{%
{\textstyle\sum}
}}\left\langle \Phi_{p},X^{p}\right\rangle . \label{DSP3}%
\end{equation}

\end{definition}

We emphasize that the scalar $\Phi_{p}(e_{j_{1}},\ldots,e_{j_{p}}%
)X^{p}(\varepsilon^{j_{1}},\ldots,\varepsilon^{j_{p}})$ has frame independent
character, i.e., it does not depend on the pair of dual bases $\left\{
e_{j},\varepsilon^{j}\right\}  $ used in its evaluation, since $\Phi_{p}$ and
$X^{p}$ are $p$-linear mappings.

Note that for all $\omega\in V^{\ast}$ and $v\in V$ it holds
\begin{equation}
\left\langle v,\omega\right\rangle =\left\langle \omega,v\right\rangle
=\omega(v). \label{DSP4}%
\end{equation}

We present now two noticeable properties for the duality scalar product
between $p$-forms and $p$-vectors:

\textbf{(i) }For all $\Phi_{p}\in%
{\textstyle\bigwedge^{p}}
V^{\ast},$ and $v_{1},\ldots,v_{p}\in V:$%
\begin{equation}
\left\langle \Phi_{p},v_{1}\wedge\cdots\wedge v_{p}\right\rangle =\left\langle
v_{1}\wedge\cdots\wedge v_{p},\Phi_{p}\right\rangle =\Phi_{p}(v_{1}%
,\ldots,v_{p}). \label{DSP5}%
\end{equation}

\textbf{(ii)} For all $\omega^{1},\ldots,\omega^{p}\in V^{\ast}$ and
$v_{1},\ldots,v_{p}\in V:$%
\begin{equation}
\left\langle \omega^{1}\wedge\cdots\wedge\omega^{p},v_{1}\wedge\cdots\wedge
v_{p}\right\rangle =\det\left(
\begin{array}
[c]{lll}%
\omega^{1}(v_{1}) & \cdots & \omega^{1}(v_{p})\\
\vdots & \vdots & \vdots\\
\omega^{p}(v_{1}) & \cdots & \omega^{p}(v_{p})
\end{array}
\right)  . \label{DSP6}%
\end{equation}

The basic properties for the duality scalar product are the non-degeneracy and
the distributive laws on the left and on the right with respect to addition of
either multiforms or multivectors, i.e.,%

\begin{align}
\left\langle \Phi,X\right\rangle  &  =0,\text{ for all }\Phi,\text{ then
}X=0,\nonumber\\
\left\langle \Phi,X\right\rangle  &  =0,\text{ for all }X,\text{ then }\Phi=0,
\label{DSP7}%
\end{align}

\begin{align}
\left\langle \Phi+\Psi,X\right\rangle  &  =\left\langle \Phi,X\right\rangle
+\left\langle \Psi,X\right\rangle ,\nonumber\\
\left\langle \Phi,X+Y\right\rangle  &  =\left\langle \Phi,X\right\rangle
+\left\langle \Phi,Y\right\rangle . \label{DSP8}%
\end{align}

\subsection{Duality Contracted Products of Multivectors and Multiforms}

\subsubsection{Left Contracted Product}

\begin{definition}
The left contracted product of a multiform $\Phi$ with a multivector $X$ (or,
a multivector $X$ with\ a multiform $\Phi$) is the multivector $\left\langle
\Phi,X\right\vert $ (respectively, the multiform $\left\langle X,\Phi
\right\vert $) defined by the following axioms:

For all $\Phi_{p}\in%
{\textstyle\bigwedge\nolimits^{p}}
V^{\ast}$ and $X^{p}\in%
{\textstyle\bigwedge\nolimits^{p}}
V$ with $0\leq p\leq n:$%
\begin{equation}
\left\langle \Phi_{p},X^{p}\right\vert =\left\langle X^{p},\Phi_{p}\right\vert
=\left\langle \widetilde{\Phi}_{p},X^{p}\right\rangle =\left\langle \Phi
_{p},\widetilde{X}^{p}\right\rangle . \label{DCP1}%
\end{equation}

For all $\Phi_{p}\in%
{\textstyle\bigwedge\nolimits^{p}}
V^{\ast}$ and $X^{q}\in%
{\textstyle\bigwedge\nolimits^{q}}
V$ (or $X^{p}\in%
{\textstyle\bigwedge\nolimits^{p}}
V$ and $\Phi_{q}\in%
{\textstyle\bigwedge\nolimits^{q}}
V^{\ast}$) with $0\leq p<q\leq n:$%
\begin{align}
\left\langle \Phi_{p},X^{q}\right\vert  &  =\frac{1}{(q-p)!}\left\langle
\widetilde{\Phi}_{p}\wedge\varepsilon^{j_{1}}\wedge\cdots\wedge\varepsilon
^{j_{q-p}},X^{q}\right\rangle e_{j_{1}}\wedge\cdots\wedge e_{j_{q-p}%
},\label{DCP2}\\
\left\langle X^{p},\Phi_{q}\right\vert  &  =\frac{1}{(q-p)!}\left\langle
\widetilde{X}^{p}\wedge e_{j_{1}}\wedge\cdots\wedge e_{j_{q-p}},\Phi
_{q}\right\rangle \varepsilon^{j_{1}}\wedge\cdots\wedge\varepsilon^{j_{q-p}},
\label{DCP3}%
\end{align}
where $\left\{  e_{j},\varepsilon^{j}\right\}  $ is any pair of dual bases for
$V$ and $V^{\ast}$.

For all $\Phi\in%
{\textstyle\bigwedge}
V^{\ast}$ and $X\in%
{\textstyle\bigwedge}
V:$ if $\Phi=\Phi_{0}+\Phi_{1}+\cdots+\Phi_{n}$ and $X=X^{0}+X^{1}%
+\cdots+X^{n},$ then%
\begin{align}
\left\langle \Phi,X\right\vert  &  =\overset{n}{\underset{k=0}{%
{\textstyle\sum}
}}\underset{j=0}{\overset{n-k}{%
{\textstyle\sum}
}}\left\langle \Phi_{j},X^{k+j}\right\vert ,\label{DCP4}\\
\left\langle X,\Phi\right\vert  &  =\overset{n}{\underset{k=0}{%
{\textstyle\sum}
}}\underset{j=0}{\overset{n-k}{%
{\textstyle\sum}
}}\left\langle X^{j},\Phi_{k+j}\right\vert . \label{DCP5}%
\end{align}

\end{definition}

Note that the $(q-p)$-vector $\left\langle \Phi_{p},X^{q}\right\vert $ and the
$(q-p)$-form $\left\langle X^{p},\Phi_{q}\right\vert $ have frame independent
character, i.e., they do not depend on the pair of frames $\left\{
e_{j},\varepsilon^{j}\right\}  $ chosen for calculating them.\medskip

The left contracted product has the following basic properties:\smallskip

\begin{proposition}
For all $\Phi_{p}\in%
{\textstyle\bigwedge\nolimits^{p}}
V^{\ast}$ and $X^{q}\in%
{\textstyle\bigwedge\nolimits^{q}}
V$ with $0\leq p\leq q\leq n.$\ For all $\Psi_{q-p}\in%
{\textstyle\bigwedge\nolimits^{q-p}}
V^{\ast}$, it holds
\begin{equation}
\left\langle \left\langle \Phi_{p},X^{q}\right\vert ,\Psi_{q-p}\right\rangle
=\left\langle X^{q},\widetilde{\Phi}_{p}\wedge\Psi_{q-p}\right\rangle .
\label{DCP6}%
\end{equation}

For all $X^{p}\in%
{\textstyle\bigwedge\nolimits^{p}}
V$ and $\Phi_{q}\in%
{\textstyle\bigwedge\nolimits^{q}}
V^{\ast}$ with $0\leq p\leq q\leq n.$ For all $Y^{q-p}\in%
{\textstyle\bigwedge\nolimits^{q-p}}
V$, it holds
\begin{equation}
\left\langle \left\langle X^{p},\Phi_{q}\right\vert ,Y^{q-p}\right\rangle
=\left\langle \Phi_{q},\widetilde{X}^{p}\wedge Y^{q-p}\right\rangle .
\label{DCP7}%
\end{equation}

For all $X\in%
{\textstyle\bigwedge}
V$ and $\Phi,\Psi\in%
{\textstyle\bigwedge}
V^{\ast}:$
\begin{equation}
\left\langle \left\langle \Phi,X\right\vert ,\Psi\right\rangle =\left\langle
X,\widetilde{\Phi}\wedge\Psi\right\rangle . \label{DCP8}%
\end{equation}

For all $\Phi\in%
{\textstyle\bigwedge}
V^{\ast}$ and $X,Y\in%
{\textstyle\bigwedge}
V:$
\begin{equation}
\left\langle \left\langle X,\Phi\right\vert ,Y\right\rangle =\left\langle
\Phi,\widetilde{X}\wedge Y\right\rangle . \label{DCP9}%
\end{equation}

\end{proposition}

\begin{proposition}
The left contracted product satisfies the distributive laws on the left and on
the right.

For all $\Phi,\Psi\in%
{\textstyle\bigwedge}
V^{\ast}$ and $X,Y\in%
{\textstyle\bigwedge}
V:$
\begin{align}
\left\langle (\Phi+\Psi),X\right\vert  &  =\left\langle \Phi,X\right\vert
+\left\langle \Psi,X\right\vert ,\nonumber\\
\left\langle \Phi,(X+Y)\right\vert  &  =\left\langle \Phi,X\right\vert
+\left\langle \Phi,Y\right\vert . \label{DCP10}%
\end{align}

For all $X,Y\in%
{\textstyle\bigwedge}
V$ and $\Phi,\Psi\in%
{\textstyle\bigwedge}
V^{\ast}:$
\begin{align}
\left\langle (X+Y),\Phi\right\vert  &  =\left\langle X,\Phi\right\vert
+\left\langle Y,\Phi\right\vert ,\nonumber\\
\left\langle X,(\Phi+\Psi)\right\vert  &  =\left\langle X,\Phi\right\vert
+\left\langle X,\Psi\right\vert . \label{DCP11}%
\end{align}

\end{proposition}

\begin{proof}
We present only the proof of the property given by Eq.(\ref{DCP6}), the other
proofs being somewhat analogous.

First note that if $\Psi_{q-p}\in%
{\textstyle\bigwedge\nolimits^{q-p}}
V^{\ast}$ and $\left\{  e_{j},\varepsilon^{j}\right\}  $ is any pair of dual
bases for $V$ and $V^{\ast},$\ we can write
\[
\Psi_{q-p}=\frac{1}{\left(  q-p\right)  !}\left\langle \Psi_{q-p},e_{j_{1}%
}\wedge...\wedge e_{j_{q-p}}\right\rangle \varepsilon^{j_{1}}\wedge
...\wedge\varepsilon^{j_{q-p}}.
\]
Then, using the axiom of the Eq.(\ref{DCP2}) and the above equation we have%
\[%
\begin{array}
[c]{l}%
\left\langle \left\langle \Phi_{p},X^{q}\right\vert ,\Psi_{q-p}\right\rangle
\\
=\frac{1}{\left(  q-p\right)  !}\left\langle \left\langle \widetilde{\Phi}%
_{p}\wedge\varepsilon^{j_{1}}\wedge\cdots\wedge\varepsilon^{j_{q-p}}%
,X^{q}\right\rangle e_{j_{1}}\wedge\cdots\wedge e_{j_{q-p}},\Psi
_{q-p}\right\rangle \\
=\frac{1}{\left(  q-p\right)  !}\left\langle \widetilde{\Phi}_{p}%
\wedge\varepsilon^{j_{1}}\wedge\cdots\wedge\varepsilon^{j_{q-p}}%
,X^{q}\right\rangle \left\langle e_{j_{1}}\wedge\cdots\wedge e_{j_{q-p}}%
,\Psi_{q-p}\right\rangle \\
=\left\langle X^{q},\frac{1}{\left(  q-p\right)  !}\left\langle \Psi
_{q-p},e_{j_{1}}\wedge\cdots\wedge e_{j_{q-p}}\right\rangle \widetilde{\Phi
}_{p}\wedge\varepsilon^{j_{1}}\wedge\cdots\wedge\varepsilon^{j_{q-p}%
}\right\rangle \\
=\left\langle X^{q},\widetilde{\Phi}_{p}\wedge\frac{1}{\left(  q-p\right)
!}\left\langle \Psi_{q-p},e_{j_{1}}\wedge\cdots\wedge e_{j_{q-p}}\right\rangle
\varepsilon^{j_{1}}\wedge\cdots\wedge\varepsilon^{j_{q-p}}\right\rangle \\
=\left\langle X^{q},\widetilde{\Phi}_{p}\wedge\Psi_{q-p}\right\rangle ,
\end{array}
\]
and the result is proved.
\end{proof}

\subsubsection{Right Contracted Product}

\begin{definition}
The right contracted product of a multiform $\Phi$ with a multivector $X$ (or,
a multivector $X$ with a multiform $\Phi$) is the multiform $\left\vert
\Phi,X\right\rangle $ (respectively, the multivector $\left\vert
X,\Phi\right\rangle $) given by the following axioms:

For all $\Phi_{p}\in%
{\textstyle\bigwedge\nolimits^{p}}
V^{\ast}$ and $X^{p}\in%
{\textstyle\bigwedge\nolimits^{p}}
V%
\ddot{}%
$ with $n\geq p\geq0:$%
\begin{equation}
\left\vert \Phi_{p},X^{p}\right\rangle =\left\vert X^{p},\Phi_{p}\right\rangle
=\left\langle \widetilde{\Phi}_{p},X^{p}\right\rangle =\left\langle \Phi
_{p},\widetilde{X}^{p}\right\rangle . \label{DCP12}%
\end{equation}

For all $\Phi_{p}\in%
{\textstyle\bigwedge\nolimits^{p}}
V^{\ast}$ and $X^{q}\in%
{\textstyle\bigwedge\nolimits^{q}}
V%
\ddot{}%
$ (or $X^{p}\in%
{\textstyle\bigwedge\nolimits^{p}}
V$ and $\Phi_{q}\in%
{\textstyle\bigwedge\nolimits^{q}}
V^{\ast}$) with $n\geq p>q\geq0:$%
\begin{align}
\left\vert \Phi_{p},X^{q}\right\rangle  &  =\frac{1}{(p-q)!}\left\langle
\Phi_{p},e_{j_{1}}\wedge\cdots\wedge e_{j_{p-q}}\wedge\widetilde{X}%
^{q}\right\rangle \varepsilon^{j_{1}}\wedge\cdots\wedge\varepsilon^{j_{p-q}%
},\label{DCP13}\\
\left\vert X^{p},\Phi_{q}\right\rangle  &  =\frac{1}{(p-q)!}\left\langle
X^{p},\varepsilon^{j_{1}}\wedge\cdots\wedge\varepsilon^{j_{p-q}}%
\wedge\widetilde{\Phi}_{q}\right\rangle e_{j_{1}}\wedge\cdots\wedge
e_{j_{p-q}}, \label{DCP14}%
\end{align}
where $\left\{  e_{j},\varepsilon^{j}\right\}  $ is any pair of dual bases for
$V$ and $V^{\ast}$.

For all $\Phi\in%
{\textstyle\bigwedge}
V^{\ast}$ and $X\in%
{\textstyle\bigwedge}
V:$ if $\Phi=\Phi_{0}+\Phi_{1}+\cdots+\Phi_{n}$ and $X=X^{0}+X^{1}%
+\cdots+X^{n}$ it is:%
\begin{align}
\left\vert \Phi,X\right\rangle  &  =\overset{n}{\underset{k=0}{%
{\textstyle\sum}
}}\underset{j=0}{\overset{n-k}{%
{\textstyle\sum}
}}\left\vert \Phi_{k+j},X^{j}\right\rangle ,\label{DCP15}\\
\left\vert X,\Phi\right\rangle  &  =\overset{n}{\underset{k=0}{%
{\textstyle\sum}
}}\underset{j=0}{\overset{n-k}{%
{\textstyle\sum}
}}\left\vert X^{k+j},\Phi_{j}\right\rangle . \label{DCP16}%
\end{align}

\end{definition}

The right contracted product satisfies properties similar to the left
contracted product (see \cite{fmcr1}).

\subsection{Duality Adjoint of Extensors}

In the Appendix D we briefly recall the extensor concept and the exterior
product of extensors associated with a real vector space $V$ of finite
dimension which is need for the following. Let $%
{\textstyle\bigwedge\nolimits^{\diamond}}
V$ be any sum of homogeneous subspaces of $%
{\textstyle\bigwedge}
V.$ There exist $\mu$ integer numbers $p_{1},\ldots,p_{\mu}$ with $0\leq
p_{1}<\cdots<p_{\mu}\leq n$ such that $%
{\textstyle\bigwedge\nolimits^{\diamond}}
V=%
{\textstyle\bigwedge\nolimits^{p_{1}}}
V+\cdots+%
{\textstyle\bigwedge\nolimits^{p_{\mu}}}
V.$ Analogously, if $%
{\textstyle\bigwedge\nolimits^{\diamond}}
V^{\ast}$ is any sum of homogeneous subspaces of $%
{\textstyle\bigwedge}
V^{\ast},$ then there exist $\nu$ integer numbers $q_{1},\ldots,q_{\nu}$ with
$0\leq q_{1}<\cdots<q_{\nu}\leq n$ such that $%
{\textstyle\bigwedge\nolimits^{\diamond}}
V^{\ast}=%
{\textstyle\bigwedge\nolimits^{q_{1}}}
V^{\ast}+\cdots+%
{\textstyle\bigwedge\nolimits^{q_{\nu}}}
V^{\ast}.$

\begin{definition}
The linear mappings
\[%
{\textstyle\bigwedge}
V\ni X\mapsto\left\langle X\right\rangle ^{%
{\textstyle\bigwedge\nolimits^{\diamond}}
V}\in%
{\textstyle\bigwedge}
V\text{ and }%
{\textstyle\bigwedge}
V^{\ast}\ni\Phi\mapsto\left\langle \Phi\right\rangle _{%
{\textstyle\bigwedge\nolimits^{\diamond}}
V^{\ast}}\in%
{\textstyle\bigwedge}
V^{\ast}%
\]
such that if $%
{\textstyle\bigwedge\nolimits^{\diamond}}
V=%
{\textstyle\bigwedge\nolimits^{p_{1}}}
V+\cdots+%
{\textstyle\bigwedge\nolimits^{p_{\mu}}}
V$ and $%
{\textstyle\bigwedge\nolimits^{\diamond}}
V^{\ast}=%
{\textstyle\bigwedge\nolimits^{q_{1}}}
V^{\ast}+\cdots+%
{\textstyle\bigwedge\nolimits^{q_{\nu}}}
V^{\ast},$ then
\begin{equation}
\left\langle X\right\rangle ^{%
{\textstyle\bigwedge\nolimits^{\diamond}}
V}=\left\langle X\right\rangle ^{p_{1}}+\cdots+\left\langle X\right\rangle
^{p_{\mu}}\text{ and }\left\langle \Phi\right\rangle _{%
{\textstyle\bigwedge\nolimits^{\diamond}}
V^{\ast}}=\left\langle \Phi\right\rangle _{q_{1}}+\cdots+\left\langle
\Phi\right\rangle _{q_{\nu}} \label{DAE1}%
\end{equation}
are called the \emph{\ }$%
{\textstyle\bigwedge\nolimits^{\diamond}}
V$\emph{-part operator} \emph{for multivectors} and $%
{\textstyle\bigwedge\nolimits^{\diamond}}
V^{\ast}$\emph{-part operator for multiforms,} respectively.
\end{definition}

It should be evident that for all $X\in%
{\textstyle\bigwedge}
V$ and $\Phi\in%
{\textstyle\bigwedge}
V^{\ast}:$%
\begin{align}
\left\langle X\right\rangle ^{%
{\textstyle\bigwedge\nolimits^{k}}
V}  &  =\left\langle X\right\rangle ^{k},\label{DAE2}\\
\left\langle \Phi\right\rangle _{%
{\textstyle\bigwedge\nolimits^{k}}
V^{\ast}}  &  =\left\langle \Phi\right\rangle _{k}. \label{DAE3}%
\end{align}
Thus, $%
{\textstyle\bigwedge\nolimits^{\diamond}}
V$-part operator and $%
{\textstyle\bigwedge\nolimits^{\diamond}}
V^{\ast}$-part operator are the generalizations of $\left\langle \left.
{}\right.  \right\rangle ^{k}$ and $\left\langle \left.  {}\right.
\right\rangle _{k}.$

\begin{definition}
Let $\tau$ be a multivector extensor of either one multivector variable or one
multiform variable. The \emph{duality adjoint of }$\tau$ is given by:
\end{definition}

If $\tau\in ext(%
{\textstyle\bigwedge\nolimits_{1}^{\diamond}}
V;%
{\textstyle\bigwedge\nolimits_{2}^{\diamond}}
V)$ (or, $\tau\in ext(%
{\textstyle\bigwedge\nolimits_{3}^{\diamond}}
V^{\ast};%
{\textstyle\bigwedge\nolimits_{4}^{\diamond}}
V)$), then $\tau^{\bigtriangleup}\in ext(%
{\textstyle\bigwedge\nolimits_{2}^{\diamond}}
V^{\ast};$ $%
{\textstyle\bigwedge\nolimits_{1}^{\diamond}}
V^{\ast})$ (respectively, $\tau^{\bigtriangleup}\in ext(%
{\textstyle\bigwedge\nolimits_{4}^{\diamond}}
V^{\ast};%
{\textstyle\bigwedge\nolimits_{3}^{\diamond}}
V)$) defined by
\begin{align}
\tau^{\bigtriangleup}(\Phi)  &  =\left\langle \Phi,\tau(\left\langle
1\right\rangle ^{%
{\textstyle\bigwedge\nolimits_{1}^{\diamond}}
V})\right\rangle \nonumber\\
&  +\overset{n}{\underset{k=1}{%
{\textstyle\sum}
}}\frac{1}{k!}\left\langle \Phi,\tau(\left\langle e_{j_{1}}\wedge\cdots\wedge
e_{j_{k}}\right\rangle ^{%
{\textstyle\bigwedge\nolimits_{1}^{\diamond}}
V})\right\rangle \varepsilon^{j_{1}}\wedge\cdots\wedge\varepsilon^{j_{k}}
\label{DAE4}%
\end{align}
for each $\Phi\in%
{\textstyle\bigwedge\nolimits_{2}^{\diamond}}
V^{\ast}$ (respectively, for each $\Phi\in%
{\textstyle\bigwedge\nolimits_{4}^{\diamond}}
V^{\ast}$)%

\begin{align}
\tau^{\bigtriangleup}(\Phi)  &  =\left\langle \Phi,\tau(\left\langle
1\right\rangle _{%
{\textstyle\bigwedge\nolimits_{3}^{\diamond}}
V^{\ast}})\right\rangle \nonumber\\
&  +\overset{n}{\underset{k=1}{%
{\textstyle\sum}
}}\frac{1}{k!}\left\langle \Phi,\tau(\left\langle \varepsilon^{j_{1}}%
\wedge\cdots\wedge\varepsilon^{j_{k}}\right\rangle _{%
{\textstyle\bigwedge\nolimits_{3}^{\diamond}}
V^{\ast}})\right\rangle e_{j_{1}}\wedge\cdots\wedge e_{j_{k}} \label{DAE5}%
\end{align}

The basic properties of the adjoint of multivector extensors are given by:

\begin{proposition}
Let $\tau\in ext(%
{\textstyle\bigwedge\nolimits_{1}^{\diamond}}
V;%
{\textstyle\bigwedge\nolimits_{2}^{\diamond}}
V).$ For all $X\in%
{\textstyle\bigwedge\nolimits_{1}^{\diamond}}
V$ and $\Phi\in%
{\textstyle\bigwedge\nolimits_{2}^{\diamond}}
V^{\ast},$ it holds%
\begin{equation}
\left\langle \tau(X),\Phi\right\rangle =\left\langle X,\tau^{\bigtriangleup
}(\Phi)\right\rangle . \label{DAE6}%
\end{equation}

Let $\tau\in ext(%
{\textstyle\bigwedge\nolimits_{3}^{\diamond}}
V^{\ast};%
{\textstyle\bigwedge\nolimits_{4}^{\diamond}}
V).$ For all $\Phi\in%
{\textstyle\bigwedge\nolimits_{3}^{\diamond}}
V^{\ast}$ and $\Psi\in%
{\textstyle\bigwedge\nolimits_{4}^{\diamond}}
V^{\ast},$ it holds
\begin{equation}
\left\langle \tau(\Phi),\Psi\right\rangle =\left\langle \Phi,\tau
^{\bigtriangleup}(\Psi)\right\rangle . \label{DAE7}%
\end{equation}

\end{proposition}

\begin{proof}
We present only the proof of the property given by Eq.(\ref{DAE6}), the other
one is similar. First, observe that if $X\in%
{\textstyle\bigwedge\nolimits_{1}^{\diamond}}
V,$ then there exists $\mu$ integer numbers $p_{1},\ldots,p_{\mu}$ with $0\leq
p_{1}<\cdots<p_{\mu}\leq n$ such that $X=X^{p_{1}}+...+X^{p_{\mu}}$ with
$X^{p_{i}}\in%
{\textstyle\bigwedge\nolimits^{p_{i}}}
V$, where the $%
{\textstyle\bigwedge\nolimits^{p_{i}}}
V$ are homogeneous subspace of $%
{\textstyle\bigwedge}
V,$ thus if $\tau\in ext(%
{\textstyle\bigwedge\nolimits_{1}^{\diamond}}
V;%
{\textstyle\bigwedge\nolimits_{2}^{\diamond}}
V)$, we have that%
\[
\tau:%
{\textstyle\bigwedge\nolimits_{1}^{\diamond}}
V\rightarrow%
{\textstyle\bigwedge\nolimits_{2}^{\diamond}}
V\qquad\text{or\qquad}\tau:%
{\textstyle\bigwedge\nolimits^{p_{1}}}
V+...+%
{\textstyle\bigwedge\nolimits^{p_{\mu}}}
V\rightarrow%
{\textstyle\bigwedge\nolimits^{q_{1}}}
V+...+%
{\textstyle\bigwedge\nolimits^{q_{\mu}}}
V,
\]
where $%
{\textstyle\bigwedge\nolimits_{1}^{\diamond}}
V=%
{\textstyle\bigwedge\nolimits^{p_{1}}}
V+...+%
{\textstyle\bigwedge\nolimits^{p_{\mu}}}
V$ and $%
{\textstyle\bigwedge\nolimits_{2}^{\diamond}}
V=%
{\textstyle\bigwedge\nolimits^{q_{1}}}
V+...+%
{\textstyle\bigwedge\nolimits^{q_{\mu}}}
V$,\ such that
\[
\tau\left(  X\right)  =\tau\left(  X^{p_{1}}+...+X^{p_{\mu}}\right)
=\tau\left(  X^{p_{1}}\right)  +...+\tau\left(  X^{p_{\mu}}\right)  \in%
{\textstyle\bigwedge\nolimits^{q_{1}}}
V+...+%
{\textstyle\bigwedge\nolimits^{q_{\mu}}}
V,
\]
thus, we can put
\[
\tau\left(  X^{p_{i}}\right)  \in%
{\textstyle\bigwedge\nolimits^{q_{i}}}
V\qquad\text{or\qquad}\tau\left\vert _{%
{\textstyle\bigwedge\nolimits^{p_{i}}}
V}\right.  \equiv\tau_{p_{i}}\in ext(%
{\textstyle\bigwedge\nolimits^{p_{i}}}
V;%
{\textstyle\bigwedge\nolimits^{q_{i}}}
V).
\]
Now, with the above observation, we have
\begin{equation}%
\begin{array}
[c]{ll}%
\left\langle \tau(X),\Phi\right\rangle  & =\left\langle \tau\left(  X^{p_{1}%
}+...+X^{p_{\mu}}\right)  ,\Phi_{q_{1}}+...+\Phi_{q_{\mu}}\right\rangle
\medskip\\
& =\left\langle \tau_{p_{1}}\left(  X^{p_{1}}\right)  +...+\tau_{p_{\mu}%
}\left(  X^{p_{\mu}}\right)  ,\Phi_{q_{1}}+...+\Phi_{q_{\mu}}\right\rangle
\medskip\\
& =%
{\textstyle\sum\limits_{i,j}}
\left\langle \tau_{p_{i}}(X^{p_{i}}),\Phi^{q_{j}}\right\rangle ,
\end{array}
\label{EAN1}%
\end{equation}
but, by axiom \ given in Eq.(\ref{DSP3}), we have that%
\[
\left\langle \tau_{p_{i}}(X^{p_{i}}),\Phi^{q_{j}}\right\rangle \left\{
\begin{array}
[c]{l}%
=0\qquad\text{if \ }p_{i}\neq q_{j}\\
\neq0\qquad\text{if \ }p_{i}=q_{j}=s_{l}%
\end{array}
\right.  ,
\]
and from Eq. (\ref{EAN1}), we can write%
\begin{equation}
\left\langle \tau(X),\Phi\right\rangle =\left\langle \tau_{s_{1}}(X^{s_{1}%
}),\Phi^{s_{1}}\right\rangle +...+\left\langle \tau_{s_{\mu}}(X^{s_{\mu}%
}),\Phi^{s_{\mu}}\right\rangle . \label{EN2}%
\end{equation}
Now, if we see $\left\langle \tau_{s_{l}}(X^{s_{l}}),\Phi^{s_{l}}\right\rangle
$ as a scalar product of $s_{l}$-vectors, then we have%
\begin{equation}
\left\langle \tau_{s_{l}}(X^{s_{l}}),\Phi^{s_{l}}\right\rangle =\left\langle
X^{s_{l}},\tau_{s_{l}}^{\bigtriangleup}\Phi^{s_{l}}\right\rangle , \label{EN3}%
\end{equation}
and from Eqs. (\ref{EN2}), (\ref{EN3}), and taking into account Eq.
(\ref{DAE6}) we have%
\[%
\begin{array}
[c]{ll}%
\left\langle \tau(X),\Phi\right\rangle  & =\left\langle \tau_{s_{1}}(X^{s_{1}%
}),\Phi^{s_{1}}\right\rangle +...+\left\langle \tau_{s_{\mu}}(X^{s_{\mu}%
}),\Phi^{s_{\mu}}\right\rangle \\
& =\left\langle X^{s_{1}},\tau_{s_{1}}^{\bigtriangleup}\Phi^{s_{1}%
}\right\rangle +...+\left\langle X^{s_{\mu}},\tau_{s_{\mu}}^{\bigtriangleup
}\Phi^{s_{\mu}}\right\rangle \\
& =\left\langle X^{s_{1}}+...+X^{s_{\mu}},\tau_{s_{1}}^{\bigtriangleup}%
\Phi^{s_{1}}+...+\tau_{s_{\mu}}^{\bigtriangleup}\Phi^{s_{\mu}}\right\rangle \\
& =\left\langle X,\tau^{\bigtriangleup}\Phi\right\rangle ,
\end{array}
\]
and the result is proved.
\end{proof}

\begin{definition}
Let $\sigma$ be a multiform extensor of either one multivector variable or one
multiform variable. \ If $\sigma\in ext(%
{\textstyle\bigwedge\nolimits_{1}^{\diamond}}
V;%
{\textstyle\bigwedge\nolimits_{2}^{\diamond}}
V^{\ast})$ (or, $\sigma\in ext(%
{\textstyle\bigwedge\nolimits_{3}^{\diamond}}
V^{\ast};%
{\textstyle\bigwedge\nolimits_{4}^{\diamond}}
V^{\ast})$), then The \emph{duality adjoint operator } $\left(  \left.
{}\right.  \right)  ^{\bigtriangleup}$ of $\sigma^{\bigtriangleup}\in ext(%
{\textstyle\bigwedge\nolimits_{2}^{\diamond}}
V;$ $%
{\textstyle\bigwedge\nolimits_{1}^{\diamond}}
V^{\ast})$ (respectively, $\sigma^{\bigtriangleup}\in ext(%
{\textstyle\bigwedge\nolimits_{4}^{\diamond}}
V;%
{\textstyle\bigwedge\nolimits_{3}^{\diamond}}
V)$) is given by%
\begin{align}
\sigma^{\bigtriangleup}(X)  &  =\left\langle X,\sigma(\left\langle
1\right\rangle ^{%
{\textstyle\bigwedge\nolimits_{1}^{\diamond}}
V})\right\rangle \nonumber\\
&  +\overset{n}{\underset{k=1}{%
{\textstyle\sum}
}}\frac{1}{k!}\left\langle X,\sigma(\left\langle e_{j_{1}}\wedge\cdots\wedge
e_{j_{k}}\right\rangle ^{%
{\textstyle\bigwedge\nolimits_{1}^{\diamond}}
V})\right\rangle \varepsilon^{j_{1}}\wedge\cdots\wedge\varepsilon^{j_{k}}
\label{DAE8}%
\end{align}
for each $X\in%
{\textstyle\bigwedge\nolimits_{2}^{\diamond}}
V$ (respectively,
\begin{align}
\sigma^{\bigtriangleup}(X)  &  =\left\langle X,\sigma(\left\langle
1\right\rangle _{%
{\textstyle\bigwedge\nolimits_{3}^{\diamond}}
V^{\ast}})\right\rangle \nonumber\\
&  +\overset{n}{\underset{k=1}{%
{\textstyle\sum}
}}\frac{1}{k!}\left\langle X,\sigma(\left\langle \varepsilon^{j_{1}}%
\wedge\cdots\wedge\varepsilon^{j_{k}}\right\rangle _{%
{\textstyle\bigwedge\nolimits_{3}^{\diamond}}
V^{\ast}})\right\rangle e_{j_{1}}\wedge\cdots\wedge e_{j_{k}} \label{DAE9}%
\end{align}
for each $X\in%
{\textstyle\bigwedge\nolimits_{4}^{\diamond}}
V$).
\end{definition}

The basic properties for the adjoint of multiform extensors are given
by:\smallskip

Let $\sigma\in ext(%
{\textstyle\bigwedge\nolimits_{1}^{\diamond}}
V;%
{\textstyle\bigwedge\nolimits_{2}^{\diamond}}
V^{\ast})$. For all $X\in%
{\textstyle\bigwedge\nolimits_{1}^{\diamond}}
V$ and $Y\in%
{\textstyle\bigwedge\nolimits_{2}^{\diamond}}
V,$ it holds
\begin{equation}
\left\langle \sigma(X),Y\right\rangle =\left\langle X,\sigma^{\bigtriangleup
}(Y)\right\rangle . \label{DAE10}%
\end{equation}

Let $\sigma\in ext(%
{\textstyle\bigwedge\nolimits_{3}^{\diamond}}
V^{\ast};%
{\textstyle\bigwedge\nolimits_{4}^{\diamond}}
V^{\ast}).$ For all $\Phi\in%
{\textstyle\bigwedge\nolimits_{3}^{\diamond}}
V^{\ast}$ and $X\in%
{\textstyle\bigwedge\nolimits_{4}^{\diamond}}
V,$ it holds
\begin{equation}
\left\langle \sigma(\Phi),X\right\rangle =\left\langle \Phi,\sigma
^{\bigtriangleup}(X)\right\rangle . \label{DAE11}%
\end{equation}

\subsection{Extended Operators on Extensors}

Let $\lambda$ be an invertible linear operator on $V$. As we saw in appendix B
and C, $\lambda$ can be extended by considering two different procedures,
non-equivalent, denoted by $\underline{\lambda}$ and $\underset{\smile}%
{\gamma}$ both mapping multivectors over $V$ into multivectors over $V.$
However, it is possible to extend the action of $\underline{\lambda}$ and
$\underset{\smile}{\gamma}$\ in such a way that they map multivector extensors
over $V$ into multivector extensors over $V$.

\subsubsection{Exterior Power Extension Operators on Extensors}

\begin{definition}
Let $\lambda$ be an invertible linear operator on $V$. The exterior power
extension of a vector operator $\lambda$ on multivector extensors is the
linear mapping%
\[
\left.  \overset{\left.  {}\right.  }{ext}\right.  _{k}^{l}(V)\ni\tau
\mapsto\underline{\lambda}\tau\in\left.  \overset{\left.  {}\right.  }%
{ext}\right.  _{k}^{l}(V)
\]
such that%
\begin{equation}
\underline{\lambda}\tau(X_{1},\ldots,X_{k},\Phi^{1},\ldots,\Phi^{l}%
)=\underline{\lambda}\circ\tau(\underline{\lambda}^{-1}(X_{1}),\ldots
,\underline{\lambda}^{-1}(X_{k}),\underline{\lambda}^{\bigtriangleup}(\Phi
^{1}),\ldots,\underline{\lambda}^{\bigtriangleup}(\Phi^{l})) \label{EPO18}%
\end{equation}
for each $X_{1},\ldots,X_{k}\in%
{\textstyle\bigwedge}
V$ and $\Phi^{1},\ldots,\Phi^{l}\in%
{\textstyle\bigwedge}
V^{\ast}$.
\end{definition}

It that way $\underline{\lambda}$ can be thought as a linear multivector
extensor operator.

For instance, for $\tau\in\left.  \overset{\left.  {}\right.  }{ext}\right.
_{1}^{0}(V)$ the above definition means%
\begin{equation}
\underline{\lambda}\tau=\underline{\lambda}\circ\tau\circ\underline{\lambda
}^{-1}, \label{EPO19}%
\end{equation}
and for $\tau\in\left.  \overset{\left.  {}\right.  }{ext}\right.  _{0}%
^{1}(V)$ it yields%
\begin{equation}
\underline{\lambda}\tau=\underline{\lambda}\circ\tau\circ\underline{\lambda
}^{\bigtriangleup}. \label{EPO20}%
\end{equation}

Let $\lambda$ be an invertible linear operator on $V^{\ast}.$ Analogously to
the above case, it is possible to extend the action of $\underline{\lambda}$
in such a way to get a linear operator on $\left.  \overset{\left.
\ast\right.  }{ext}\right.  _{k}^{l}(V).$ We define%
\[
\left.  \overset{\left.  \ast\right.  }{ext}\right.  _{k}^{l}(V)\ni\tau
\mapsto\underline{\lambda}\tau\in\left.  \overset{\left.  \ast\right.  }%
{ext}\right.  _{k}^{l}(V)
\]
such that%
\begin{equation}
\underline{\lambda}\tau(X_{1},\ldots,X_{k},\Phi^{1},\ldots,\Phi^{l}%
)=\underline{\lambda}\circ\tau(\underline{\lambda}^{\bigtriangleup}%
(X_{1}),\ldots,\underline{\lambda}^{\bigtriangleup}(X_{k}),\underline{\lambda
}^{-1}(\Phi^{1}),\ldots,\underline{\lambda}^{-1}(\Phi^{l})) \label{EPO25}%
\end{equation}
for each $X_{1},\ldots,X_{k}\in%
{\textstyle\bigwedge}
V$ and $\Phi^{1},\ldots,\Phi^{l}\in%
{\textstyle\bigwedge}
V^{\ast}$.

For instance, for $\tau\in\left.  \overset{\left.  {}\right.  }{ext}\right.
_{1}^{0}(V)$ the above definition yields%
\begin{equation}
\underline{\lambda}\tau=\underline{\lambda}\circ\tau\circ\underline{\lambda
}^{\bigtriangleup}, \label{EPO26}%
\end{equation}
and for $\tau\in\left.  \overset{\left.  {}\right.  }{ext}\right.  _{0}%
^{1}(V)$ it implies that%
\begin{equation}
\underline{\lambda}\tau=\underline{\lambda}\circ\tau\circ\underline{\lambda
}^{-1}. \label{EPO27}%
\end{equation}

We give some of the properties of the action of the exterior power extension
of a vector operator $\lambda$ on multivector extensors.

\begin{proposition}
For all $\tau\in\left.  \overset{\left.  {}\right.  }{ext}\right.  _{k}%
^{l}(V)$ and $\sigma\in\left.  \overset{\left.  {}\right.  }{ext}\right.
_{r}^{s}(V):$%
\begin{equation}
\underline{\lambda}(\tau\wedge\sigma)=(\underline{\lambda}\tau)\wedge
(\underline{\lambda}\sigma). \label{EPO21}%
\end{equation}

For all $\tau\in\left.  \overset{\ast}{ext}\right.  _{k}^{l}(V)$ and
$\sigma\in\left.  \overset{\left.  {}\right.  }{ext}\right.  _{r}^{s}(V):$%
\begin{align}
\underline{\lambda}\left\langle \tau,\sigma\right\rangle  &  =\left\langle
\underline{\lambda}^{-\bigtriangleup}\tau,\underline{\lambda}\sigma
\right\rangle ,\label{EPO22}\\
\underline{\lambda}\left\langle \tau,\sigma\right\vert  &  =\left\langle
\underline{\lambda}^{-\bigtriangleup}\tau,\underline{\lambda}\sigma\right\vert
,\label{EPO23}\\
\underline{\lambda}\left\vert \sigma,\tau\right\rangle  &  =\left\vert
\underline{\lambda}\sigma,\underline{\lambda}^{-\bigtriangleup}\tau
\right\rangle . \label{EPO24}%
\end{align}

\end{proposition}

We present the proof of the property given by Eq.(\ref{EPO23}), the other
proof are similar.

\begin{proof}
Without any loss of generality, we give the proof for the particular case
where $\tau\in ext(%
{\textstyle\bigwedge\nolimits_{1}^{\diamond}}
V,%
{\textstyle\bigwedge\nolimits_{2}^{\diamond}}
V^{\ast};%
{\textstyle\bigwedge\nolimits^{\diamond}}
V^{\ast})$ and $\sigma\in ext(%
{\textstyle\bigwedge\nolimits_{3}^{\diamond}}
V,%
{\textstyle\bigwedge\nolimits_{4}^{\diamond}}
V^{\ast};%
{\textstyle\bigwedge\nolimits^{\diamond}}
V)$. Take $X\in%
{\textstyle\bigwedge\nolimits_{1}^{\diamond}}
V,$ $Y\in%
{\textstyle\bigwedge\nolimits_{3}^{\diamond}}
V$ and $\Phi\in%
{\textstyle\bigwedge\nolimits_{2}^{\diamond}}
V^{\ast},$ $\Psi\in%
{\textstyle\bigwedge\nolimits_{4}^{\diamond}}
V^{\ast}.$ A straightforward calculation, using Eq.(\ref{EPO18}),
Eq.(\ref{E7}), Eq.(\ref{EPO13}) and Eq.(\ref{EPO25}), gives
\begin{align*}
\underline{\lambda}\left\langle \tau,\sigma\right\vert (X,Y,\Phi,\Psi)  &
=\underline{\lambda}\circ\left\langle \tau,\sigma\right\vert (\underline
{\lambda}^{-1}(X),\underline{\lambda}^{-1}(Y),\underline{\lambda
}^{\bigtriangleup}(\Phi),\underline{\lambda}^{\bigtriangleup}(\Psi))\\
&  =\underline{\lambda}\left\langle \tau(\underline{\lambda}^{-1}%
(X),\underline{\lambda}^{\bigtriangleup}(\Phi)),\sigma(\underline{\lambda
}^{-1}(Y),\underline{\lambda}^{\bigtriangleup}(\Psi))\right\vert \\
&  =\left\langle \underline{\lambda}^{-\bigtriangleup}\circ\tau(\underline
{\lambda}^{-1}(X),\underline{\lambda}^{\bigtriangleup}(\Phi)),\underline
{\lambda}\circ\sigma(\underline{\lambda}^{-1}(Y),\underline{\lambda
}^{\bigtriangleup}(\Psi))\right\vert \\
&  =\left\langle \underline{\lambda}^{-\bigtriangleup}\tau(X,\Phi
),\underline{\lambda}\sigma(Y,\Psi)\right\vert \\
&  =\left\langle \underline{\lambda}^{-\bigtriangleup}\tau,\underline{\lambda
}\sigma\right\vert (X,Y,\Phi,\Psi),
\end{align*}
whence, the expected result follows.
\end{proof}

We present \ now some of the properties of the action of the exterior power
extension of a form operator $\lambda$ on multiform extensors.

\begin{proposition}
For all $\tau\in\left.  \overset{\ast}{ext}\right.  _{k}^{l}(V)$ and
$\sigma\in\left.  \overset{\ast}{ext}\right.  _{r}^{s}(V):$%
\begin{equation}
\underline{\lambda}(\tau\wedge\sigma)=(\underline{\lambda}\tau)\wedge
\underline{\lambda}(\sigma). \label{EPO28}%
\end{equation}

For all $\tau\in\left.  \overset{\ast}{ext}\right.  _{k}^{l}(V)$ and
$\sigma\in\left.  \overset{\left.  {}\right.  }{ext}\right.  _{r}^{s}(V):$%
\begin{align}
\underline{\lambda}\left\langle \tau,\sigma\right\rangle  &  =\left\langle
\underline{\lambda}\tau,\underline{\lambda}^{-\bigtriangleup}\sigma
\right\rangle ,\label{EPO29}\\
\underline{\lambda}\left\langle \tau,\sigma\right\vert  &  =\left\langle
\underline{\lambda}\tau,\underline{\lambda}^{-\bigtriangleup}\sigma\right\vert
,\label{EPO30}\\
\underline{\lambda}\left\vert \sigma,\tau\right\rangle  &  =\left\vert
\underline{\lambda}^{-\bigtriangleup}\sigma,\underline{\lambda}\tau
\right\rangle . \label{EPO31}%
\end{align}

\end{proposition}

The proofs are similar to the one of the previous proposition.

\subsubsection{Action of Contracted Extension Operators on Extensors}

\begin{definition}
Let $\gamma$ be a linear operator on $V$. Its \textit{contracted extension
operator} is an operator $\underset{\smile}{\gamma}$ which maps map
multivector extensors over $V$ into multivector extensors over $V$. We have%
\[
\left.  \overset{\left.  {}\right.  }{ext}\right.  _{k}^{l}(V)\ni\tau
\mapsto\underset{\smile}{\gamma}\tau\in\left.  \overset{\left.  {}\right.
}{ext}\right.  _{k}^{l}(V)
\]
such that%
\begin{align}
\underset{\smile}{\gamma}\tau(X_{1},\ldots,X_{k},\Phi^{1},\ldots,\Phi^{l})  &
=\underset{\smile}{\gamma}\circ\tau(X_{1},\ldots,X_{k},\Phi^{1},\ldots
,\Phi^{l})\nonumber\\
&  -\tau(\underset{\smile}{\gamma}(X_{1}),\ldots,X_{k},\Phi^{1},\ldots
,\Phi^{l})\nonumber\\
&  \cdots-\tau(X_{1},\ldots,\underset{\smile}{\gamma}(X_{k}),\Phi^{1}%
,\ldots,\Phi^{l})\label{GPO18}\\
&  +\tau(X_{1},\ldots,X_{k},\underset{\smile}{\gamma}^{\bigtriangleup}%
(\Phi^{1}),\ldots,\Phi^{l})\nonumber\\
&  \cdots+\tau(X_{1},\ldots,X_{k},\Phi^{1},\ldots,\underset{\smile}{\gamma
}^{\bigtriangleup}(\Phi^{l})\nonumber
\end{align}
for each $X_{1},\ldots,X_{k}\in%
{\textstyle\bigwedge}
V$ and $\Phi^{1},\ldots,\Phi^{l}\in%
{\textstyle\bigwedge}
V^{\ast}$.
\end{definition}

Thus, we can think of $\underset{\smile}{\gamma}$ as a linear
\emph{multivector extensor operator}. The definitions of $\underset{\smile
}{\gamma}$ for $%
{\textstyle\bigwedge}
V$ and $%
{\textstyle\bigwedge}
V^{\ast}$ are given by Eq.(\ref{GPO1}) and Eq.(\ref{GPO6}) respectively.

For instance, for $\tau\in\left.  \overset{\left.  {}\right.  }{ext}\right.
_{1}^{0}(V)$ the above definition gives%
\begin{equation}
\underset{\smile}{\gamma}\tau=\underset{\smile}{\gamma}\circ\tau-\tau
\circ\underset{\smile}{\gamma}=\left[  \underset{\smile}{\gamma},\tau\right]
, \label{GPO19}%
\end{equation}
and for $\tau\in\left.  \overset{\left.  {}\right.  }{ext}\right.  _{0}%
^{1}(V)$ it yields%
\begin{equation}
\underset{\smile}{\gamma}\tau=\underset{\smile}{\gamma}\circ\tau+\tau
\circ\underset{\smile}{\gamma}^{\bigtriangleup}. \label{GPO20}%
\end{equation}

\begin{definition}
Let $\gamma$ be an invertible linear operator on $V^{\ast}$. Analogously to
the above case, it is possible to generalize the action of $\underset{\smile
}{\gamma}$ in such a way to get a linear operator on $\left.  \overset{\left.
\ast\right.  }{ext}\right.  _{k}^{l}(V)$. We have%
\[
\left.  \overset{\left.  \ast\right.  }{ext}\right.  _{k}^{l}(V)\ni\tau
\mapsto\underset{\smile}{\gamma}\tau\in\left.  \overset{\left.  \ast\right.
}{ext}\right.  _{k}^{l}(V)
\]
such that%
\begin{align}
\underset{\smile}{\gamma}\tau(X_{1},\ldots,X_{k},\Phi^{1},\ldots,\Phi^{l})  &
=\underset{\smile}{\gamma}\circ\tau(X_{1},\ldots,X_{k},\Phi^{1},\ldots
,\Phi^{l})\nonumber\\
&  +\tau(\underset{\smile}{\gamma}^{\bigtriangleup}(X_{1}),\ldots,X_{k}%
,\Phi^{1},\ldots,\Phi^{l})\nonumber\\
&  \cdots+\tau(X_{1},\ldots,\underset{\smile}{\gamma}^{\bigtriangleup}%
(X_{k}),\Phi^{1},\ldots,\Phi^{l})\nonumber\\
&  -\tau(X_{1},\ldots,X_{k},\underset{\smile}{\gamma}(\Phi^{1}),\ldots
,\Phi^{l})\nonumber\\
&  \cdots-\tau(X_{1},\ldots,X_{k},\Phi^{1},\ldots,\underset{\smile}{\gamma
}(\Phi^{l}) \label{GPO25}%
\end{align}
for each $X_{1},\ldots,X_{k}\in%
{\textstyle\bigwedge}
V$ and $\Phi^{1},\ldots,\Phi^{l}\in%
{\textstyle\bigwedge}
V^{\ast}.$
\end{definition}

For instance, for $\tau\in\left.  \overset{\ast}{ext}\right.  _{1}^{0}(V)$ the
above definition yields%
\begin{equation}
\underset{\smile}{\gamma}\tau=\underset{\smile}{\gamma}\circ\tau+\tau
\circ\underset{\smile}{\gamma}^{\bigtriangleup}, \label{GPO26}%
\end{equation}
and for $\tau\in\left.  \overset{\left.  \ast\right.  }{ext}\right.  _{0}%
^{1}(V)$ it holds%
\begin{equation}
\underset{\smile}{\gamma}\tau=\underset{\smile}{\gamma}\circ\tau-\tau
\circ\underset{\smile}{\gamma}=\left[  \underset{\smile}{\gamma},\tau\right]
. \label{GPO27}%
\end{equation}

We present some of the main properties of the action of the contracted
extension operator of a vector operator $\gamma$ on multivector extensors.

\begin{proposition}
For all $\tau\in\left.  \overset{\left.  {}\right.  }{ext}\right.  _{k}%
^{l}(V)$ and $\sigma\in\left.  \overset{\left.  {}\right.  }{ext}\right.
_{r}^{s}(V):$%
\begin{equation}
\underset{\smile}{\gamma}(\tau\wedge\sigma)=(\underset{\smile}{\gamma}%
\tau)\wedge\sigma+\tau\wedge(\underset{\smile}{\gamma}\sigma). \label{GPO21}%
\end{equation}

For all $\tau\in\left.  \overset{\ast}{ext}\right.  _{k}^{l}(V)$ and
$\sigma\in\left.  \overset{\left.  {}\right.  }{ext}\right.  _{r}^{s}(V):$%
\begin{align}
\underset{\smile}{\gamma}\left\langle \tau,\sigma\right\rangle  &
=-\left\langle \underset{\smile}{\gamma}^{\bigtriangleup}\tau,\sigma
\right\rangle +\left\langle \tau,\underset{\smile}{\gamma}\sigma\right\rangle
,\label{GPO22}\\
\underset{\smile}{\gamma}\left\langle \tau,\sigma\right\vert  &
=-\left\langle \underset{\smile}{\gamma}^{\bigtriangleup}\tau,\sigma
\right\vert +\left\langle \tau,\underset{\smile}{\gamma}\sigma\right\vert
,\label{GPO23}\\
\underset{\smile}{\gamma}\left\vert \sigma,\tau\right\rangle  &  =\left\vert
\underset{\smile}{\gamma}\sigma,\tau\right\rangle -\left\vert \sigma
,\underset{\smile}{\gamma}^{\bigtriangleup}\tau\right\rangle . \label{GPO24}%
\end{align}

\end{proposition}

\begin{proof}
We present only the proof for the property given by Eq. (\ref{GPO21}), the
others are similar. Without any loss of generality, we give the proof for the
particular case where $\tau\in ext(%
{\textstyle\bigwedge\nolimits_{1}^{\diamond}}
V,%
{\textstyle\bigwedge\nolimits_{2}^{\diamond}}
V^{\ast};%
{\textstyle\bigwedge\nolimits^{\diamond}}
V^{\ast})$ and $\sigma\in ext(%
{\textstyle\bigwedge\nolimits_{3}^{\diamond}}
V,%
{\textstyle\bigwedge\nolimits_{4}^{\diamond}}
V^{\ast};%
{\textstyle\bigwedge\nolimits^{\diamond}}
V)$. Take $X\in%
{\textstyle\bigwedge\nolimits_{1}^{\diamond}}
V,$ $Y\in%
{\textstyle\bigwedge\nolimits_{3}^{\diamond}}
V$ and $\Phi\in%
{\textstyle\bigwedge\nolimits_{2}^{\diamond}}
V^{\ast},$ $\Psi\in%
{\textstyle\bigwedge\nolimits_{4}^{\diamond}}
V^{\ast}$. By using the definition of $\underset{\smile}{\gamma}$\ , we have%
\begin{equation}%
\begin{array}
[c]{ll}%
\underset{\smile}{\gamma}(\tau\wedge\sigma)\left(  X,Y,\Phi,\Psi\right)  &
=\underset{\smile}{\gamma}\circ(\tau\wedge\sigma)\left(  X,Y,\Phi,\Psi\right)
-(\tau\wedge\sigma)\left(  \underset{\smile}{\gamma}X,Y,\Phi,\Psi\right) \\
& -(\tau\wedge\sigma)\left(  X,\underset{\smile}{\gamma}Y,\Phi,\Psi\right)
+(\tau\wedge\sigma)\left(  X,Y,\underset{\smile}{\gamma}\Phi,\Psi\right) \\
& +(\tau\wedge\sigma)\left(  X,Y,\Phi,\underset{\smile}{\gamma}\Psi\right)  .
\end{array}
\label{EN4}%
\end{equation}
Now, using the property (\ref{GPO5}) we can write the first term of right side
of the Eq. (\ref{EN4}) as%
\[
\underset{\smile}{\gamma}\circ(\tau\wedge\sigma)\left(  X,Y,\Phi,\Psi\right)
=\underset{\smile}{\gamma}\circ\tau\left(  X,\Phi\right)  \wedge\sigma\left(
Y,\Psi\right)  +\tau\left(  X,\Phi\right)  \wedge\underset{\smile}{\gamma
}\circ\sigma\left(  Y,\Psi\right)  ,
\]
and remembering that
\[
(\tau\wedge\sigma)\left(  X,Y,\Phi,\Psi\right)  =\tau\left(  X,\Phi\right)
\wedge\sigma\left(  Y,\Psi\right)  ,
\]
the Eq. (\ref{EN4}) can be written as
\[%
\begin{array}
[c]{ll}%
\underset{\smile}{\gamma}(\tau\wedge\sigma)\left(  X,Y,\Phi,\Psi\right)  &
=\left[  \underset{\smile}{\gamma}\circ\tau\left(  X,\Phi\right)  -\tau\left(
\underset{\smile}{\gamma}X,\Phi\right)  +\tau\left(  X,\underset{\smile
}{\gamma^{\bigtriangleup}}\Phi\right)  \right]  \wedge\sigma\left(
Y,\Psi\right) \\
& +\tau\left(  X,\Phi\right)  \wedge\left[  \underset{\smile}{\gamma}%
\circ\sigma\left(  Y,\Psi\right)  -\sigma\left(  \underset{\smile}{\gamma
}Y,\Psi\right)  +\sigma\left(  Y,\underset{\smile}{\gamma^{\bigtriangleup}%
}\Psi\right)  \right]
\end{array}
\]
or
\[%
\begin{array}
[c]{ll}%
\underset{\smile}{\gamma}(\tau\wedge\sigma)\left(  X,Y,\Phi,\Psi\right)  &
=\left[  \left(  \underset{\smile}{\gamma}\tau\right)  \left(  X,\Phi\right)
\right]  \wedge\sigma\left(  Y,\Psi\right)  +\tau\left(  X,\Phi\right)
\wedge\left[  \left(  \underset{\smile}{\gamma}\sigma\right)  \left(
Y,\Psi\right)  \right] \\
& =\left(  \underset{\smile}{\gamma}\tau\wedge\sigma+\tau\wedge\underset
{\smile}{\gamma}\sigma\right)  \left(  X,Y,\Phi,\Psi\right)  ,
\end{array}
\]
and the property is proved.
\end{proof}

We present some properties for the action of the contacted extension operator
of a form operator $\gamma$ on multiform extensors.

\begin{proposition}
For all $\tau\in\left.  \overset{\left.  \ast\right.  }{ext}\right.  _{k}%
^{l}(V)$ and $\sigma\in\left.  \overset{\left.  \ast\right.  }{ext}\right.
_{r}^{s}(V):$%
\begin{equation}
\underset{\smile}{\gamma}(\tau\wedge\sigma)=(\underset{\smile}{\gamma}%
\tau)\wedge\sigma+\tau\wedge(\underset{\smile}{\gamma}\sigma). \label{GPO28}%
\end{equation}

For all $\tau\in\left.  \overset{\ast}{ext}\right.  _{k}^{l}(V)$ and
$\sigma\in\left.  \overset{\left.  {}\right.  }{ext}\right.  _{r}^{s}(V):$%
\begin{align}
\underset{\smile}{\gamma}\left\langle \tau,\sigma\right\rangle  &
=\left\langle \underset{\smile}{\gamma}\tau,\sigma\right\rangle -\left\langle
\tau,\underset{\smile}{\gamma}^{\bigtriangleup}\sigma\right\rangle
,\label{GPO29}\\
\underset{\smile}{\gamma}\left\langle \tau,\sigma\right\vert  &  =\left\langle
\underset{\smile}{\gamma}\tau,\sigma\right\vert -\left\langle \tau
,\underset{\smile}{\gamma}^{\bigtriangleup}\sigma\right\vert ,\label{GPO30}\\
\underset{\smile}{\gamma}\left\vert \sigma,\tau\right\rangle  &  =-\left\vert
\underset{\smile}{\gamma}^{\bigtriangleup}\sigma,\tau\right\rangle +\left\vert
\sigma,\underset{\smile}{\gamma}\tau\right\rangle . \label{GPO31}%
\end{align}

\end{proposition}

Proofs are analogous to the ones of the previous proposition.

\subsection{Multivector and Multiform Fields}

Let $U$ be an open set of a smooth manifold $M$ (i.e., $U\subseteq M$). The
set of smooth\footnote{Smooth in this paper means, $\mathcal{C}^{\infty}%
$-differentiable or at least enough differentiable in order for our statements
to hold.} scalar fields on $U,$ as well-known, has a natural structure of
\emph{ring} (\emph{with identity}), and it will be denoted by $\mathcal{S}%
(U).$ The set of smooth vector fields on $U,$ as well-known, have natural
structure of \emph{modules over} $\mathcal{S}(U).$ It will be denoted by
$\mathcal{V}(U)$. The set of smooth form fields on $U$ \emph{can be
identified} with the \emph{dual module} for $\mathcal{V}(U)$. It could be
denoted by $\mathcal{V}^{\ast}(U)$.

Let $M$ be a $n$-dimensional differentiable manifold and let $U\subset M$ be
an open set. A $k$-vector mapping%
\[
X^{k}:U\longrightarrow\underset{p\in U}{%
{\textstyle\bigcup}
}%
{\textstyle\bigwedge\nolimits^{k}}
T_{p}^{\ast}M,
\]
such that for each $p\in U,$ $X_{(p)}^{k}\in%
{\textstyle\bigwedge\nolimits^{k}}
T_{p}^{\ast}M$ is called a $k$\emph{-vector field on }$U.$

Such $X^{k}$ with $1\leq k\leq n$ is said to be a smooth $k$-vector field on
$U,$ if and only if, for all $\omega^{1},\ldots,\omega^{k}\in\mathcal{V}%
^{\ast}(U)$, the scalar mapping defined by%
\begin{equation}
U\ni p\mapsto X_{(p)}^{k}(\omega_{(p)}^{1},\ldots,\omega_{(p)}^{k}%
)\in\mathbb{R} \label{MMF2}%
\end{equation}
is a smooth scalar field on $U$.

A multivector mapping%
\[
X:U\longrightarrow\underset{p\in U}{%
{\textstyle\bigcup}
}%
{\textstyle\bigwedge}
T_{p}^{\ast}M,
\]
such that for each $p\in U$, $X_{(p)}\in%
{\textstyle\bigwedge}
T_{p}^{\ast}M$ is called a \emph{multivector field on }$U.$

Any multivector at $p\in M$ can be written as a sum of $k$-vectors (i.e.,
homogeneous multivectors of degree $k$) at $p\in M,$ with $k$ running from
$k=0$ to $k=n$, i.e., there exist exactly $n+1$ homogeneous multivector of
degree $k$ fields on $U,$ conveniently denote by $X^{0},X^{1},\ldots,X^{n},$
such that for every $p\in U,$%
\begin{equation}
X_{(p)}=X_{(p)}^{0}+X_{(p)}^{1}+\cdots+X_{(p)}^{n}. \label{MMF4}%
\end{equation}

We say that $X$ is a smooth multivector field on $U$ when each of one of the
$X^{0},X^{1},\ldots,X^{n}$ is a smooth $k$-vector field on $U.$

We emphasize that according with the definitions of smoothness as given above,
a smooth $k$-vector field on $U$ \emph{can be identified} to a $k$-vector over
$\mathcal{V}(U)$, and a smooth multivector field on $U$ \emph{can be seen
properly} as a multivector over $\mathcal{V}(U).$ Thus, the set of smooth
$k$-vector fields on $U$ may be denoted by $%
{\textstyle\bigwedge\nolimits^{k}}
\mathcal{V}(U)$, and the set of smooth multivector fields on $U$ may be
denoted by $%
{\textstyle\bigwedge}
\mathcal{V}(U)$.

A $k$-form mapping%
\[
\Phi_{k}:U\longrightarrow\underset{p\in U}{%
{\textstyle\bigcup}
}%
{\textstyle\bigwedge\nolimits^{k}}
T_{p}^{\ast}M
\]
such that for each $p\in U$, $\Phi_{k(p)}\in%
{\textstyle\bigwedge\nolimits^{k}}
T_{p}^{\ast}M$ is called a $k$\emph{-form field on }$U.$

Such $\Phi_{k}$ with $1\leq k\leq n$ is said to be a smooth $k$-form field on
$U,$ if and only if, for all $v_{1},\ldots,v_{k}\in\mathcal{V}(U)$, the scalar
mapping defined by%
\begin{equation}
U\ni p\mapsto\Phi_{k(p)}(v_{1(p)},\ldots,v_{k(p)})\in\mathbb{R} \label{MMF6}%
\end{equation}
is a smooth scalar field on $U$.

A multiform mapping%
\[
\Phi:U\longrightarrow\underset{p\in U}{%
{\textstyle\bigcup}
}%
{\textstyle\bigwedge}
T_{p}^{\ast}M
\]
such that for each $p\in U$, $\Phi_{(p)}\in%
{\textstyle\bigwedge}
T_{p}^{\ast}M$ is called a \emph{multiform field on }$U.$

Any multiform at $p\in M$ can be written (see, e.g.,\cite{rodoliv2006}) as a
sum of $k$-forms (i.e., homogeneous multiforms of degree $k$) at $p\in M$ with
$k$ running from $k=0$ to $k=n.$ It follows that there exist exactly $n+1$
homogeneous multiform of degree $k$ fields on $U,$ named as $\Phi_{0},\Phi
_{1},\ldots,\Phi_{n}$ such that%
\begin{equation}
\Phi_{(p)}=\Phi_{0(p)}+\Phi_{1(p)}+\cdots+\Phi_{n(p)} \label{MMF8}%
\end{equation}
for every $p\in U.$

We say that $\Phi$ is a smooth multiform field on $U,$ if and only if, each of
$\Phi_{0},\Phi_{1},\ldots,\Phi_{n}$ is just a smooth $k$-form field on $U. $

Note that according with the definitions of smoothness as given above, a
smooth $k$-form field on $U$ \emph{can be identified} with a $k$-form over
$\mathcal{V}(U),$ and a smooth multiform field on $U$ \emph{can be seen} as a
multiform over $\mathcal{V}(U).$ Thus, the set of smooth $k$-form fields on
$U$ will be denoted by $%
{\textstyle\bigwedge\nolimits^{k}}
\mathcal{V}^{\ast}(U),$ and the set of smooth multiform fields on $U$ will be
denoted by $%
{\textstyle\bigwedge}
\mathcal{V}^{\ast}(U).$

\subsubsection{Algebra of Multivector and Multiform Fields}

We recall first the module (over a ring) structure operations of the set of
smooth multivector fields on $U$ and of the set of smooth multiform fields on
$U$. We recall also the concept of the exterior product both of smooth
multivector fields as well as of smooth multiform fields on $U$ and we present
the definitions of the duality products of a given smooth multivector fields
on $U$ by a smooth multiform field on $U.$

The addition of multivector fields $X$ and $Y,$ or multiform fields $\Phi$ and
$\Psi,$ is defined by%
\begin{align}
\left(  X+Y\right)  _{(p)}  &  =X_{(p)}+Y_{(p)},\label{MMF9}\\
\left(  \Phi+\Psi\right)  _{(p)}  &  =\Phi_{(p)}+\Psi_{(p)}, \label{MMF10}%
\end{align}
for every $p\in U.$

The scalar multiplication of a multivector field $X,$ or a multiform field
$\Phi,$ by a scalar field $f,$ is defined by%
\begin{align}
\left(  fX\right)  _{(p)}  &  =f(p)X_{(p)},\label{MMF11}\\
\left(  f\Phi\right)  _{(p)}  &  =f(p)\Phi_{(p)}, \label{MMF12}%
\end{align}
for every $p\in U.$

The exterior product of multivector fields $X$ and $Y,$ and the exterior
product of multiform fields $\Phi$ and $\Psi,$ are defined by%
\begin{align}
\left(  X\wedge Y\right)  _{(p)}  &  =X_{(p)}\wedge Y_{(p)},\label{MMF13}\\
\left(  \Phi\wedge\Psi\right)  _{(p)}  &  =\Phi_{(p)}\wedge\Psi_{(p)},
\label{MMF14}%
\end{align}
for every $p\in U.$

Each module, of either the smooth multivector fields on $U$, or the smooth
multiform fields on $U$, endowed with the respective exterior product has a
natural structure of associative algebra. They are called \emph{the exterior
algebras of multivector and multiform fields on} $U.$

The duality scalar product of a multiform field $\Phi$ with a multivector
field $X$ is (see algebraic details in \cite{fmcr1}) the scalar field
$\left\langle \Phi,X\right\rangle $ defined by%
\begin{equation}
\left\langle \Phi,X\right\rangle (p)=\left\langle \Phi_{(p)},X_{(p)}%
\right\rangle , \label{MMF15}%
\end{equation}
for every $p\in U.$

The duality left contracted product of a multiform field $\Phi$ with a
multivector field $X$ (or, a multivector field $X$ with a multiform field
$\Phi$) is the multivector field $\left\langle \Phi,X\right\vert $
(respectively, the multiform field $\left\langle X,\Phi\right\vert $) defined
by%
\begin{align}
\left\langle \Phi,X\right\vert _{(p)}  &  =\left\langle \Phi_{(p)}%
,X_{(p)}\right\vert ,\label{MMF16}\\
\left\langle X,\Phi\right\vert _{(p)}  &  =\left\langle X_{(p)},\Phi
_{(p)}\right\vert , \label{MMF17}%
\end{align}
for every $p\in U.$

The duality right contracted product of a multiform field $\Phi$ with a
multivector field $X$ (or, a multivector field $X$ with a multiform field
$\Phi$) is the multiform field $\left\vert \Phi,X\right\rangle $
(respectively, the multivector field $\left\vert X,\Phi\right\rangle $)
defined by%
\begin{align}
\left\vert \Phi,X\right\rangle _{(p)}  &  =\left\vert \Phi_{(p)}%
,X_{(p)}\right\rangle ,\label{MMF18}\\
\left\vert X,\Phi\right\rangle _{(p)}  &  =\left\vert X_{(p)},\Phi
_{(p)}\right\rangle , \label{MMF19}%
\end{align}
for every $p\in U.$

Each duality contracted product of smooth multivector fields on $U$ with
smooth multiform fields on $U$ yields a natural structure of
\textit{non-associative} algebra.

\subsection{Extensor Fields}

Let $T_{p}M$ the set of all tangent vectors to $M$ at $p.$ A multivector
extensor mapping
\[
\tau:U\longrightarrow\underset{p\in U}{%
{\textstyle\bigcup}
}\left.  \overset{\left.  {}\right.  }{ext}\right.  _{k}^{l}(T_{p}M)
\]
such that for each $p\in U$, $\tau_{(p)}\in\left.  \overset{\left.  {}\right.
}{ext}\right.  _{k}^{l}(T_{p}M)$ is called a\emph{\ multivector extensor field
of }$k$\emph{\ multivector and }$l$\emph{\ multiform variables on }$U.$

A multiform extensor mapping
\[
\upsilon:U\longrightarrow\underset{p\in U}{%
{\textstyle\bigcup}
}\left.  \overset{\ast}{ext}\right.  _{k}^{l}(T_{p}M)
\]
such that for each $p\in U$, $\upsilon_{(p)}\in\left.  \overset{\ast}%
{ext}\right.  _{k}^{l}(T_{p}M)$ is called a\emph{\ multiform }%
\textit{extensor}\emph{\textbf{\ }field of }$k$\emph{\ multivector and }%
$l$\emph{\ multiform variables on }$U.$

In the above formulas, $\left.  \overset{\left.  {}\right.  }{ext}\right.
_{k}^{l}(T_{p}M)$ is a short notation for the space of multivector extensors
of $k$ multivector and $l$ multiform variables over $T_{p}M,$ i.e., for each
$p\in U:$%
\[
\left.  \overset{\left.  {}\right.  }{ext}\right.  _{k}^{l}(T_{p}M):=ext(%
{\displaystyle\bigwedge\nolimits_{1}^{\Diamond}}
T_{p}M,\ldots,%
{\displaystyle\bigwedge\nolimits_{k}^{\Diamond}}
T_{p}M,%
{\displaystyle\bigwedge\nolimits_{1}^{\Diamond}}
T_{p}^{\star}M,\ldots,%
{\displaystyle\bigwedge\nolimits_{l}^{\Diamond}}
T_{p}^{\star}M;%
{\displaystyle\bigwedge\nolimits^{\Diamond}}
T_{p}M),
\]
and $\left.  \overset{\ast}{ext}\right.  _{k}^{l}(T_{p}M)$ is a short notation
for the space of multiform \textit{extensors} of $k$ multivector and $l$
multiform variables over $T_{p}M,$ i.e., for each $p\in U:$%
\[
\left.  \overset{\ast}{ext}\right.  _{k}^{l}(T_{p}M):=ext(%
{\displaystyle\bigwedge\nolimits_{1}^{\Diamond}}
T_{p}M,\ldots,%
{\displaystyle\bigwedge\nolimits_{k}^{\Diamond}}
T_{p}M,%
{\displaystyle\bigwedge\nolimits_{1}^{\Diamond}}
T_{p}^{\star}M,\ldots,%
{\displaystyle\bigwedge\nolimits_{l}^{\Diamond}}
T_{p}^{\star}M;%
{\displaystyle\bigwedge\nolimits^{\Diamond}}
T_{p}^{\star}M),
\]
where $T_{p}^{\star}M$ is the dual space of $T_{p}M.$

Let us denote the smooth multivector fields: $U\ni p\mapsto X_{1(p)}\in%
{\displaystyle\bigwedge\nolimits_{1}^{\Diamond}}
T_{p}M,\ldots$, $U\ni p\mapsto X_{k(p)}\in%
{\displaystyle\bigwedge\nolimits_{k}^{\Diamond}}
T_{p}M$, and $U\ni p\mapsto X_{(p)}\in%
{\displaystyle\bigwedge\nolimits^{\Diamond}}
T_{p}M$, respectively by $%
{\displaystyle\bigwedge\nolimits_{1}^{\Diamond}}
\mathcal{V}(U),\ldots$, $%
{\displaystyle\bigwedge\nolimits_{k}^{\Diamond}}
\mathcal{V}(U)$, and $%
{\displaystyle\bigwedge\nolimits^{\Diamond}}
\mathcal{V}(U)$. Let us denote the smooth multiform fields: $U\ni p\mapsto
\Phi_{(p)}^{1}\in%
{\displaystyle\bigwedge\nolimits_{1}^{\Diamond}}
T_{p}^{\ast}M,\ldots$, $U\ni p\mapsto\Phi_{(p)}^{l}\in%
{\displaystyle\bigwedge\nolimits_{l}^{\Diamond}}
T_{p}^{\ast}M$, and $U\ni p\mapsto\Phi_{(p)}^{l}\in%
{\displaystyle\bigwedge\nolimits^{\Diamond}}
T_{p}^{\ast}M$, by $%
{\displaystyle\bigwedge\nolimits_{1}^{\Diamond}}
\mathcal{V}^{\ast}(U),$ $\ldots,$ $%
{\displaystyle\bigwedge\nolimits_{l}^{\Diamond}}
\mathcal{V}^{\ast}(U)$ and $%
{\displaystyle\bigwedge\nolimits^{\Diamond}}
\mathcal{V}^{\ast}(U)$.

Such a multivector extensor field $\tau$ will be said to be smooth, if and
only if, for all $X_{1}\in%
{\displaystyle\bigwedge\nolimits_{1}^{\Diamond}}
\mathcal{V}(U),\ldots,X_{k}\in%
{\displaystyle\bigwedge\nolimits_{k}^{\Diamond}}
\mathcal{V}(U),$ and for all $\Phi^{1}\in%
{\displaystyle\bigwedge\nolimits_{1}^{\Diamond}}
\mathcal{V}^{\ast}(U),\ldots,\Phi^{l}\in%
{\displaystyle\bigwedge\nolimits_{l}^{\Diamond}}
\mathcal{V}^{\star}(U)$, the multivector mapping defined by
\begin{equation}
U\ni p\mapsto\tau_{(p)}(X_{1(p)},\ldots,X_{k(p)},\Phi_{(p)}^{1},\ldots
,\Phi_{(p)}^{l})\in%
{\displaystyle\bigwedge\nolimits^{\Diamond}}
T_{p}M \label{EF3}%
\end{equation}
is a smooth multivector field on $U,$ (i.e., an object living on $%
{\displaystyle\bigwedge\nolimits^{\Diamond}}
\mathcal{V}(U)$).

Such a multiform extensor field $\upsilon$ will be said to be smooth , if and
only if, for all $X_{1}\in%
{\displaystyle\bigwedge\nolimits_{1}^{\Diamond}}
\mathcal{V}(U),\ldots,X_{k}\in%
{\displaystyle\bigwedge\nolimits_{k}^{\Diamond}}
\mathcal{V}(U),$ and for all $\Phi^{1}\in%
{\displaystyle\bigwedge\nolimits_{1}^{\Diamond}}
\mathcal{V}^{\ast}(U),\ldots,\Phi^{l}\in%
{\displaystyle\bigwedge\nolimits_{l}^{\Diamond}}
\mathcal{V}^{\ast}(U)$, the multiform mapping defined by
\begin{equation}
U\ni p\mapsto\upsilon_{(p)}(X_{1(p)},\ldots,X_{k(p)},\Phi_{(p)}^{1}%
,\ldots,\Phi_{(p)}^{l})\in%
{\displaystyle\bigwedge\nolimits^{\Diamond}}
T_{p}^{\star}M \label{EF4}%
\end{equation}
is a smooth multiform field on $U,$ (i.e., an object living on $%
{\displaystyle\bigwedge^{\Diamond}}
\mathcal{V}^{\star}(U)$).

We emphasize\footnote{A \emph{short name} for a multivector (or, multiform)
extensor of $k$ multivector and $l$ multiform variables could be: a $\binom
{l}{k}$ multivector (respectively, multiform) extensor.} that according with
the definitions of smoothness as given above, a smooth\emph{\ }$\binom{l}{k}$
multivector extensor field on $U$ \emph{can be identified} to a $\binom{l}{k}
$ multivector extensor over $\mathcal{V}(U).$ It is also true that a
smooth\emph{\ }$\binom{l}{k}$ multiform extensor field on $U$ \emph{can be
properly seen }as a $\binom{l}{k}$ multiform extensor over $\mathcal{V}(U).$

Thus, the set of smooth\emph{\ }$\binom{l}{k}$ multivector extensor fields on
$U$ is just a module over $\mathcal{S}(U)$ which could be denoted by $\left.
\overset{\left.  {}\right.  }{ext}\right.  _{k}^{l}\mathcal{V}(U).$ And, the
set of smooth\emph{\ }$\binom{l}{k}$ multiform extensor fields on $U$ is also
a module over $\mathcal{S}(U)$ which can be symbolized as $\left.
\overset{\ast}{ext}\right.  _{k}^{l}\mathcal{V}(U).$

\subsubsection{Algebras of Extensor Fields}

We define now the exterior products of smooth multivector extensor fields on
$U$ and smooth multiform extensor\textbf{\ }fields on $U$. We also present the
definitions of smooth multivector extensor\textbf{\ }fields on $U$ with smooth
multiform extensor\textbf{\ }fields on $U.$

The exterior product of either multivector extensor fields or multiform
extensor fields $\tau$ and $\sigma$ is defined as%
\begin{equation}
\left(  \tau\wedge\sigma\right)  _{(p)}=\tau_{(p)}\wedge\sigma_{(p)}
\label{EF5}%
\end{equation}
for every $p\in U$.

Each module over $\mathcal{S}(U)$ of either the smooth multivector extensor
fields on $U$ or the smooth multiform extensor fields on $U$ endowed with the
respective exterior product is an associative algebra.

The duality scalar product of a multiform extensor field $\tau$ with a
multivector extensor field $\sigma$ is a scalar extensor field $\left\langle
\tau,\sigma\right\rangle $ defined by%
\begin{equation}
\left\langle \tau,\sigma\right\rangle _{(p)}=\left\langle \tau_{(p)}%
,\sigma_{(p)}\right\rangle , \label{EF6}%
\end{equation}
for every $p\in U.$

The duality left contracted product of a multiform extensor field $\tau$ with
a multivector extensor field $\sigma$ (or, a multivector extensor field $\tau$
with a multiform extensor field $\sigma$) is the multivector extensor field
(respectively, the multiform extensor field) denoted by $\left\langle
\tau,\sigma\right\vert $ and defined by%
\begin{equation}
\left\langle \tau,\sigma\right\vert _{(p)}=\left\langle \tau_{(p)}%
,\sigma_{(p)}\right\vert , \label{EF7}%
\end{equation}
for every $p\in U.$

The duality right contracted product of a multiform extensor field $\tau$ with
a multivector extensor field $\sigma$ (or, a multivector extensor field $\tau$
with a multiform extensor field $\sigma$) is the multiform extensor field
(respectively, the multivector extensor field) named as $\left\vert
\tau,\sigma\right\rangle $, and defined by%
\begin{equation}
\left\vert \tau,\sigma\right\rangle _{(p)}=\left\vert \tau_{(p)},\sigma
_{(p)}\right\rangle \label{EF8}%
\end{equation}
for every $p\in U.$

As in the case of multivectors and multiforms, each duality contracted product
of smooth multivector extensor fields on $U$ with smooth multiform fields on
$U$ yields a non-associative algebra.

\section{Parallelism Structure and Covariant Derivatives}

In this section we present a theory of a general parallelism structure on an
arbitrary real differentiable manifold $M$ of dimension $n$. We present, with
this structure, a detailed theory of the covariant derivatives, deformed
covariant derivatives and relative covariant derivatives of
\textit{multivector}, \textit{multiform and extensor }fields.

We detailed some particular formulas valid for a symmetric parallelism
structure. These concepts play an important role in different physical
theories, in particular they are essential for those that want to have a deep
understanding of the geometric theories of the gravitational field.(see, e.g.,
\cite{fr2010, Sach, Eddington}) and do not want to mislead the curvature of a
connection defined in a manifold $M$ ($dimM=m)$ with the fact that $M$ may be
a bended submanifold (a brane) embedded in a Euclidean or speudo-Euclidean
manifold $\mathring{M}\simeq\mathbb{R}^{n}$ ($n>m$). Bending of a brane is
characterized by the shape extensor. See \cite{rw2013} for details.

\subsection{Parallelism Structure}

Let $U$ be an open set of the smooth manifolf $M.$

\begin{definition}
A connection for $M$ is a smooth $2$ vector variables vector operator field
$U\subset M$,
\[
\Gamma:\mathcal{V}(U)\times\mathcal{V}(U)\longrightarrow\mathcal{V}(U),
\]
for all $U\subset M$ such that it satisfies the following axioms:

\emph{(i)}\textbf{ }for all $f,g\in\mathcal{S}(U)$ and $a,b,v\in
\mathcal{V}(U)$%
\begin{equation}
\Gamma(fa+gb,v)=f\Gamma(a,v)+g\Gamma(b,v), \label{PS1}%
\end{equation}

\emph{(ii}\textbf{)} for all $f,g\in\mathcal{S}(U)$ and $a,v,w\in
\mathcal{V}(U)$%
\begin{equation}
\Gamma(a,fv+gw)=(af)v+(ag)w+f\Gamma(a,v)+g\Gamma(a,w), \label{PS2}%
\end{equation}

\end{definition}

The behavior of $\Gamma$ with respect to the first variable will be called
\emph{strong linearity}, and the behavior of $\Gamma$ with respect to the
second variable will be called \emph{quasi linearity}.

The restriction of $\Gamma$ to $U,$ may be called \emph{parallelism structure}
on $U.$ We will denote this by $\left\langle U,\Gamma\right\rangle $.

\subsection{Covariant Derivative of Multivector and Multiform Fields}

\begin{definition}
The $\emph{a}$\emph{-Directional Covariant Derivative} ($\emph{a}$%
\emph{-DCD})\emph{\ }of a smooth multivector field on $U$, associated with
$\left\langle U,\Gamma\right\rangle $, is the mapping%
\[%
{\displaystyle\bigwedge}
\mathcal{V}(U)\ni X\mapsto\nabla_{a}X\in%
{\displaystyle\bigwedge}
\mathcal{V}(U),
\]
such that the following axioms are satisfied:

\emph{(i)}\textbf{ }For all $f\in\mathcal{S}(U):$%
\begin{equation}
\nabla_{a}f=af. \label{CDMMF1}%
\end{equation}

\emph{(ii) }For all $X^{k}\in%
{\textstyle\bigwedge\nolimits^{k}}
\mathcal{V}(U)$ with $k\geq1:$%
\begin{align}
\nabla_{a}X^{k}(\omega^{1},\ldots,\omega^{k})  &  =aX^{k}(\omega^{1}%
,\ldots,\omega^{k})\nonumber\\
&  -X^{k}(\nabla_{a}\omega^{1},\ldots,\omega^{k})\cdots-X^{k}(\omega
^{1},\ldots,\nabla_{a}\omega^{k}), \label{CDMMF2}%
\end{align}
for every $\omega^{1},\ldots,\omega^{k}\in\mathcal{V}^{\star}(U).$
\end{definition}

\emph{(iii)} For all $X\in%
{\textstyle\bigwedge}
\mathcal{V}(U)$, if $X=\overset{n}{\underset{k=0}{%
{\textstyle\sum}
}}X^{k}$, then%
\begin{equation}
\nabla_{a}X=\overset{n}{\underset{k=0}{%
{\textstyle\sum}
}}\nabla_{a}X^{k}. \label{CDMMF3}%
\end{equation}

The basic properties of the $a$-\textit{DCD} of smooth multivector fields are:

The $a$-Directional Covariant Derivative Operator ($a$-\textit{DCDO)}
$\nabla_{a}$ when acting on multivector fields is grade-preserving, i.e.,
\begin{equation}
\text{if }X\in%
{\textstyle\bigwedge\nolimits^{k}}
\mathcal{V}(U),\text{ then }\nabla_{a}X\in%
{\textstyle\bigwedge\nolimits^{k}}
\mathcal{V}(U). \label{CDMMF4}%
\end{equation}

For all $f\in\mathcal{S}(U),$ $a,b\in\mathcal{V}(U)$ and $X\in%
{\displaystyle\bigwedge}
\mathcal{V}(U)$%
\begin{align}
\nabla_{a+b}X  &  =\nabla_{a}X+\nabla_{b}X,\nonumber\\
\nabla_{fa}X  &  =f\nabla_{a}X. \label{CDMMF5}%
\end{align}

For all $f\in\mathcal{S}(U),$ $a\in\mathcal{V}(U)$ and $X,Y\in%
{\displaystyle\bigwedge}
\mathcal{V}(U)$%
\begin{align}
\nabla_{a}(X+Y)  &  =\nabla_{a}X+\nabla_{a}Y,\nonumber\\
\nabla_{a}(fX)  &  =(af)X+f\nabla_{a}X. \label{CDMMF6}%
\end{align}

For all $a\in\mathcal{V}(U)$ and $X,Y\in%
{\displaystyle\bigwedge}
\mathcal{V}(U)$%
\begin{equation}
\nabla_{a}(X\wedge Y)=(\nabla_{a}X)\wedge Y+X\wedge\nabla_{a}Y. \label{CDMMF7}%
\end{equation}

\begin{definition}
The $a$-\textit{DCD} of a smooth multiform field on $U$ associated with
$\left\langle U,\Gamma\right\rangle $ is the mapping%
\[%
{\displaystyle\bigwedge}
\mathcal{V}^{\star}(U)\ni\Phi\mapsto\nabla_{a}\Phi\in%
{\displaystyle\bigwedge}
\mathcal{V}^{\star}(U),
\]
such that the following axioms are satisfied:
\end{definition}

\emph{(i)}\textbf{ }For all $f\in\mathcal{S}(U):$%
\begin{equation}
\nabla_{a}f=af. \label{CDMMF8}%
\end{equation}

\emph{(ii)} For all $\Phi_{k}\in%
{\displaystyle\bigwedge^{k}}
\mathcal{V}^{\star}(U)$ with $k\geq1:$%
\begin{align}
\nabla_{a}\Phi_{k}(v_{1},\ldots,v_{k})  &  =a\Phi_{k}(v_{1},\ldots
,v_{k})\nonumber\\
&  -\Phi_{k}(\nabla_{a}v_{1},\ldots,v_{k})\cdots-\Phi_{k}(v_{1},\ldots
,\nabla_{a}v_{k}), \label{CDMMF9}%
\end{align}
for every $v_{1},\ldots,v_{k}\in\mathcal{V}(U).$

\emph{(iii)} For all $\Phi\in%
{\displaystyle\bigwedge}
\mathcal{V}^{\star}(U):$ if $\Phi=\underset{k=0}{\overset{n}{%
{\textstyle\sum}
}}\Phi_{k},$ then%
\begin{equation}
\nabla_{a}\Phi=\underset{k=0}{\overset{n}{%
{\textstyle\sum}
}}\nabla_{a}\Phi_{k}. \label{CDMMF10}%
\end{equation}

The basic properties for the $a$-\textit{DCD} of smooth multiform fields are:

The $a$-\textit{DCDO} $\nabla_{a}$ when acting on multiform fields is
grade-preserving, i.e.,
\begin{equation}
\text{if }\Phi\in%
{\textstyle\bigwedge\nolimits^{k}}
\mathcal{V}^{\star}(U),\text{ then }\nabla_{a}\Phi\in%
{\textstyle\bigwedge\nolimits^{k}}
\mathcal{V}^{\star}(U). \label{CDMMF11}%
\end{equation}

For all $f\in\mathcal{S}(U),$ $a,b\in\mathcal{V}(U)$ and $\Phi\in%
{\displaystyle\bigwedge}
\mathcal{V}^{\star}(U)$
\begin{align}
\nabla_{a+b}\Phi &  =\nabla_{a}\Phi+\nabla_{b}\Phi,\nonumber\\
\nabla_{fa}\Phi &  =f\nabla_{a}\Phi. \label{CDMMF12}%
\end{align}

For all $f\in\mathcal{S}(u),$ $a\in\mathcal{V}(U)$ and $\Phi,\Psi\in%
{\displaystyle\bigwedge}
\mathcal{V}^{\star}(U)$
\begin{align}
\nabla_{a}(\Phi+\Psi)  &  =\nabla_{a}\Phi+\nabla_{a}\Psi,\nonumber\\
\nabla_{a}(f\Phi)  &  =(af)\Phi+f\nabla_{a}\Phi. \label{CDMMF13}%
\end{align}

For all $a\in\mathcal{V}(U)$ and $\Phi,\Psi\in%
{\displaystyle\bigwedge}
\mathcal{V}^{\star}(U)$
\begin{equation}
\nabla_{a}(\Phi\wedge\Psi)=(\nabla_{a}\Phi)\wedge\Psi+\Phi\wedge\nabla_{a}%
\Psi. \label{CDMMF14}%
\end{equation}
\medskip

We now present three remarkable properties involving the action of $\nabla
_{a}$ on the duality products of multivector and multiform fields.

\begin{proposition}
When $\nabla_{a}$ acts on the duality scalar product of $\Phi\in%
{\displaystyle\bigwedge}
\mathcal{V}^{\star}(U)$ with $X\in%
{\displaystyle\bigwedge}
\mathcal{V}(U)$ follows the Leibniz rule, i.e.,%
\begin{equation}
a\left\langle \Phi,X\right\rangle =\left\langle \nabla_{a}\Phi,X\right\rangle
+\left\langle \Phi,\nabla_{a}X\right\rangle . \label{CDMMF15}%
\end{equation}

$\nabla_{a}$ acting on the duality left contracted product of $\Phi\in%
{\displaystyle\bigwedge}
\mathcal{V}^{\star}(U)$ with $X\in%
{\displaystyle\bigwedge}
\mathcal{V}(U)$ (or, $X\in%
{\displaystyle\bigwedge}
\mathcal{V}(U)$ with $\Phi\in%
{\displaystyle\bigwedge}
\mathcal{V}^{\star}(U)$) satisfies the Leibniz rule, i.e.,
\begin{align}
\nabla_{a}\left\langle \Phi,X\right\vert  &  =\left\langle \nabla_{a}%
\Phi,X\right\vert +\left\langle \Phi,\nabla_{a}X\right\vert ,\label{CDMMF16}\\
\nabla_{a}\left\langle X,\Phi\right\vert  &  =\left\langle \nabla_{a}%
X,\Phi\right\vert +\left\langle X,\nabla_{a}\Phi\right\vert . \label{CDMMF17}%
\end{align}

$\nabla_{a}$ acting on the duality right contracted product of $\Phi\in%
{\displaystyle\bigwedge}
\mathcal{V}^{\star}(U)$ with $X\in%
{\displaystyle\bigwedge}
\mathcal{V}(U)$ (or, $X\in%
{\displaystyle\bigwedge}
\mathcal{V}^{^{\star}}(U)$ with $\Phi\in%
{\displaystyle\bigwedge}
\mathcal{V}(U)$) satisfies the Leibniz rule, i.e.,%
\begin{align}
\nabla_{a}\left\vert \Phi,X\right\rangle  &  =\left\vert \nabla_{a}%
\Phi,X\right\rangle +\left\vert \Phi,\nabla_{a}X\right\rangle ,
\label{CDMMF18}\\
\nabla_{a}\left\vert X,\Phi\right\rangle  &  =\left\vert \nabla_{a}%
X,\Phi\right\rangle +\left\vert X,\nabla_{a}\Phi\right\rangle .
\label{CDMMF19}%
\end{align}

\end{proposition}

\begin{proof}
We will prove only the statement given by Eq.(\ref{CDMMF16}). Take $\Phi\in%
{\displaystyle\bigwedge}
\mathcal{V}^{\star}(U),$ $X\in%
{\displaystyle\bigwedge}
\mathcal{V}(U)$ and $\Psi\in%
{\displaystyle\bigwedge}
\mathcal{V}^{\star}(U)$. By a property of the duality left contracted product,
we have%
\begin{equation}
\left\langle \left\langle \Phi,X\right\vert ,\Psi\right\rangle =\left\langle
X,\widetilde{\Phi}\wedge\Psi\right\rangle . \tag{a}%
\end{equation}

Now, by using Eq.(\ref{CDMMF15}) and Eq.(\ref{CDMMF14}), we can write%
\begin{align}
&  \left\langle \nabla_{a}\left\langle \Phi,X\right\vert ,\Psi\right\rangle
+\left\langle \left\langle \Phi,X\right\vert ,\nabla_{a}\Psi\right\rangle
\nonumber\\
&  =\left\langle \nabla_{a}X,\widetilde{\Phi}\wedge\Psi\right\rangle
+\left\langle X,\nabla_{a}(\widetilde{\Phi}\wedge\Psi)\right\rangle
\nonumber\\
&  =\left\langle \nabla_{a}X,\widetilde{\Phi}\wedge\Psi\right\rangle
+\left\langle X,(\nabla_{a}\widetilde{\Phi})\wedge\Psi)\right\rangle
+\left\langle X,\widetilde{\Phi}\wedge\nabla_{a}\Psi)\right\rangle . \tag{b}%
\end{align}
Thus, taking into account Eq.(\ref{CDMMF11}) and by recalling once again a
property of the duality left contracted product, it follows that%
\begin{equation}
\left\langle \nabla_{a}\left\langle \Phi,X\right\vert ,\Psi\right\rangle
=\left\langle \left\langle \Phi,\nabla_{a}X\right\vert ,\Psi\right\rangle
+\left\langle \left\langle \nabla_{a}\Phi,X\right\vert ,\Psi)\right\rangle .
\tag{c}%
\end{equation}
Then, by the non-degeneracy of the duality scalar product, the required result follows.
\end{proof}

\subsection{Covariant Derivative of Extensor Fields}

\begin{definition}
Let $\left\langle U,\Gamma\right\rangle $ be a parallelism structure on $U,$
and let us take $a\in\mathcal{V}(U)$. The $a$\emph{-DCD,} associated with
$\left\langle U,\Gamma\right\rangle $,\emph{\ }of a smooth
\textit{multivector} extensor field on $U$ or a smooth \textit{multiform}
extensor field on $U$ are the mappings
\[
\left.  \overset{\left.  {}\right.  }{ext}\right.  _{k}^{l}\mathcal{V}%
(U)\ni\tau\mapsto\nabla_{a}\tau\in\left.  \overset{\left.  {}\right.  }%
{ext}\right.  _{k}^{l}\mathcal{V}(U),
\]
and
\[
\left.  \overset{\ast}{ext}\right.  _{k}^{l}\mathcal{V}(U)\ni\tau\mapsto
\nabla_{a}\tau\in\left.  \overset{\ast}{ext}\right.  _{k}^{l}\mathcal{V}(U),
\]
such that for all $X_{1}\in%
{\displaystyle\bigwedge\nolimits_{1}^{\Diamond}}
\mathcal{V}(U),\ldots,X_{k}\in%
{\displaystyle\bigwedge\nolimits_{k}^{\Diamond}}
\mathcal{V}(U),$ and for all $\Phi^{1}\in%
{\displaystyle\bigwedge\nolimits_{1}^{\Diamond}}
\mathcal{V}^{\ast}(U),\ldots,$ $\Phi^{l}\in%
{\displaystyle\bigwedge\nolimits_{l}^{\Diamond}}
\mathcal{V}^{\ast}(U)$ we have%
\begin{align}
\nabla_{a}\tau(X_{1},...,X_{k},\Phi^{1},...,\Phi^{l})  &  =\nabla_{a}%
(\tau(X_{1},...,X_{k},\Phi^{1},...,\Phi^{l}))\nonumber\\
&  -\tau(\nabla_{a}X_{1},...,X_{k},\Phi^{1},...,\Phi^{l})-\cdots\nonumber\\
&  -\tau(X_{1},...,\nabla_{a}X_{k},\Phi^{1},...,\Phi^{l})\nonumber\\
&  -\tau(X_{1},...,X_{k},\nabla_{a}\Phi^{1},...,\Phi^{l})-\cdots\nonumber\\
&  -\tau(X_{1},...,X_{k},\Phi^{1},...,\nabla_{a}\Phi^{l}). \label{CDEF1}%
\end{align}

\end{definition}

The covariant derivative of smooth multivector (or multiform) extensor fields
has two basic properties:

(i) For $f\in\mathcal{S}(U),$ and $a,b\in\mathcal{V}(U),\mathcal{\ }$and
$\tau\in\left.  \overset{\left.  {}\right.  }{ext}\right.  _{k}^{l}%
\mathcal{V}(U)$ (or $\tau\in\left.  \overset{\ast}{ext}\right.  _{k}%
^{l}\mathcal{V}(U)$)%
\begin{align}
\nabla_{a+b}\tau &  =\nabla_{a}\tau+\nabla_{b}\tau\label{CDEF2a}\\
\nabla_{fa}\tau &  =f\nabla_{a}\tau. \label{CDEF2b}%
\end{align}

(ii) For $f\in\mathcal{S}(U),$ and $a\in\mathcal{V}(U),\mathcal{\ }$and
$\tau,\sigma\in\left.  \overset{\left.  {}\right.  }{ext}\right.  _{k}%
^{l}\mathcal{V}(U)$ (or $\tau,\sigma\in\left.  \overset{\ast}{ext}\right.
_{k}^{l}\mathcal{V}(U)$)%
\begin{align}
\nabla_{a}(\tau+\sigma)  &  =\nabla_{a}\tau+\nabla_{a}\sigma,\label{CDEF3a}\\
\nabla_{a}(f\tau)  &  =(af)\tau+f\nabla_{a}\tau. \label{CDEF3b}%
\end{align}

The covariant differentiation of the exterior product of smooth multivector
(or multiform) extensor fields satisfies the Leibniz's rule. We have

\begin{proposition}
For all $\tau\in\left.  \overset{\left.  {}\right.  }{ext}\right.  _{k}%
^{l}\mathcal{V}(U)$ and $\sigma\in\left.  \overset{\left.  {}\right.  }%
{ext}\right.  _{r}^{s}\mathcal{V}(U)$ (or, $\tau\in\left.  \overset{\ast}%
{ext}\right.  _{k}^{l}\mathcal{V}(U)$ and $\sigma\in\left.  \overset{\ast
}{ext}\right.  _{r}^{s}\mathcal{V}(U)$), it holds
\begin{equation}
\nabla_{a}(\tau\wedge\sigma)=(\nabla_{a}\tau)\wedge\sigma+\tau\wedge\left(
\nabla_{a}\sigma\right)  . \label{CDEF4}%
\end{equation}

\end{proposition}

\begin{proof}
Without loss of generality, we prove this statement only for multivector
extensor fields $(X,\Phi)\mapsto\tau(X,\Phi)$ and $(Y,\Psi)\mapsto
\sigma(Y,\Psi).$ Using Eq.(\ref{CDEF1}), we can write
\begin{align*}
&  \nabla_{a}(\tau\wedge\sigma)(X,Y,\Phi,\Psi)\\
&  =\nabla_{a}((\tau\wedge\sigma)(X,Y,\Phi,\Psi))\\
&  -(\tau\wedge\sigma)(\nabla_{a}X,Y,\Phi,\Psi)-(\tau\wedge\sigma
)(X,\nabla_{a}Y,\Phi,\Psi)\\
&  -(\tau\wedge\sigma)(X,Y,\nabla_{a}\Phi,\Psi)-(\tau\wedge\sigma
)(X,Y,\Phi,\nabla_{a}\Psi).
\end{align*}

Using Eq.(\ref{EF5}) and recalling Leibniz's rule for the covariant
differentiation of the exterior product of multivector fields, we have
\begin{align*}
&  \nabla_{a}(\tau\wedge\sigma)(X,Y,\Phi,\Psi)\\
&  =\nabla_{a}(\tau(X,\Phi))\wedge\sigma(Y,\Psi)+\tau(X,\Phi)\wedge\nabla
_{a}(\sigma(Y,\Psi))\\
&  -\tau(\nabla_{a}X,\Phi)\wedge\sigma(Y,\Psi)-\tau(X,\Phi)\wedge\sigma
(\nabla_{a}Y,\Psi)\\
&  -\tau(X,\nabla_{a}\Phi)\wedge\sigma(Y,\Psi)-\tau(X,\Phi)\wedge
\sigma(Y,\nabla_{a}\Psi),
\end{align*}
i.e.,
\begin{align*}
&  \nabla_{a}(\tau\wedge\sigma)(X,Y,\Phi,\Psi)\\
&  =(\nabla_{a}(\tau(X,\Phi))-\tau(\nabla_{a}X,\Phi)-\tau(X,\nabla_{a}%
\Phi))\wedge\sigma(Y,\Psi)\\
&  +\tau(X,\Phi)\wedge(\nabla_{a}(\sigma(Y,\Psi))-\sigma(\nabla_{a}%
Y,\Psi)-\sigma(Y,\nabla_{a}\Psi)).
\end{align*}
Then, using once again Eq.(\ref{CDEF1}) and Eq.(\ref{EF5}), the expected
result follows.
\end{proof}

The covariant differentiation of the duality scalar product and each one of
the duality contracted products of smooth extensor fields satisfies the
Leibniz's rule.

\begin{proposition}
For all $\tau\in\left.  \overset{\ast}{ext}\right.  _{k}^{l}\mathcal{V}(U)$
and $\sigma\in\left.  \overset{\left.  {}\right.  }{ext}\right.  _{r}%
^{s}\mathcal{V}(U)$ (or, $\tau\in\left.  \overset{\left.  {}\right.  }%
{ext}\right.  _{k}^{l}\mathcal{V}(U)$ and $\tau\in\left.  \overset{\ast}%
{ext}\right.  _{r}^{s}\mathcal{V}(U)$), we have that
\begin{equation}
\nabla_{a}\left\langle \tau,\sigma\right\rangle =\left\langle \nabla_{a}%
\tau,\sigma\right\rangle +\left\langle \tau,\nabla_{a}\sigma\right\rangle .
\label{CDEF5}%
\end{equation}

For all $\tau\in\left.  \overset{\ast}{ext}\right.  _{k}^{l}\mathcal{V}(U)$
and $\sigma\in\left.  \overset{\left.  {}\right.  }{ext}\right.  _{r}%
^{s}\mathcal{V}(U)$ (or, $\tau\in\left.  \overset{\left.  {}\right.  }%
{ext}\right.  _{k}^{l}\mathcal{V}(U)$ and $\tau\in\left.  \overset{\ast}%
{ext}\right.  _{r}^{s}\mathcal{V}(U)$), it holds
\begin{align}
\nabla_{a}\left\langle \tau,\sigma\right\vert  &  =\left\langle \nabla_{a}%
\tau,\sigma\right\vert +\left\langle \tau,\nabla_{a}\sigma\right\vert
,\label{CDEF6a}\\
\nabla_{a}\left\langle \sigma,\tau\right\vert  &  =\left\langle \nabla
_{a}\sigma,\tau\right\vert +\left\langle \sigma,\nabla_{a}\tau\right\vert .
\label{CDEF6b}%
\end{align}

For all $\tau\in\left.  \overset{\ast}{ext}\right.  _{k}^{l}\mathcal{V}(U)$
and $\sigma\in\left.  \overset{\left.  {}\right.  }{ext}\right.  _{r}%
^{s}\mathcal{V}(U)$ (or, $\tau\in\left.  \overset{\left.  {}\right.  }%
{ext}\right.  _{k}^{l}\mathcal{V}(U)$ and $\tau\in\left.  \overset{\ast}%
{ext}\right.  _{r}^{s}\mathcal{V}(U)$), it holds
\begin{align}
\nabla_{a}\left\vert \tau,\sigma\right\rangle  &  =\left\vert \nabla_{a}%
\tau,\sigma\right\rangle +\left\vert \tau,\nabla_{a}\sigma\right\rangle
,\label{CDEF7a}\\
\nabla_{a}\left\vert \tau,\sigma\right\rangle  &  =\left\vert \nabla_{a}%
\tau,\sigma\right\rangle +\left\vert \tau,\nabla_{a}\sigma\right\rangle .
\label{CDEF7b}%
\end{align}

\end{proposition}

\begin{proof}
We present only the proof of the property given by Eq.(\ref{CDEF6a}). Without
loss of generality, we prove this statement only for a multiform extensor
field $\tau$ and a multivector extensor field $\sigma$ such that
$(X,\Phi)\mapsto\tau(X,\Phi)$ and $(Y,\Psi)\mapsto\sigma(Y,\Psi)$.%
\[%
\begin{array}
[c]{ll}%
\nabla_{a}\left\langle \tau,\sigma\left\vert \left(  X,Y,\Phi,\Psi\right)
\right.  \right.  & =\nabla_{a}\left(  \left\langle \tau,\sigma\left\vert
\left(  X,Y,\Phi,\Psi\right)  \right.  \right.  \right)  -\left\langle
\tau,\sigma\left\vert \left(  \nabla_{a}X,Y,\Phi,\Psi\right)  \right.  \right.
\\
& -\left\langle \tau,\sigma\left\vert \left(  X,\nabla_{a}Y,\Phi,\Psi\right)
\right.  \right.  -\left\langle \tau,\sigma\left\vert \left(  X,Y,\nabla
_{a}\Phi,\Psi\right)  \right.  \right. \\
& -\left\langle \tau,\sigma\left\vert \left(  X,Y,\Phi,\nabla_{a}\Psi\right)
\right.  \right.
\end{array}
\]
or recalling that $\left\langle \tau,\sigma\left\vert _{\left(  p\right)
}\right.  \right.  =\left\langle \tau_{\left(  p\right)  },\sigma_{\left(
p\right)  }\left\vert {}\right.  \right.  $,%
\begin{equation}%
\begin{array}
[c]{ll}%
\nabla_{a}\left\langle \tau,\sigma\left\vert \left(  X,Y,\Phi,\Psi\right)
\right.  \right.  & =\nabla_{a}\left(  \left\langle \tau\left(  X,\Phi\right)
,\sigma\left(  Y,\Psi\right)  \left\vert {}\right.  \right.  \right)
-\left\langle \tau\left(  \nabla_{a}X,\Phi\right)  ,\sigma\left(
Y,\Psi\right)  \left\vert {}\right.  \right. \\
& -\left\langle \tau\left(  X,\Phi\right)  ,\sigma\left(  \nabla_{a}%
Y,\Psi\right)  \left\vert {}\right.  \right.  -\left\langle \tau\left(
X,\nabla_{a}\Phi\right)  ,\sigma\left(  Y,\Psi\right)  \left\vert {}\right.
\right. \\
& -\left\langle \tau\left(  X,\Phi\right)  ,\sigma\left(  Y,\nabla_{a}%
\Psi\right)  \left\vert {}\right.  \right.  .
\end{array}
\label{ENC1}%
\end{equation}
On the other hand, from Eq.(\ref{CDMMF18}), we can write
\[
\nabla_{a}\left(  \left\langle \tau\left(  X,\Phi\right)  ,\sigma\left(
Y,\Psi\right)  \left\vert {}\right.  \right.  \right)  =\left\langle
\nabla_{a}\tau\left(  X,\Phi\right)  ,\sigma\left(  Y,\Psi\right)  \left\vert
{}\right.  \right.  +\left\langle \tau\left(  X,\Phi\right)  ,\nabla_{a}%
\sigma\left(  Y,\Psi\right)  \left\vert {}\right.  \right.  ,
\]
Eq.(\ref{ENC1}) can by written as%
\[%
\begin{array}
[c]{l}%
\nabla_{a}\left\langle \tau,\sigma\left\vert \left(  X,Y,\Phi,\Psi\right)
\right.  \right. \\
=\left\langle \nabla_{a}\tau\left(  X,\Phi\right)  -\tau\left(  \nabla
_{a}X,\Phi\right)  -\tau\left(  X,\nabla_{a}\Phi\right)  ,\sigma\left(
Y,\Psi\right)  \left\vert {}\right.  \right. \\
+\left\langle \tau\left(  X,\Phi\right)  ,\nabla_{a}\sigma\left(
Y,\Psi\right)  -\sigma\left(  \nabla_{a}Y,\Psi\right)  -\sigma\left(
Y,\nabla_{a}\Psi\right)  \left\vert {}\right.  \right. \\
=\left\langle \left(  \nabla_{a}\tau\right)  \left(  X,\Phi\right)
,\sigma\left(  Y,\Psi\right)  \left\vert {}\right.  \right.  +\left\langle
\tau\left(  X,\Phi\right)  ,\left(  \nabla_{a}\sigma\right)  \left(
Y,\Psi\right)  \left\vert {}\right.  \right. \\
=\left\langle \nabla_{a}\tau,\sigma\left\vert \left(  X,Y,\Phi,\Psi\right)
\right.  \right.  +\left\langle \tau,\nabla_{a}\sigma\left\vert \left(
X,Y,\Phi,\Psi\right)  \right.  \right.  ,
\end{array}
\]
which proves our result.
\end{proof}

Finally we prove that the duality adjoint operator commutes with the
$a$-\textit{DCDO}, i.e., we have

\begin{proposition}
If $\tau$ is any one of the four smooth one-variable extensor fields on $U$,
then
\begin{equation}
(\nabla_{a}\tau)^{\bigtriangleup}=\nabla_{a}\tau^{\bigtriangleup}.
\label{CDEF8}%
\end{equation}

\end{proposition}

\begin{proof}
Without loss of generality, we prove this statement only for $\tau\in ext(%
{\displaystyle\bigwedge\nolimits_{1}^{\diamond}}
\mathcal{V}(U),%
{\displaystyle\bigwedge\nolimits^{\diamond}}
\mathcal{V}(U)).$

Let us $X\in%
{\displaystyle\bigwedge\nolimits_{1}^{\diamond}}
\mathcal{V}(U)$ and $\Phi\in%
{\displaystyle\bigwedge\nolimits^{\diamond}}
\mathcal{V}^{\ast}(U).$ We must prove that
\[
\left\langle \nabla_{a}\tau^{\bigtriangleup}(\Phi),X\right\rangle
=\left\langle \Phi,\nabla_{a}\tau(X)\right\rangle .
\]

By using Eq.(\ref{CDEF1}) and recalling the Leibniz's rule for the covariant
differentiation of the duality scalar product of multiform fields with
multivector fields, we can write
\[
\left\langle \nabla_{a}\tau^{\bigtriangleup}(\Phi),X\right\rangle
=\left\langle \nabla_{a}(\tau^{\bigtriangleup}(\Phi)),X\right\rangle
-\left\langle \tau^{\bigtriangleup}(\nabla_{a}(\Phi)),X\right\rangle ,
\]
and from Eq.(\ref{CDMMF15})%
\[
\left\langle \nabla_{a}\tau^{\bigtriangleup}(\Phi),X\right\rangle
=a\left\langle \tau^{\bigtriangleup}(\Phi),X\right\rangle -\left\langle
\tau^{\bigtriangleup}(\Phi),\nabla_{a}X\right\rangle -\left\langle
\tau^{\bigtriangleup}(\nabla_{a}\Phi),X\right\rangle .
\]

Recalling now the fundamental property of the duality \textit{adjoint}, we
get
\[
\left\langle \nabla_{a}\tau^{\bigtriangleup}(\Phi),X\right\rangle
=a\left\langle \Phi,\tau(X)\right\rangle -\left\langle \Phi,\tau(\nabla
_{a}X)\right\rangle -\left\langle \nabla_{a}\Phi,\tau(X)\right\rangle .
\]
Using once again the Leibniz's rule, Eq.(\ref{CDMMF15}), for the covariant
differentiation of the duality scalar product we get%
\[
\left\langle \nabla_{a}\tau^{\bigtriangleup}(\Phi),X\right\rangle
=\left\langle \Phi,\nabla_{a}(\tau(X))\right\rangle -\left\langle \Phi
,\tau(\nabla_{a}X)\right\rangle ,
\]
and thus, using once again Eq.(\ref{CDEF1}), the required result follows.
\end{proof}

In particular, the $a$\emph{-DCD,} \emph{of a smooth }$k$\emph{-covariant and
}$l$\emph{-contravariant vector }(\emph{or form})\emph{\ extensor field} on
$U,$ and associated with $\left\langle U,\Gamma\right\rangle ,$ are defined
by
\[
\left.  \overset{}{ext}\right.  _{k}^{l}\mathcal{V}(U)\ni\tau\mapsto\nabla
_{a}\tau\in\left.  \overset{}{ext}\right.  _{k}^{l}\mathcal{V}(U)\text{ and
}\left.  \overset{\ast}{ext}\right.  _{k}^{l}\mathcal{V}(U)\ni\tau
\mapsto\nabla_{a}\tau\in\left.  \overset{\ast}{ext}\right.  _{k}%
^{l}\mathcal{V}(U),
\]
respectively, such that for every $v_{1},\ldots,v_{k}\in\mathcal{V}(U),$ and
$\omega^{1},\ldots,\omega^{l}\in\mathcal{V}^{\ast}(U)$:%
\begin{align}
\nabla_{a}\tau(v_{1},\ldots,v_{k},\omega^{1},\ldots,\omega^{l})  &
=\nabla_{a}(\tau(v_{1},\ldots,v_{k},\omega^{1},\ldots,\omega^{l}))\nonumber\\
&  -\tau(\nabla_{a}v_{1},\ldots,v_{k},\omega^{1},\ldots,\omega^{l}%
)-\cdots\nonumber\\
&  -\tau(v_{1},\ldots,\nabla_{a}v_{k},\omega^{1},\ldots,\omega^{l})\nonumber\\
&  -\tau(v_{1},\ldots,v_{k},\nabla_{a}\omega^{1},\ldots,\omega^{l}%
)-\cdots\nonumber\\
&  -\tau(v_{1},\ldots,v_{k},\omega^{1},\ldots,\nabla_{a}\omega^{l}).
\label{EEF8}%
\end{align}

Where we can easily see that it meets the basic properties (\ref{CDEF2a}),
(\ref{CDEF2b}), (\ref{CDEF3a}) and (\ref{CDEF3b}).

It is also worth recalling here that the $a$\emph{-DCD\ }of a smooth vector
field on $U$, associated with $\left\langle U,\Gamma\right\rangle $, is the
mapping
\[
\mathcal{V}(U)\ni v\mapsto\nabla_{a}v\in\mathcal{V}(U)
\]
such that
\begin{equation}
\nabla_{a}v=\Gamma(a,v), \label{CDV1}%
\end{equation}
and the $a$\emph{-DCD }of a smooth form field on $U,$ is the mapping
\[
V^{\ast}(U)\ni\omega\mapsto\nabla_{a}\omega\in\mathcal{V}^{\ast}(U)
\]
such that for every $v\in\mathcal{V}(U)$
\begin{equation}
\nabla_{a}\omega(v)=a\omega(v)-\omega(\nabla_{a}v). \label{CDF1}%
\end{equation}

\subsection{Deformed Covariant Derivative}

Let $\left\langle U,\Gamma\right\rangle $ be a parallelism structure on $U.$
Let us take an invertible smooth extensor operator field $\lambda$ on
$V\supseteq U,$ i.e., $\lambda:\mathcal{V}(U)\rightarrow\mathcal{V}(U).$ The
$\lambda$\emph{-deformation} of $\left\langle U,\Gamma\right\rangle $ is a
well-defined connection on $U,$ namely $\overset{\lambda}{\Gamma}$, given by
\[
\mathcal{V}(U)\times\mathcal{V}(U)\ni(a,v)\mapsto\overset{\lambda}{\Gamma
}(a,v)\in\mathcal{V}(U)
\]
such that
\begin{equation}
\overset{\lambda}{\Gamma}(a,v)=\lambda(\Gamma(a,\lambda^{-1}(v))).
\label{DPS1}%
\end{equation}
It is easy to see that $\overset{\lambda}{\Gamma}$ is indeed a connection on
$U$ since it satisfies Eq.(\ref{PS1}) and Eq.(\ref{PS2}).

The parallelism structure $\left\langle U,\overset{\lambda}{\Gamma
}\right\rangle $ is said to be the $\lambda$\emph{-deformation} of
$\left\langle U,\Gamma\right\rangle $.

Let us take $a\in\mathcal{V}(U).$ The $a$-\emph{DCDO} associated with
$\left\langle U,\overset{\lambda}{\Gamma}\right\rangle ,$ namely
$\overset{\lambda}{\nabla}_{a}$, has the basic properties:

\begin{proposition}
For all $v\in\mathcal{V}(U)$%
\begin{equation}
\overset{\lambda}{\nabla}_{a}v=\lambda(\nabla_{a}\lambda^{-1}(v)).
\label{DPS2}%
\end{equation}
It follows from Eq.(\ref{CDV1}) and Eq.(\ref{DPS1}).

For all $\omega\in\mathcal{V}^{\ast}(U)$%
\begin{equation}
\overset{\lambda}{\nabla}_{a}\omega=\lambda^{-\bigtriangleup}(\nabla
_{a}\lambda^{\bigtriangleup}(\omega)). \label{DPS4}%
\end{equation}

\end{proposition}

\begin{proof}
Eq.(\ref{DPS2} follows from Eq.(\ref{CDV1}) and Eq.(\ref{DPS1})).\ To prove
Eq.(\ref{DPS4}) take $v\in\mathcal{V}(U)$. Then using Eq.(\ref{CDF1}) and
Eq.(\ref{DPS2}), we have%
\begin{align}
\overset{\lambda}{\nabla}_{a}\omega(v)  &  =a\omega(v)-\omega(\overset
{\lambda}{\nabla}_{a}v)\nonumber\\
&  =a\omega(v)-\omega(\lambda(\nabla_{a}\lambda^{-1}(v)))\nonumber\\
&  =a\left\langle \omega,v\right\rangle -\left\langle \omega,\lambda
(\nabla_{a}\lambda^{-1}(v))\right\rangle , \label{a}%
\end{align}
but, by recalling the fundamental property of the \emph{duality adjoint} and
by using once again Eq.(\ref{CDF1}), the second term in Eq.(\ref{a}) can be
written
\begin{align}
\left\langle \omega,\lambda(\nabla_{a}\lambda^{-1}(v))\right\rangle  &
=\left\langle \lambda^{\bigtriangleup}(\omega),\nabla_{a}\lambda
^{-1}(v)\right\rangle \nonumber\\
&  =a\left\langle \lambda^{\bigtriangleup}(\omega),\lambda^{-1}%
(v)\right\rangle -\left\langle \nabla_{a}\lambda^{\bigtriangleup}%
(\omega),\lambda^{-1}(v)\right\rangle \nonumber\\
&  =a\left\langle \omega,v\right\rangle -\left\langle \lambda^{-\bigtriangleup
}\nabla_{a}\lambda^{\bigtriangleup}(\omega),v\right\rangle , \label{b}%
\end{align}
Finally, putting Eq.(\ref{b}) into Eq.(\ref{a}), the expected result follows.
\end{proof}

\subsubsection{Deformed Covariant Derivative of Multivector and Multiform
Fields}

Recall that $\lambda^{\bigtriangleup}$ is the so-called \emph{duality adjoint}
of $\lambda,$ and $\lambda^{-\bigtriangleup}$ is a \emph{short notation} for
$(\lambda^{\bigtriangleup})^{-1}=(\lambda^{-1})^{\bigtriangleup}$. We give now
two properties of $\overset{\lambda}{\text{ }\nabla}_{a}$ which are
generalizations of the basic properties (\ref{DPS2}) and (\ref{DPS4}) above.

\begin{proposition}
For all $X\in%
{\displaystyle\bigwedge}
\mathcal{V}(U)$%
\begin{equation}
\overset{\lambda}{\nabla}_{a}X=\underline{\lambda}(\nabla_{a}\underline
{\lambda}^{-1}(X)), \label{DCD3}%
\end{equation}
where $\underline{\lambda}$ is the so-called \emph{exterior power extension}
of $\lambda,$ and $\underline{\lambda}^{-1}$ is a \emph{more simple notation}
for $(\underline{\lambda})^{-1}=\underline{(\lambda^{-1})}.$

For all $\Phi\in%
{\displaystyle\bigwedge}
\mathcal{V}^{\star}(U)$%
\begin{equation}
\overset{\lambda}{\nabla}_{a}\Phi=\underline{\lambda}^{-\bigtriangleup}%
(\nabla_{a}\underline{\lambda}^{\bigtriangleup}(\Phi)), \label{DCD4}%
\end{equation}
where $\underline{\lambda}^{\bigtriangleup}=(\underline{\lambda}%
)^{\bigtriangleup}=\underline{(\lambda^{\bigtriangleup})}$ and $\underline
{\lambda}^{-\bigtriangleup}=(\underline{\lambda}^{\bigtriangleup})^{-1}.$
\end{proposition}

\begin{proof}
To prove the property for smooth multivector fields as given by Eq.(\ref{DCD3}%
), we make use the following criterion:

(\textbf{i}) We check the statement for scalar fields $f\in\mathcal{S}(U)$
(use Eq.(\ref{CDMMF1})).

(\textbf{ii}) We check the statement for simple $k$-vector fields $v_{1}%
\wedge\cdots\wedge v_{k}\in%
{\textstyle\bigwedge\nolimits^{k}}
\mathcal{V}(U)$ (use mathematical induction).

(\textbf{iii}) We check the statement for\ a finite addition of simple
$k$-vector fields $X^{k}+\cdots Z^{k}\in%
{\textstyle\bigwedge\nolimits^{k}}
\mathcal{V}(U)$ (use Eq.(\ref{CDMMF6}) and recalling the linear operator
character for the extended of a linear operator)

(\textbf{iv}) Then it becomes easy to prove the statement for multivector
fields $X\in%
{\textstyle\bigwedge}
\mathcal{V}(U)$.
\end{proof}

\subsubsection{Deformed Covariant Derivative of Extensor Fields}

We present now two properties for $\overset{\lambda}{\nabla}_{a}$ which are
generalizations of the properties (\ref{DCD3}) and (\ref{DCD4}).

\begin{proposition}
For all $\tau\in\left.  \overset{\left.  {}\right.  }{ext}\right.  _{k}%
^{l}\mathcal{V}(U):$%
\begin{equation}
\overset{\lambda}{\nabla}_{a}\tau=\underline{\lambda}\nabla_{a}\underline
{\lambda}^{-1}\tau, \label{DCD3a}%
\end{equation}
where $\underline{\lambda}^{-1}\tau$ means the action of $\underline{\lambda
}^{-1}$ on the smooth multivector extensor field $\tau,$ and $\underline
{\lambda}\nabla_{a}\underline{\lambda}^{-1}\tau$ is the action of
$\underline{\lambda}$ on the smooth multivector extensor field $\nabla
_{a}\underline{\lambda}^{-1}\tau$.
\end{proposition}

\begin{proposition}
For all $\upsilon\in\left.  \overset{\ast}{ext}\right.  _{k}^{l}%
\mathcal{V}(U):$%
\begin{equation}
\overset{\lambda}{\nabla}_{a}\upsilon=\underline{\lambda}^{-\bigtriangleup
}\nabla_{a}\underline{\lambda}^{\bigtriangleup}\upsilon, \label{DCD3b}%
\end{equation}
where $\underline{\lambda}^{\bigtriangleup}\upsilon$ means the action of
$\underline{\lambda}^{\bigtriangleup}$ on the smooth multiform extensor field
$\upsilon,$ and $\underline{\lambda}^{-\bigtriangleup}\nabla_{a}%
\underline{\lambda}^{\bigtriangleup}\upsilon$ is the action of $\underline
{\lambda}^{-\bigtriangleup}$ on the smooth multiform extensor field
$\nabla_{a}\underline{\lambda}^{\bigtriangleup}\upsilon.$
\end{proposition}

\begin{proof}
We will prove only the property for smooth multivector extensor fields as
given by Eq.(\ref{DCD3a}). Without restrictions on generality, we can work
with a multivector extensor field $(X,\Phi)\mapsto\tau(X,\Phi)$.

Then, using Eq.(\ref{CDEF1}) and taking into account the properties
(\ref{DCD3}) and (\ref{DCD4}), we can write%
\begin{align*}
&  (\overset{\lambda}{\nabla}_{a}\tau)(X,\Phi)\\
&  =\overset{\lambda}{\nabla}_{a}\tau(X,\Phi)-\tau(\overset{\lambda}{\nabla
}_{a}X,\Phi)-\tau(X,\overset{\lambda}{\nabla}_{a}\Phi)\\
&  =\underline{\lambda}(\nabla_{a}\underline{\lambda}^{-1}\circ\tau
(X,\Phi))-\tau(\underline{\lambda}(\nabla_{a}\underline{\lambda}^{-1}%
(X)),\Phi)-\tau(X,\underline{\lambda}^{-\bigtriangleup}(\nabla_{a}%
\underline{\lambda}^{\bigtriangleup}(\Phi))),
\end{align*}
i.e.,%
\begin{align*}
&  (\underline{\lambda}^{-1}\circ\overset{\lambda}{\nabla}_{a}\tau)(X,\Phi)\\
&  =\nabla_{a}\underline{\lambda}^{-1}\circ\tau(X,\Phi)-\underline{\lambda
}^{-1}\circ\tau(\underline{\lambda}(\nabla_{a}\underline{\lambda}%
^{-1}(X)),\Phi)\\
&  -\underline{\lambda}^{-1}\circ\tau(X,\underline{\lambda}^{-\bigtriangleup
}(\nabla_{a}\underline{\lambda}^{\bigtriangleup}(\Phi))).
\end{align*}

By recalling the action of an extended operator $\underline{\lambda}^{-1}$ on
a multivector extensor $\tau,$ we get%
\begin{align*}
&  (\underline{\lambda}^{-1}\circ\overset{\lambda}{\nabla}_{a}\tau)(X,\Phi)\\
&  =\nabla_{a}\underline{\lambda}^{-1}\tau(\underline{\lambda}^{-1}%
(X),\underline{\lambda}^{\bigtriangleup}(\Phi))-\underline{\lambda}^{-1}%
\tau(\nabla_{a}\underline{\lambda}^{-1}(X)),\underline{\lambda}%
^{\bigtriangleup}(\Phi))\\
&  -\underline{\lambda}^{-1}\tau(\underline{\lambda}^{-1}(X)),\nabla
_{a}\underline{\lambda}^{\bigtriangleup}(\Phi)).
\end{align*}

Using once again Eq.(\ref{CDEF1}), we have%
\[
(\overset{\lambda}{\nabla}_{a}\tau)(X,\Phi)=\underline{\lambda}\circ
(\nabla_{a}\underline{\lambda}^{-1}\tau)(\underline{\lambda}^{-1}%
(X),\underline{\lambda}^{\bigtriangleup}(\Phi)),
\]
and finally recalling once again the action of an exterior power extension
operator $\underline{\lambda}$ on a multivector extensor $\nabla_{a}%
\underline{\lambda}^{-1}\tau$, the required result follows.
\end{proof}

\subsection{Relative Covariant Derivative}

\begin{definition}
Let $\left\{  b_{\mu},\beta^{\mu}\right\}  $ be a pair of dual frame fields
for $U\subseteq M$. Associated with $\left\{  b_{\mu},\beta^{\mu}\right\}  $
we can construct a well-defined connection on\emph{\ }$U$ given by the
mapping
\[
B:\mathcal{V}(U)\times\mathcal{V}(U)\longrightarrow\mathcal{V}(U),
\]
such that
\begin{equation}
B(a,v)=\left[  a\beta^{\sigma}(v)\right]  b_{\sigma}. \label{RPS1}%
\end{equation}
$B$ is called \emph{relative connection }on $U$ with respect to\emph{\ }%
$\left\{  b_{\mu},\beta^{\mu}\right\}  $ \emph{(}or simply relative connection
for short\emph{)}.
\end{definition}

The parallelism structure $\left\langle U,B\right\rangle $ will be called
\emph{relative parallelism structure} with respect to $\left\{  b_{\mu}%
,\beta^{\mu}\right\}  .$

The \emph{\ }$a$\emph{-DCDO }induced by the relative connection will be
denoted by $\partial_{a}$. According with Eq.(\ref{CDV1}), the $a$\emph{-DCD}
of a smooth vector field is given by
\begin{equation}
\partial_{a}v=\left[  a\beta^{\sigma}(v)\right]  b_{\sigma}. \label{RPS2}%
\end{equation}

Note that $\partial_{a}$ is the unique $a$\emph{-DCDO }which satisfies the
condition
\begin{equation}
\partial_{a}b_{\mu}=0. \label{RPS3}%
\end{equation}

On the other hand, given a parallelism structure $\left\langle U_{0}%
,\Gamma\right\rangle $ and any relative parallelism structure $\left\langle
U,B\right\rangle ,$ such that $U_{0}\cap U\neq\phi$. There exists a smooth
$2$\emph{-covariant vector extensor field} on $U_{0}\cap U$, namely $\gamma$,
defined by
\[
\mathcal{V}(U_{0}\cap U)\times\mathcal{V}(U_{0}\cap U)\ni(a,v)\mapsto
\gamma(a,v)\in\mathcal{V}(U_{0}\cap U)
\]
such that
\begin{equation}
\gamma(a,v)=\beta^{\mu}(v)\nabla_{a}b_{\mu} \label{STh1}%
\end{equation}
which satisfies
\begin{equation}
\Gamma(a,v)=B(a,v)+\gamma(a,v). \label{STh2}%
\end{equation}

Such a extensor field $\gamma$ will be called the \emph{relative connection
extensor field}\footnote{The properties of the tensor field $\gamma_{\mu\nu
}^{\alpha}$ such that $\gamma(\partial_{\mu},\partial_{\nu})=\gamma_{\mu\nu
}^{\alpha}\partial_{\alpha}$ where $\{\partial_{\mu}\}$ is a basis for
$\mathcal{V}(U_{0}\cap U)$ are studied in details in \cite{rodoliv2006}.} on
$U_{0}\cap U$.

From Eq.(\ref{CDV1}),this means that for all $v\in\mathcal{V}(U_{0}\cap U):$
\begin{equation}
\nabla_{a}v=\partial_{a}v+\gamma_{a}(v), \label{STh3}%
\end{equation}
where $\nabla_{a}$ is the $a$-\emph{DCDO} associated with $\left\langle
U_{0},\Gamma\right\rangle $ and $\partial_{a}$ is the $a$-\emph{DCDO}
associated with $\left\langle U,B\right\rangle $ (note that $\gamma_{a}$ is a
smooth \emph{vector operator field }on $U_{0}\cap U$ defined by $\gamma
_{a}(v)=\gamma(a,v)$).

By using Eq.(\ref{CDF1}) and Eq.(\ref{STh3}), we get that for all $\omega
\in\mathcal{V}^{\ast}(U_{0}\cap U):$
\begin{equation}
\nabla_{a}\omega=\partial_{a}\omega-\gamma_{a}^{\bigtriangleup}(\omega),
\label{STh4}%
\end{equation}
where $\gamma_{a}^{\bigtriangleup}$ is the \emph{dual adjoint} of $\gamma_{a}$
(i.e., $\left\langle \gamma_{a}^{\bigtriangleup}(\omega),v\right\rangle
=\left\langle \omega,\gamma_{a}(v)\right\rangle $).

\subsubsection{Relative Covariant Derivative of Multivector and Multiform
Fields}

Let $\left\langle U_{0},\Gamma\right\rangle $ be a parallelism structure on
$U_{0},$ and let $\nabla_{a}$ be its associated $a$-\textit{DCDO}. Take any
relative parallelism structure $\left\langle U,B\right\rangle $ compatible
with $\left\langle U_{0},\Gamma\right\rangle $ (i.e., $U_{0}\cap
U\neq\emptyset$)$.$

We present now the split theorem, i.e. a generalization of Eqs. (\ref{STh3})
and (\ref{STh4}), for smooth multivector fields and for multiform fields. We
know that for all $v\in\mathcal{V}(U_{0}\cap U)$, $\nabla_{a}v=\partial
_{a}v+\gamma_{a}(v)$, and for all $\omega\in\mathcal{V}^{\star}(U_{0}\cap U)$,
$\nabla_{a}\omega=\partial_{a}\omega-\gamma_{a}^{\bigtriangleup}(\omega).$

\begin{theorem}
\emph{(a)}\textbf{\ }For all $X\in%
{\textstyle\bigwedge}
\mathcal{V}(U_{0}\cap U)$%
\begin{equation}
\nabla_{a}X=\partial_{a}X+\underset{\smile}{\gamma_{a}}(X), \label{RCD3}%
\end{equation}
where $\underset{\smile}{\gamma_{a}}$ is \emph{contracted extension
operator}\footnote{See appendix 6.3 and \cite{fmcr1}, to recall the notion of
the generalization procedure.}\emph{\ of} $\gamma_{a}$

\emph{(b)} For all $\Phi\in%
{\textstyle\bigwedge}
\mathcal{V}^{\star}(U_{0}\cap U)$%
\begin{equation}
\nabla_{a}\Phi=\partial_{a}\Phi-\underset{\smile}{\gamma_{a}^{\bigtriangleup}%
}(\Phi), \label{RCD4}%
\end{equation}
where $\underset{\smile}{\gamma_{a}^{\bigtriangleup}}$ is the so-called
\emph{contracted extension operator} of $\gamma_{a}^{\bigtriangleup}$ which as
we know coincides with the so-called duality adjoint of $\underset{\smile
}{\gamma_{a}}.$
\end{theorem}

\begin{proof}
To prove the property for smooth multivector fields as given by Eq.(\ref{RCD3}%
), we use as above the following procedure:.

(\textbf{i}) We check the statement for scalar fields $f\in\mathcal{S}(U).$

(\textbf{ii}) Next, we check the statement for simple $k$-vector fields
$v_{1}\wedge\cdots\wedge v_{k}\in%
{\textstyle\bigwedge\nolimits^{k}}
\mathcal{V}(U).$

(\textbf{iii}) Next we check the statement for a finite addition of simple $k
$-vector fields $X^{k}+\cdots Z^{k}\in%
{\textstyle\bigwedge\nolimits^{k}}
\mathcal{V}(U). $

(\textbf{iv}) We then can easily prove the statement for multivector fields
$X\in%
{\textstyle\bigwedge}
\mathcal{V}(U)$.
\end{proof}

Let $\left\langle U,B\right\rangle $ and $\left\langle U^{\prime},B^{\prime
}\right\rangle $, $U\cap U^{\prime}\neq\emptyset$, be two compatible
parallelism structures taken on a smooth manifold $M$. The $a$-\textit{DCDO}'s
associated with $\left\langle U,B\right\rangle $ and $\left\langle U^{\prime
},B^{\prime}\right\rangle $ are denoted by $\partial_{a}$ and $\partial
_{a}^{\prime}$, respectively. As we already know \cite{fmcr2}, there exists a
well-defined smooth extensor operator field on $U\cap U^{\prime}$, the
\emph{Jacobian field} $J$ (see Appendix$\ $E), such that the following two
basic properties are satisfied: for all $v\in\mathcal{V}(U\cap U^{\prime})$,
$\partial_{a}^{\prime}v=J(\partial_{a}J^{-1}(v))$, and for all $\omega
\in\mathcal{V}^{\star}(U\cap U^{\prime})$, $\partial_{a}^{\prime}%
\omega=J^{-\bigtriangleup}(\partial_{a}J^{\bigtriangleup}(\omega))$.

We can see immediately that the basic properties just recalled implies that
$\partial_{a}^{\prime}$ is the $J$-deformation of $\partial_{a}.$

We present now two properties for the relative covariant derivatives which are
generalizations of the basic properties just recalled above.

\begin{proposition}
For all $X\in%
{\displaystyle\bigwedge}
\mathcal{V}(U\cap U^{\prime})$%
\begin{equation}
\partial_{a}^{\prime}X=\underline{J}(\partial_{a}\underline{J}^{-1}(X)).
\label{RCD7}%
\end{equation}

For all $\Phi\in%
{\displaystyle\bigwedge}
\mathcal{V}^{\star}(U\cap U^{\prime})$%
\begin{equation}
\partial_{a}^{\prime}\Phi=\underline{J}^{-\bigtriangleup}(\partial
_{a}\underline{J}^{\bigtriangleup}(\Phi)). \label{RCD8}%
\end{equation}

\begin{proof}
Eq.(\ref{RCD7}) is an immediate consequence of Eq.(\ref{DCD3}).,
Eq.(\ref{RCD8}) is an immediate consequence of Eq.(\ref{DCD4})\emph{.}
\end{proof}
\end{proposition}

\subsubsection{Relative Covariant Derivative of Extensor Fields}

As we know from (\ref{RCD3}) and (\ref{RCD4}),\ there exists a well-defined
smooth vector operator field on $U_{0}\cap U$, called the\emph{\ relative
connection field} $\gamma_{a}$, which satisfies the \emph{split theorem} valid
for smooth multivector fields and for smooth multiform fields, i.e.: for all
$X\in%
{\textstyle\bigwedge}
\mathcal{V}(U_{0}\cap U)$,$\nabla_{a}X=\partial_{a}X+\underset{\smile}%
{\gamma_{a}}(X)$, and for all $\Phi\in%
{\textstyle\bigwedge}
\mathcal{V}^{\star}(U_{0}\cap U)$, $\nabla_{a}\Phi=\partial_{a}\Phi
-\underset{\smile}{\gamma_{a}^{\bigtriangleup}}(\Phi)$.

We now present a split theorem for smooth multivector extensor fields and for
smooth multiform extensor fields, which are the generalizations of the
properties just recalled above.

\begin{theorem}
\ \emph{(i)} For all $\tau\in\left.  \overset{\left.  {}\right.  }%
{ext}\right.  _{k}^{l}\mathcal{V}(U_{0}\cap U):$%
\begin{equation}
\nabla_{a}\tau=\partial_{a}\tau+\underset{\smile}{\gamma_{a}}\tau,
\label{RCD3a}%
\end{equation}
where $\underset{\smile}{\gamma_{a}}\tau$ means the action of $\underset
{\smile}{\gamma_{a}}$ on the smooth multivector extensor field $\tau.$
\end{theorem}

\emph{(ii)} For all $\upsilon\in\left.  \overset{\ast}{ext}\right.  _{k}%
^{l}\mathcal{V}(U_{0}\cap U):$%
\begin{equation}
\nabla_{a}\upsilon=\partial_{a}\upsilon-\underset{\smile}{\gamma
_{a}^{\bigtriangleup}}\upsilon, \label{RCD3b}%
\end{equation}
where $\underset{\smile}{\gamma_{a}^{\bigtriangleup}}\upsilon$ means the
action of $\underset{\smile}{\gamma_{a}^{\bigtriangleup}}$\ on the smooth
multiform extensor field $\upsilon.$

\begin{proof}
We prove only the property for smooth multivector extensor fields, i.e.,
Eq.(\ref{RCD3a}). Without loss of generality, we check this statement for a
multivector extensor field $(X,\Phi)\mapsto\tau(X,\Phi)$.

Using Eq.(\ref{CDEF1}), and by taking into account the properties just
recalled above, we have%
\begin{align*}
\nabla_{a}\tau(X,\Phi)  &  =\nabla_{a}\tau(X,\Phi)-\tau(\nabla_{a}X,\Phi
)-\tau(X,\nabla_{a}\Phi)\\
&  =\partial_{a}\tau(X,\Phi)+\underset{\smile}{\gamma_{a}}(\tau(X,\Phi))\\
&  -\tau(\partial_{a}X+\underset{\smile}{\gamma_{a}}(X),\Phi)-\tau
(X,\partial_{a}\Phi-\underset{\smile}{\gamma_{a}^{\bigtriangleup}}(\Phi))\\
&  =\partial_{a}\tau(X,\Phi)-\tau(\partial_{a}X,\Phi)-\tau(X,\partial_{a}%
\Phi)\\
&  +\underset{\smile}{\gamma_{a}}(\tau(X,\Phi))-\tau(\underset{\smile}%
{\gamma_{a}}(X),\Phi)+\tau(X,\underset{\smile}{\gamma_{a}^{\bigtriangleup}%
}(\Phi)).
\end{align*}
Then, using once again Eq.(\ref{CDEF1}) and recalling the action of a
generalized operator $\underset{\smile}{\gamma_{a}}$ on a multivector extensor
$\tau$, see Eq.(\ref{GPO18}), we get
\[
(\nabla_{a}\tau)(X,\Phi)=(\partial_{a}\tau)(X,\Phi)+(\underset{\smile}%
{\gamma_{a}}\tau)(X,\Phi),
\]
and the proposition is proved.
\end{proof}

\section{Torsion and Curvature}

\begin{definition}
Let $\left\langle U,\Gamma\right\rangle $ be a parallelism structure on $U$.
The smooth $2$\emph{-covariant vector extensor field} on $U$, defined by
\[
\mathcal{V}(U)\times\mathcal{V}(U)\ni(a,b)\mapsto\tau(a,b)\in\mathcal{V}(U)
\]
such that
\begin{equation}
\tau(a,b)=\nabla_{a}b-\nabla_{b}a-\left[  a,b\right]  , \label{TF1}%
\end{equation}
will be called the \emph{fundamental torsion extensor field\footnote{$\tau$ is
skew-symmetric, i.e.,
\[
\tau\left(  b,a\right)  =-\tau\left(  a,b\right)  .
\]
}} of $\left\langle U,\Gamma\right\rangle $.
\end{definition}

Accordingly, there exists a smooth $(2,1)$\emph{-extensor field }on $U,$
defined by%
\[%
{\textstyle\bigwedge\nolimits^{2}}
\mathcal{V}(U)\ni X^{2}\mapsto\mathcal{T}(X^{2})\in\mathcal{V}(U)
\]
such that%
\begin{equation}
\mathcal{T}(X^{2})=\frac{1}{2}\left\langle \varepsilon^{\mu}\wedge
\varepsilon^{\nu},X^{2}\right\rangle \tau(e_{\mu},e_{\nu}), \label{TF3}%
\end{equation}
where $\left\{  e_{\mu},\varepsilon^{\mu}\right\}  $ is any pair of dual frame
fields on $V\supseteq U.$ It should be emphasized that $\mathcal{T}$, as
extensor field associated with $\tau$, is well-defined since the vector field
$\mathcal{T}(X^{2})$ does not depend on the choice of $\left\{  e_{\mu
},\varepsilon^{\mu}\right\}  .$

It should be remarked that the well known \emph{torsion tensor field, }see,
e.g., \cite{Choquet,rodoliv2006},\emph{ }is just given by%
\[
\mathcal{V}(U)\times\mathcal{V}(U)\times\mathcal{V}^{\ast}(U)\ni
(a,b,\omega)\mapsto T(a,b,\omega)\in\mathcal{S}(U)
\]
such that%
\begin{equation}
T(a,b,\omega)=\left\langle \omega,\tau(a,b)\right\rangle . \label{TF5}%
\end{equation}

Note that it is possible to get $\tau$ in terms of $T,$ i.e.,%
\begin{equation}
\tau(a,b)=T(a,b,\varepsilon^{\mu})e_{\mu}. \label{TF6}%
\end{equation}

It is also possible to introduce a \emph{third torsion extensor field} for
$\left\langle U,\Gamma\right\rangle $ by defining the smooth $(1,2)$%
\emph{-extensor field }on $U$,%
\[
\mathcal{V}^{\ast}(U)\ni\omega\mapsto\Theta(\omega)\in%
{\textstyle\bigwedge\nolimits^{2}}
\mathcal{V}^{\ast}(U),
\]
such that
\begin{equation}
\Theta(\omega)=\frac{1}{2}\left\langle \omega,\tau(e_{\mu},e_{\nu
})\right\rangle \varepsilon^{\mu}\wedge\varepsilon^{\nu}, \label{TF7}%
\end{equation}
We call $\Theta$ the \emph{Cartan torsion extensor field} of $\left\langle
U,\Gamma\right\rangle $.\smallskip

\begin{definition}
The smooth $3$\emph{-covariant extensor vector field} on $U$, defined by
\[
\mathcal{V}(U)\times\mathcal{V}(U)\times\mathcal{V}(U)\ni(a,b,c)\mapsto
\rho(a,b,c)\in\mathcal{V}(U),
\]
such that
\begin{equation}
\rho(a,b,c)=\left[  \nabla_{a},\nabla_{b}\right]  c-\nabla_{\left[
a,b\right]  }c, \label{CF1}%
\end{equation}
will be called the \emph{fundamental curvature extensor field\footnote{As can
be easily verified, $\rho$ is skew-symmetric with respect to the first and the
second variables, i.e.,$\rho(b,a,c)=-\rho(a,b,c).$}} of $\left\langle
U,\Gamma\right\rangle $.
\end{definition}

Thus, there exists a smooth $1$\emph{\ bivector and }$1$\emph{\ vector
variables vector extensor field }on $U$, defined by%
\[%
{\displaystyle\bigwedge\nolimits^{2}}
\mathcal{V}(U)\times\mathcal{V}(U)\ni(X^{2},c)\mapsto\mathcal{R}(X^{2}%
,c)\in\mathcal{V}(U)
\]
such that%
\begin{equation}
\mathcal{R}(X^{2},c)=\frac{1}{2}\left\langle \varepsilon^{\mu}\wedge
\varepsilon^{\nu},X^{2}\right\rangle \rho(e_{\mu},e_{\nu},c), \label{CF3}%
\end{equation}
where $\left\{  e_{\mu},\varepsilon^{\mu}\right\}  $ is any pair of dual frame
fields on $V\supseteq U.$ We note that $\mathcal{R},$ as extensor field
associated with $\rho,$ is well-defined since the vector field $\mathcal{R}%
(X^{2},c)$ does not depend on the choice of $\left\{  e_{\mu},\varepsilon
^{\mu}\right\}  .$

\begin{proposition}
The fundamental curvature field $\rho$ satisfies a \emph{cyclic property},
i.e.,
\[
\rho(a,b,c)+\rho(b,c,a)+\rho(c,a,b)=0.
\]

\end{proposition}

\begin{proof}
Let us take $a,b,c\in\mathcal{V}(U).$ By using Eq.(\ref{CF1}), we can write
\begin{align}
\rho(a,b,c)  &  =\nabla_{a}\nabla_{b}c-\nabla_{b}\nabla_{a}c-\nabla_{ \left[
a,b\right]  }c,\tag{a}\\
\rho(b,c,a)  &  =\nabla_{b}\nabla_{c}a-\nabla_{c}\nabla_{b}a-\nabla_{ \left[
b,c\right]  }a,\tag{b}\\
\rho(c,a,b)  &  =\nabla_{c}\nabla_{a}b-\nabla_{a}\nabla_{c}b-\nabla_{ \left[
c,a\right]  }b. \tag{c}%
\end{align}

By adding Eqs. (a), (b) and (c), we\ have
\begin{align}
&  \rho(a,b,c)+\rho(b,c,a)+\rho(c,a,b)\nonumber\\
&  =\nabla_{a}(\nabla_{b}c-\nabla_{c}b)+\nabla_{b}(\nabla_{c}a-\nabla
_{a}c)+\nabla_{c}(\nabla_{a}b-\nabla_{b}a)\nonumber\\
&  -\nabla_{\left[  a,b\right]  }c-\nabla_{\left[  b,c\right]  }%
a-\nabla_{\left[  c,a\right]  }b, \tag{d}%
\end{align}
but, by taking into account Eq.(\ref{SPS1}), we get
\begin{equation}
\rho(a,b,c)+\rho(b,c,a)+\rho(c,a,b)=\left[  a,\left[  b,c\right]  \right]
+\left[  b,\left[  c,a\right]  \right]  +\left[  c,\left[  a,b\right]
\right]  , \tag{e}%
\end{equation}
whence, by recalling the Jacobi identities for the Lie product of smooth
vector fields, the expected result immediately follows.
\end{proof}

The fundamental curvature field $\rho$ satisfies the so-called
Bianchi\emph{\ identity}, i.e.,%
\[
\nabla_{w}\rho(a,b,c)+\nabla_{a}\rho(b,w,c)+\nabla_{b}\rho(w,a,c)=0.
\]
Note the cycling of letters: $a,b,w\rightarrow b,w,a\rightarrow w,a,b.$ The
proof is a single exercise, see, e.g., \cite{fmcr2}.

\subsection{Symmetric Parallelism Structure}

A parallelism structure $\left\langle U,\Gamma\right\rangle $ is said to be
\emph{symmetric} if and only if for all smooth vector fields $a$ and $b$ on
$U$ it holds
\begin{equation}
\Gamma(a,b)-\Gamma(b,a)=\left[  a,b\right]  , \label{SPS0}%
\end{equation}
i.e.,
\begin{equation}
\nabla_{a}b-\nabla_{b}a=\left[  a,b\right]  . \label{SPS1}%
\end{equation}

Now, according with Eq.(\ref{TF1}), we see that the \emph{condition of
symmetry} is completely equivalent to the \emph{condition of \ null torsion},
i.e.,
\begin{equation}
\tau(a,b)=0. \label{SPS2}%
\end{equation}
So, taking into account Eq.(\ref{TF3}) and Eq.(\ref{TF7}), we also have that
\begin{equation}
\mathcal{T}(X^{2})=0\text{ and }\Theta(\omega)=0. \label{SPS3}%
\end{equation}

\section{%
Appendix\appendix
}

\section{Multivectors and Multiforms}

Let $V$ be a vector space over $\mathbb{R}$ with finite dimension, i.e., $\dim
V=n$ with $n\in\mathbb{N},$ and let $\ V^{\ast}$ be its dual vector space.
Recall that
\begin{equation}
\dim V=\dim V^{\ast}=n. \label{MM1}%
\end{equation}

Let us consider an integer number $k$ with $0\leq k\leq n$. The real vector
spaces of $k$-vectors over $V,$ i.e., the set of skew-symmetric $k$%
-contravariant tensors over $V$, and the real vector spaces of $k$-forms over
$V$, i.e., the set of skew-symmetric $k$-covariant tensors over $V,$ will be
as usually denoted by $%
{\textstyle\bigwedge\nolimits^{k}}
V$ and $%
{\textstyle\bigwedge\nolimits^{k}}
V^{\ast}$, respectively.

We identify, as usual $0$-vectors to real numbers, i.e., $%
{\textstyle\bigwedge\nolimits^{0}}
V=\mathbb{R}$, and $1$-vectors to objects living in $V $, i.e., $%
{\textstyle\bigwedge\nolimits^{1}}
V\simeq V.$ Also, we identify $0$-forms with real numbers, i.e., $%
{\textstyle\bigwedge\nolimits^{0}}
V=\mathbb{R}$, and $1$-forms with objects living in $V^{\ast},$ i.e., $%
{\textstyle\bigwedge\nolimits^{1}}
V^{\ast}=V^{\ast}$. Recall that
\begin{equation}
\dim%
{\textstyle\bigwedge\nolimits^{k}}
V=\dim%
{\textstyle\bigwedge\nolimits^{k}}
V^{\ast}=\binom{n}{k}. \label{MM2}%
\end{equation}

The $0$-vectors, $2$-vectors,\ldots, $(n-1)$-vectors and $n$-vectors are
called scalars, bivectors,\ldots, pseudovectors and pseudoscalars,
respectively. The $0$-forms, $2$-forms,\ldots, $(n-1)$-forms and $n$-forms are
called scalars, biforms,\ldots, pseudoforms and pseudoscalars.

Given a vector space $V$ over the real field $\mathbb{R}$, we define $%
{\textstyle\bigwedge}
V$ as the exterior direct sum%

\[%
{\textstyle\bigwedge}
V=%
{\displaystyle\sum\limits_{r=0}^{n}}
\oplus%
{\textstyle\bigwedge\nolimits^{r}}
V=%
{\displaystyle\bigoplus\nolimits_{r=0}^{n}}
{\textstyle\bigwedge\nolimits^{r}}
V.
\]

To simplify the notation we sometimes write \textit{simply}
\[%
{\textstyle\bigwedge}
V=\mathbb{R}+V+%
{\textstyle\bigwedge\nolimits^{2}}
V+\cdots+%
{\textstyle\bigwedge\nolimits^{n-1}}
V+%
{\textstyle\bigwedge\nolimits^{n}}
V.
\]
As it is well known the set of multivectors over $V$ has a natural structure
of vector space over $\mathbb{R}$ and we have
\begin{align}
\dim%
{\textstyle\bigwedge}
V  &  =\dim\mathbb{R+}\dim V+\dim%
{\textstyle\bigwedge\nolimits^{2}}
V+\cdots+\dim%
{\textstyle\bigwedge\nolimits^{n-1}}
V+\dim%
{\textstyle\bigwedge\nolimits^{n}}
V\nonumber\\
&  =\binom{n}{0}+\binom{n}{1}+\binom{n}{2}+\cdots+\binom{n}{n-1}+\binom{n}%
{n}=2^{n}. \label{MM3}%
\end{align}

An element of $%
{\textstyle\bigwedge}
V$ will be called a multivector over $V$. If $X\in%
{\textstyle\bigwedge}
V$ we write:%
\[
X=X^{0}+X^{1}+X^{2}+\cdots+X^{n-1}+X^{n}.
\]

In what follows we shall need also the vector space $%
{\textstyle\bigwedge}
V^{\ast}=%
{\displaystyle\bigoplus\nolimits_{r=0}^{n}}
{\textstyle\bigwedge\nolimits^{r}}
V^{\ast}$. An element of $%
{\textstyle\bigwedge}
V^{\ast}$ will be called a \textit{multiform over} $V$. As in the case of
multivectors we simply write:
\[%
{\textstyle\bigwedge}
V^{\ast}=\mathbb{R}+V^{\ast}+%
{\textstyle\bigwedge\nolimits^{2}}
V^{\ast}+\cdots+%
{\textstyle\bigwedge\nolimits^{n-1}}
V^{\ast}+%
{\textstyle\bigwedge\nolimits^{n}}
V^{\ast},
\]
and if $\Phi\in%
{\textstyle\bigwedge}
V^{\ast}$ we write%
\[
\Phi=\Phi_{0}+\Phi_{1}+\Phi_{2}+\cdots+\Phi_{n-1}+\Phi_{n}.
\]
Of course, \ $%
{\textstyle\bigwedge}
V^{\ast}$ has a natural structure of real vector space over $\mathbb{R}$. We
have,
\begin{align}
\dim%
{\textstyle\bigwedge}
V^{\ast}  &  =\dim\mathbb{R+}\dim V^{\ast}+\dim%
{\textstyle\bigwedge\nolimits^{2}}
V^{\ast}+\cdots+\dim%
{\textstyle\bigwedge\nolimits^{n-1}}
V^{\ast}+\dim%
{\textstyle\bigwedge\nolimits^{n}}
V^{\ast}\nonumber\\
&  =\binom{n}{0}+\binom{n}{1}+\binom{n}{2}+\cdots+\binom{n}{n-1}+\binom{n}%
{n}=2^{n}. \label{MM4}%
\end{align}

We recall that $%
{\textstyle\bigwedge^{k}}
V$ is also called the homogeneous multivector space (of degree $k$), and to $%
{\textstyle\bigwedge^{k}}
V^{\ast}$ the homogeneous multiform space (of degree $k$).

Let us take an integer number $k$ with $0\leq k\leq n.$ The linear mappings%
\[%
{\textstyle\bigwedge}
V\ni X\mapsto\left\langle X\right\rangle ^{k}\in%
{\textstyle\bigwedge}
V\text{ and }%
{\textstyle\bigwedge}
V^{\ast}\ni\Phi\mapsto\left\langle \Phi\right\rangle _{k}\in%
{\textstyle\bigwedge}
V^{\ast}%
\]
such that if $X=X^{0}+X^{1}+\cdots+X^{n}$ and $\Phi=\Phi_{0}+\Phi_{1}%
+\cdots+\Phi_{n},$ then
\begin{equation}
\left\langle X\right\rangle ^{k}=X^{k}\text{ and }\left\langle \Phi
\right\rangle _{k}=\Phi_{k} \label{MM5}%
\end{equation}
are called the $k$\emph{-part operator }(\emph{for multivectors}) and the
$k$\emph{-part operator }(\emph{for multiforms}), respectively. $\left\langle
X\right\rangle ^{k}$ is read as the $k$\emph{-part of }$X$\emph{\ }and
$\left\langle \Phi\right\rangle _{k}$ is read as the $k$\emph{-part of }%
$\Phi.$

It should be evident that for all $X\in%
{\textstyle\bigwedge}
V$ and $\Phi\in%
{\textstyle\bigwedge}
V^{\ast}:$%
\begin{align}
X  &  =\underset{k=0}{\overset{n}{%
{\textstyle\sum}
}}\left\langle X\right\rangle ^{k},\label{MM6}\\
\Phi &  =\underset{k=0}{\overset{n}{%
{\textstyle\sum}
}}\left\langle \Phi\right\rangle _{k}. \label{MM7}%
\end{align}

The linear mappings%
\[%
{\textstyle\bigwedge}
V\ni X\mapsto\widehat{X}\in%
{\textstyle\bigwedge}
V\text{ and }%
{\textstyle\bigwedge}
V^{\ast}\ni\Phi\mapsto\widehat{\Phi}\in%
{\textstyle\bigwedge}
V^{\ast}%
\]
such that%
\begin{equation}
\left\langle \widehat{X}\right\rangle ^{k}=\left(  -1\right)  ^{k}\left\langle
X\right\rangle ^{k}\text{ and }\left\langle \widehat{\Phi}\right\rangle
_{k}=\left(  -1\right)  ^{k}\left\langle \Phi\right\rangle _{k} \label{MM8}%
\end{equation}
are called the \emph{grade involution operator (for multivectors)} and the
\emph{grade involution operator (for multiforms)}, respectively.

The linear mappings%
\[%
{\textstyle\bigwedge}
V\ni X\mapsto\widetilde{X}\in%
{\textstyle\bigwedge}
V\text{ and }%
{\textstyle\bigwedge}
V^{\ast}\ni\Phi\mapsto\widetilde{\Phi}\in%
{\textstyle\bigwedge}
V^{\ast}%
\]
such that%
\begin{equation}
\left\langle \widetilde{X}\right\rangle ^{k}=\left(  -1\right)  ^{\frac{1}%
{2}k(k-1)}\left\langle X\right\rangle ^{k}\text{ and }\left\langle
\widetilde{\Phi}\right\rangle _{k}=\left(  -1\right)  ^{\frac{1}{2}%
k(k-1)}\left\langle \Phi\right\rangle _{k} \label{MM9}%
\end{equation}
are called the \emph{reversion operator (for multivectors)} and the
\emph{reversion operator (for multiforms)}, respectively.

Both of $%
{\textstyle\bigwedge}
V$ and $%
{\textstyle\bigwedge}
V^{\ast}$ endowed with the exterior product $\wedge$ (of multivectors and
multiforms!) are \emph{associative algebras}, i.e., the \emph{exterior algebra
of multivectors }and the \emph{exterior algebra of multiforms}, respectively.

\section{Exterior Power Extension for Operators}

Let $\lambda$ be a linear operator on $V,$ i.e., a linear map $V\ni
v\mapsto\lambda(v)\in V.$ It can be extended in such a way as to give a linear
operator on $%
{\textstyle\bigwedge}
V,$ namely$\underline{\text{ }\lambda},$ which is defined by%
\[%
{\textstyle\bigwedge}
V\ni X\mapsto\underline{\lambda}(X)\in%
{\textstyle\bigwedge}
V
\]
such that%
\begin{equation}
\underline{\lambda}(X)=\left\langle 1,X\right\rangle +\underset{k=1}%
{\overset{n}{%
{\textstyle\sum}
}}\frac{1}{k!}\left\langle \varepsilon^{j_{1}}\wedge\cdots\wedge
\varepsilon^{j_{k}},X\right\rangle \lambda(e_{j_{1}})\wedge\cdots\wedge
\lambda(e_{j_{k}}), \label{EPO1}%
\end{equation}
where $\left\{  e_{j},\varepsilon^{j}\right\}  $ is any pair of dual bases for
$V$ and $V^{\ast}$.

We emphasize that $\underline{\lambda}$ is a well-defined linear operator on $%
{\textstyle\bigwedge}
V.$ Note that each $k$-vector $\left\langle \varepsilon^{j_{1}}\wedge
\cdots\wedge\varepsilon^{j_{k}},X\right\rangle \lambda(e_{j_{1}})\wedge
\cdots\wedge\lambda(e_{j_{k}})$ with $1\leq k\leq n$ does not depend on the
choice of $\left\{  e_{j},\varepsilon^{j}\right\}  ,$ and the linearity
follows just from the linearity of the duality scalar product. We call
$\underline{\lambda}$ the \emph{Exterior Power Extension (EPE) }of $\lambda$
(to multivector operator).

The EPE of a vector operator $\lambda$ has the following basic properties:

$\underline{\lambda}$ is grade-preserving, i.e.,
\begin{equation}
\text{if }X\in%
{\textstyle\bigwedge\nolimits^{k}}
V,\text{ then }\underline{\lambda}(X)\in%
{\textstyle\bigwedge\nolimits^{k}}
V. \label{EPO2}%
\end{equation}

For all $\alpha\in\mathbb{R},$ $v\in V,$ and $X,Y\in%
{\textstyle\bigwedge}
V:$%
\begin{align}
\underline{\lambda}(\alpha)  &  =\alpha,\label{EPO3}\\
\underline{\lambda}(v)  &  =\lambda(v),\label{EPO4}\\
\underline{\lambda}(X\wedge Y)  &  =\underline{\lambda}(X)\wedge
\underline{\lambda}(Y). \label{EPO5}%
\end{align}

We observe that the four basic properties as given by Eq.(\ref{EPO2}),
Eq.(\ref{EPO3}), Eq.(\ref{EPO4}) and Eq.(\ref{EPO5}) are completely equivalent
to the extension procedure of a vector operator.

Let $\lambda$ be a linear operator on $V^{\ast},$ i.e., a linear map $V^{\ast
}\ni\omega\mapsto\lambda(\omega)\in V^{\ast}.$ It is possible to extend
$\lambda$ in such a way to get a linear operator on $%
{\textstyle\bigwedge}
V^{\ast},$ namely the operator $\underline{\lambda\text{,}}$ defined by%
\[%
{\textstyle\bigwedge}
V^{\ast}\ni\Phi\mapsto\underline{\lambda}(\Phi)\in%
{\textstyle\bigwedge}
V^{\ast},
\]
such that%
\begin{equation}
\underline{\lambda}(\Phi)=\left\langle 1,\Phi\right\rangle +\underset
{k=1}{\overset{n}{%
{\textstyle\sum}
}}\frac{1}{k!}\left\langle e_{j_{1}}\wedge\cdots\wedge e_{j_{k}}%
,\Phi\right\rangle \lambda(\varepsilon^{j_{1}})\wedge\cdots\wedge
\lambda(\varepsilon^{j_{k}}), \label{EPO6}%
\end{equation}
where $\left\{  e_{j},\varepsilon^{j}\right\}  $ is any pair of dual bases for
$V$ and $V^{\ast}$.

We emphasize that $\underline{\lambda}$ is a well-defined linear operator on $%
{\textstyle\bigwedge}
V^{\ast}.$ We call $\underline{\lambda}$ the \emph{EPE }of $\lambda$ (to multiforms).

The EPE of a form operator $\lambda$ has the following basic properties.

$\underline{\lambda}$ is grade-preserving, i.e.,
\begin{equation}
\text{if }\Phi\in%
{\textstyle\bigwedge\nolimits^{k}}
V^{\ast},\text{ then }\underline{\lambda}(\Phi)\in%
{\textstyle\bigwedge\nolimits^{k}}
V^{\ast}. \label{EPO7}%
\end{equation}

For all $\alpha\in\mathbb{R},$ $\omega\in V^{\ast},$ and $\Phi,\Psi\in%
{\textstyle\bigwedge}
V^{\ast}:$
\begin{align}
\underline{\lambda}(\alpha)  &  =\alpha,\label{EPO8}\\
\underline{\lambda}(\omega)  &  =\lambda(\omega),\label{EPO9}\\
\underline{\lambda}(\Phi\wedge\Psi)  &  =\underline{\lambda}(\Phi
)\wedge\underline{\lambda}(\Psi). \label{EPO10}%
\end{align}

The four basic properties given by Eq.(\ref{EPO7}), Eq.(\ref{EPO8})
Eq.(\ref{EPO9}) and Eq.(\ref{EPO10}) are logically equivalent to the extension
procedure of a form operator.

There exists a relationship between the extension procedure of a vector
operator and the extension procedure of a form operator.

Let us take a vector operator (or, a form operator) $\lambda$. As we can see,
the duality adjoint of $\lambda$ is just a form operator (respectively, a
vector operator), and the duality adjoint of $\underline{\lambda}$ is just a
multiform operator (respectively, a multivector operator). It holds that the
duality adjoint of the EPE of $\lambda$ is equal to the EPE of the duality
adjoint of $\lambda,$ i.e.,%
\begin{equation}
\left(  \underline{\lambda}\right)  ^{\bigtriangleup}=\underline{\left(
\lambda^{\bigtriangleup}\right)  }. \label{EPO11}%
\end{equation}

Thus, it is possible to use the more simple notation $\underline{\lambda
}^{\bigtriangleup}$ to mean either $\left(  \underline{\lambda}\right)
^{\bigtriangleup}$ or $\underline{\left(  \lambda^{\bigtriangleup}\right)  }.$

We present some properties for the EPE of an invertible vector operator
$\lambda.$

For all $\Phi\in%
{\textstyle\bigwedge}
V^{\ast},$ and $X\in%
{\textstyle\bigwedge}
V:$%
\begin{align}
\underline{\lambda}\left\langle \Phi,X\right\rangle  &  =\left\langle
\underline{\lambda}^{-\bigtriangleup}(\Phi),\underline{\lambda}%
(X)\right\rangle ,\label{EPO12}\\
\underline{\lambda}\left\langle \Phi,X\right\vert  &  =\left\langle
\underline{\lambda}^{-\bigtriangleup}(\Phi),\underline{\lambda}(X)\right\vert
,\label{EPO13}\\
\underline{\lambda}\left\vert X,\Phi\right\rangle  &  =\left\vert
\underline{\lambda}(X),\underline{\lambda}^{-\bigtriangleup}(\Phi
)\right\rangle . \label{EPO14}%
\end{align}

We present only the proof for the property given by Eq.(\ref{EPO13}), the
other proofs are analogous.

\begin{proof}
Let us take $X\in%
{\textstyle\bigwedge}
V$ and $\Phi,\Psi\in%
{\textstyle\bigwedge}
V^{\ast}.$ A straightforward calculation, using Eq.(\ref{DAE6}),
Eq.(\ref{DCP8}), Eq.(\ref{EPO10}) and Eq.(\ref{EPO7}), yields
\begin{align*}
\left\langle \underline{\lambda}\left\langle \Phi,X\right\vert ,\Psi
\right\rangle  &  =\left\langle \left\langle \Phi,X\right\vert ,\underline
{\lambda}^{\bigtriangleup}(\Psi)\right\rangle =\left\langle X,\widetilde{\Phi
}\wedge\underline{\lambda}^{\bigtriangleup}(\Psi)\right\rangle \\
&  =\left\langle X,\underline{\lambda}^{\bigtriangleup}(\underline{\lambda
}^{-\bigtriangleup}(\widetilde{\Phi})\wedge\Psi)\right\rangle =\left\langle
\underline{\lambda}(X),\widetilde{\underline{\lambda}^{-\bigtriangleup}(\Phi
)}\wedge\Psi)\right\rangle \\
&  =\left\langle \left\langle \underline{\lambda}^{-\bigtriangleup}%
(\Phi),\underline{\lambda}(X)\right\vert ,\Psi\right\rangle ,
\end{align*}
whence, by the non-degeneracy of duality scalar product the result follows.
\end{proof}

We present now some properties for the EPE of an invertible form operator
$\lambda$.

For all $\Phi\in%
{\textstyle\bigwedge}
V^{\ast},$ and $X\in%
{\textstyle\bigwedge}
V:$%
\begin{align}
\underline{\lambda}\left\langle \Phi,X\right\rangle  &  =\left\langle
\underline{\lambda}(\Phi),\underline{\lambda}^{-\bigtriangleup}%
(X)\right\rangle ,\label{EPO15}\\
\underline{\lambda}\left\langle \Phi,X\right\vert  &  =\left\langle
\underline{\lambda}(\Phi),\underline{\lambda}^{-\bigtriangleup}(X)\right\vert
,\label{EPO16}\\
\underline{\lambda}\left\vert X,\Phi\right\rangle  &  =\left\vert
\underline{\lambda}^{-\bigtriangleup}(X),\underline{\lambda}(\Phi
)\right\rangle . \label{EPO17}%
\end{align}

\section{Contracted Extension for Operator}

Let $\gamma$ be a linear operator on $V,$ i.e., a linear map $V\ni
v\mapsto\gamma(v)\in V.$ It can be generalized in such a way to give a linear
operator on $%
{\textstyle\bigwedge}
V,$ namely $\underset{\smile}{\gamma},$ which is defined by
\[%
{\textstyle\bigwedge}
V\ni X\mapsto\underset{\smile}{\gamma}(X)\in%
{\textstyle\bigwedge}
V
\]
such that%
\begin{equation}
\underset{\smile}{\gamma}(X)=\gamma(e_{j})\wedge\left\langle \varepsilon
^{j},X\right\vert , \label{GPO1}%
\end{equation}
where $\left\{  e_{j},\varepsilon^{j}\right\}  $ is any pair of dual bases for
$V$ and $V^{\ast}$.

We note that the multivector $\gamma(e_{j})\wedge\left\langle \varepsilon
^{j},X\right\vert $ does not depend on the choice of $\left\{  e_{j}%
,\varepsilon^{j}\right\}  ,$ and that the linearity of the duality contracted
product implies the linearity of $\underset{\smile}{\gamma}.$ Thus,
$\underset{\smile}{\gamma}$ is a well-defined linear operator on $%
{\textstyle\bigwedge}
V.$ We call $\underset{\smile}{\gamma}$ the \emph{Contracted Extension (CE) of
}$\gamma$ (to multivector operator).

The CE of a vector operator $\gamma$ has the following basic properties.

$\underset{\smile}{\gamma}$ is grade-preserving, i.e.,
\begin{equation}
\text{if }X\in%
{\textstyle\bigwedge\nolimits^{k}}
V,\text{ then }\underset{\smile}{\gamma}(X)\in%
{\textstyle\bigwedge\nolimits^{k}}
V. \label{GPO2}%
\end{equation}

For all $\alpha\in\mathbb{R},$ $v\in V,$ and $X,Y\in%
{\textstyle\bigwedge}
V:$%
\begin{align}
\underset{\smile}{\gamma}(\alpha)  &  =0,\label{GPO3}\\
\underset{\smile}{\gamma}(v)  &  =\gamma(v),\label{GPO4}\\
\underset{\smile}{\gamma}(X\wedge Y)  &  =\underset{\smile}{\gamma}(X)\wedge
Y+X\wedge\underset{\smile}{\gamma}(Y). \label{GPO5}%
\end{align}

The four properties given by Eq.(\ref{GPO2}), Eq.(\ref{GPO3}). Eq.(\ref{GPO4})
and Eq.(\ref{GPO5}) are completely equivalent to the generalization procedure
for vector operators.

Let $\gamma$ be a linear operator on $V^{\ast}$, i.e., a linear map $V^{\ast
}\ni\omega\mapsto\gamma(\omega)\in V^{\ast}$. It is possible to generalize
$\gamma$ in such a way as to get a linear operator on $%
{\textstyle\bigwedge}
V^{\ast},$ namely $\underset{\smile}{\gamma},$ which is defined by%
\[%
{\textstyle\bigwedge}
V^{\ast}\ni\Phi\mapsto\underset{\smile}{\gamma}(\Phi)\in%
{\textstyle\bigwedge}
V^{\ast}%
\]
such that%
\begin{equation}
\underset{\smile}{\gamma}(\Phi)=\gamma(\varepsilon^{j})\wedge\left\langle
e_{j},\Phi\right\vert , \label{GPO6}%
\end{equation}
where $\left\{  e_{j},\varepsilon^{j}\right\}  $ is any pair of dual bases
over $V.$

We emphasize that $\underset{\smile}{\gamma}$ is a well-defined linear
operator on $%
{\textstyle\bigwedge}
V^{\ast}$, and call it the \emph{CE of} $\gamma$ (to a multiform operator).

The generalized of a form operator $\gamma$ has the following basic properties.

$\underset{\smile}{\gamma}$ is grade-preserving, i.e.,
\begin{equation}
\text{if }\Phi\in%
{\textstyle\bigwedge\nolimits^{k}}
V^{\ast},\text{ then }\underset{\smile}{\gamma}(\Phi)\in%
{\textstyle\bigwedge\nolimits^{k}}
V^{\ast}. \label{GPO7}%
\end{equation}

For all $\alpha\in\mathbb{R},$ $\omega\in V^{\ast},$ and $\Phi,\Psi\in%
{\textstyle\bigwedge}
V^{\ast}$ we have
\begin{align}
\underset{\smile}{\gamma}(\alpha)  &  =0,\label{GPO8}\\
\underset{\smile}{\gamma}(\omega)  &  =\gamma(\omega),\label{GPO9}\\
\underset{\smile}{\gamma}(\Phi\wedge\Psi)  &  =\underset{\smile}{\gamma}%
(\Phi)\wedge\Psi+\Phi\wedge\underset{\smile}{\gamma}(\Psi). \label{GPO10}%
\end{align}

The properties given by Eq.(\ref{GPO7}), Eq.(\ref{GPO8}), Eq.(\ref{GPO9}) and
Eq.(\ref{GPO10}) are logically equivalent to the generalization procedure for
\textit{form} operators.

There exists a relationship between the generalization procedure of a vector
operator and the generalization procedure of a form operator.

Let $\gamma$ a vector operator (or, a form operator). As we already know, the
duality adjoint of $\gamma$ is just a form operator (respectively, a vector
operator), and the duality adjoint of $\underset{\smile}{\gamma}$ is just a
multiform operator(respectively, a multivector operator). The duality adjoint
of the CE of $\gamma$ is equal to the generalized of the duality adjoint of
$\gamma,$ i.e.,%
\begin{equation}
\left(  \underset{\smile}{\gamma}\right)  ^{\bigtriangleup}=\underset{\smile
}{\left(  \gamma^{\bigtriangleup}\right)  }. \label{GPO11}%
\end{equation}

It follows that is possible to use a more simple notation, namely
$\underset{\smile}{\gamma}^{\bigtriangleup}$ to mean either $\left(
\underset{\smile}{\gamma}\right)  ^{\bigtriangleup}$ or $\underset{\smile
}{\left(  \gamma^{\bigtriangleup}\right)  }$.

We give some of the main properties of the CE of a vector operator $\gamma.$

For all $\Phi\in%
{\textstyle\bigwedge}
V^{\ast},$ and $X\in%
{\textstyle\bigwedge}
V:$%
\begin{align}
\underset{\smile}{\gamma}\left\langle \Phi,X\right\rangle  &  =-\left\langle
\underset{\smile}{\gamma}^{\bigtriangleup}(\Phi),X\right\rangle +\left\langle
\Phi,\underset{\smile}{\gamma}(X)\right\rangle ,\label{GPO12}\\
\underset{\smile}{\gamma}\left\langle \Phi,X\right\vert  &  =-\left\langle
\underset{\smile}{\gamma}^{\bigtriangleup}(\Phi),X\right\vert +\left\langle
\Phi,\underset{\smile}{\gamma}(X)\right\vert ,\label{GPO13}\\
\underset{\smile}{\gamma}\left\vert X,\Phi\right\rangle  &  =\left\vert
\underset{\smile}{\gamma}(X),\Phi\right\rangle -\left\vert X,\underset{\smile
}{\gamma}^{\bigtriangleup}(\Phi)\right\rangle . \label{GPO14}%
\end{align}

We prove the property given by Eq.(\ref{GPO13}), the other proofs are analogous.

\begin{proof}
Let us take $X\in%
{\textstyle\bigwedge}
V$ and $\Phi,\Psi\in%
{\textstyle\bigwedge}
V^{\ast}.$ A straightforward calculation, by using Eq.(\ref{DAE6}),
Eq.(\ref{DCP8}), Eq.(\ref{GPO7}) and Eq.(\ref{GPO10}), yields%
\begin{align*}
\left\langle \underset{\smile}{\gamma}\left\langle \Phi,X\right\vert
,\Psi\right\rangle  &  =\left\langle \left\langle \Phi,X\right\vert
,\underset{\smile}{\gamma}^{\bigtriangleup}(\Psi)\right\rangle =\left\langle
X,\widetilde{\Phi}\wedge\underset{\smile}{\gamma}^{\bigtriangleup}%
(\Psi)\right\rangle \\
&  =\left\langle X,-\widetilde{\underset{\smile}{\gamma}^{\bigtriangleup}%
(\Phi)}\wedge\Psi+\underset{\smile}{\gamma}^{\bigtriangleup}(\widetilde{\Phi
})\wedge\Psi+\widetilde{\Phi}\wedge\underset{\smile}{\gamma}^{\bigtriangleup
}(\Psi)\right\rangle \\
&  =-\left\langle \left\langle \underset{\smile}{\gamma}^{\bigtriangleup}%
(\Phi),X\right\vert ,\Psi\right\rangle +\left\langle X,\underset{\smile
}{\gamma}^{\bigtriangleup}(\widetilde{\Phi}\wedge\Psi)\right\rangle \\
&  =-\left\langle \left\langle \underset{\smile}{\gamma}^{\bigtriangleup}%
(\Phi),X\right\vert ,\Psi\right\rangle +\left\langle \left\langle
\Phi,\underset{\smile}{\gamma}(X)\right\vert ,\Psi\right\rangle \\
&  =\left\langle -\left\langle \underset{\smile}{\gamma}^{\bigtriangleup}%
(\Phi),X\right\vert +\left\langle \Phi,\underset{\smile}{\gamma}(X)\right\vert
,\Psi\right\rangle ,
\end{align*}
and by the non-degeneracy of duality scalar product, the expected result follows.
\end{proof}

We present some properties for the generalized of a form operator $\gamma.$

For all $\Phi\in%
{\textstyle\bigwedge}
V^{\ast},$ and $X\in%
{\textstyle\bigwedge}
V:$%
\begin{align}
\underset{\smile}{\gamma}\left\langle \Phi,X\right\rangle  &  =\left\langle
\underset{\smile}{\gamma}(\Phi),X\right\rangle -\left\langle \Phi
,\underset{\smile}{\gamma}^{\bigtriangleup}(X)\right\rangle ,\label{GPO15}\\
\underset{\smile}{\gamma}\left\langle \Phi,X\right\vert  &  =\left\langle
\underset{\smile}{\gamma}(\Phi),X\right\vert -\left\langle \Phi,\underset
{\smile}{\gamma}^{\bigtriangleup}(X)\right\vert ,\label{GPO16}\\
\underset{\smile}{\gamma}\left\vert X,\Phi\right\rangle  &  =-\left\vert
\underset{\smile}{\gamma}^{\bigtriangleup}(X),\Phi\right\rangle +\left\vert
X,\underset{\smile}{\gamma}(\Phi)\right\rangle . \label{GPO17}%
\end{align}

\section{Extensors}

Extensors are a new kind of geometrical objects which play a crucial role in
the theory presented here and in what follows the basics of their theory is
described. These objects have been apparently introduced by Hestenes and
Sobczyk in \cite{Hestenes} and some applications of the concept appears in
\cite{Lasenby}, but a rigorous theory was developed later, more details on the
theory of extensors may be found in \cite{fmr507}.

Let $%
{\textstyle\bigwedge_{1}^{\diamond}}
V,\ldots$ and $%
{\textstyle\bigwedge_{k}^{\diamond}}
V$ be $k$ subspaces of $%
{\textstyle\bigwedge}
V$ such that each of them is any sum of homogeneous subspaces of $%
{\textstyle\bigwedge}
V,$ \ and let $%
{\textstyle\bigwedge_{1}^{\diamond}}
V^{\ast},\ldots$ and $%
{\textstyle\bigwedge_{l}^{\diamond}}
V^{\ast}$ be $l$ subspaces of $%
{\textstyle\bigwedge}
V^{\ast}$ such that each of them is any sum of homogeneous subspaces of $%
{\textstyle\bigwedge}
V^{\ast}.$

If $%
{\textstyle\bigwedge^{\diamond}}
V$ is any sum of homogeneous subspaces of $%
{\textstyle\bigwedge}
V,$ a multilinear mapping
\begin{align}
\underset{k\text{-copies}}{\underbrace{%
{\textstyle\bigwedge\nolimits_{1}^{\diamond}}
V\times\cdots\times%
{\textstyle\bigwedge\nolimits_{k}^{\diamond}}
V}}  &  \times\underset{l\text{-copies}}{\underbrace{%
{\textstyle\bigwedge\nolimits_{1}^{\diamond}}
V^{\ast}\times\cdots\times%
{\textstyle\bigwedge\nolimits_{l}^{\diamond}}
V^{\ast}}}\ni(X_{1},\ldots,X_{k},\Phi^{1},\ldots,\Phi^{l})\nonumber\\
&  \mapsto\tau(X_{1},\ldots,X_{k},\Phi^{1},\ldots,\Phi^{l})\in%
{\textstyle\bigwedge\nolimits^{\diamond}}
V \label{E1}%
\end{align}
is called a $k$\emph{\ multivector and }$l$\emph{\ multiform variables
multivector extensor} over $V.$

If $%
{\textstyle\bigwedge\nolimits^{\diamond}}
V^{\ast}$ is any sum of homogeneous subspaces of $%
{\textstyle\bigwedge}
V^{\ast},$ a multilinear mapping
\begin{align}
\underset{k\text{-copies}}{\underbrace{%
{\textstyle\bigwedge\nolimits_{1}^{\diamond}}
V\times\cdots\times%
{\textstyle\bigwedge\nolimits_{k}^{\diamond}}
V}}  &  \times\underset{l\text{-copies}}{\underbrace{%
{\textstyle\bigwedge\nolimits_{1}^{\diamond}}
V^{\ast}\times\cdots\times%
{\textstyle\bigwedge\nolimits_{l}^{\diamond}}
V^{\ast}}}\ni(X_{1},\ldots,X_{k},\Phi^{1},\ldots,\Phi^{l})\nonumber\\
&  \mapsto\tau(X_{1},\ldots,X_{k},\Phi^{1},\ldots,\Phi^{l})\in%
{\textstyle\bigwedge\nolimits^{\diamond}}
V^{\ast} \label{E2}%
\end{align}
is called a $k$\emph{\ multivector and }$l$\emph{\ multiform variables
multiform extensor} over $V.$

The set of all the $k$ multivector and $l$ multiform variables multivector
extensors over $V$ has a natural structure of real vector space, and will be
denoted by the highly suggestive notation $ext(%
{\textstyle\bigwedge\nolimits_{1}^{\diamond}}
V,\ldots,%
{\textstyle\bigwedge\nolimits_{k}^{\diamond}}
V,%
{\textstyle\bigwedge\nolimits_{1}^{\diamond}}
V^{\ast},\ldots,%
{\textstyle\bigwedge\nolimits_{l}^{\diamond}}
V^{\ast};%
{\textstyle\bigwedge\nolimits^{\diamond}}
V)$. When no confusion arises, we use the more simple notation $\left.
\overset{\left.  {}\right.  }{ext}\right.  _{k}^{l}(V)$ for that space.

We obviously have that:
\begin{equation}
\dim\left.  \overset{\left.  {}\right.  }{ext}\right.  _{k}^{l}(V)=\dim%
{\textstyle\bigwedge\nolimits_{1}^{\diamond}}
V\ldots\dim%
{\textstyle\bigwedge\nolimits_{k}^{\diamond}}
V\dim%
{\textstyle\bigwedge\nolimits_{1}^{\diamond}}
V^{\ast}\ldots\dim%
{\textstyle\bigwedge\nolimits_{l}^{\diamond}}
V^{\ast}\dim%
{\textstyle\bigwedge\nolimits^{\diamond}}
V. \label{E3}%
\end{equation}

The set of all the $k$ multivector and $l$ multiform variables multiform
extensors over $V$ has also a natural structure of real vector space, and will
be denoted by $ext(%
{\textstyle\bigwedge\nolimits_{1}^{\diamond}}
V,\ldots,%
{\textstyle\bigwedge\nolimits_{k}^{\diamond}}
V,%
{\textstyle\bigwedge\nolimits_{1}^{\diamond}}
V^{\ast},\ldots,%
{\textstyle\bigwedge\nolimits_{l}^{\diamond}}
V^{\ast};%
{\textstyle\bigwedge\nolimits^{\diamond}}
V^{\ast})$, and when no confusion arises, we use the simple notation $\left.
\overset{\ast}{ext}\right.  _{k}^{l}(V)$ for this space. Also, we have,
\begin{equation}
\dim\left.  \overset{\ast}{ext}\right.  _{k}^{l}(V)=\dim%
{\textstyle\bigwedge\nolimits_{1}^{\diamond}}
V\ldots\dim%
{\textstyle\bigwedge\nolimits_{k}^{\diamond}}
V\dim%
{\textstyle\bigwedge\nolimits_{1}^{\diamond}}
V^{\ast}\ldots\dim%
{\textstyle\bigwedge\nolimits_{l}^{\diamond}}
V^{\ast}\dim%
{\textstyle\bigwedge\nolimits^{\diamond}}
V^{\ast}. \label{E4}%
\end{equation}

In particular, when $%
{\textstyle\bigwedge\nolimits_{1}^{\diamond}}
V=\cdots=%
{\textstyle\bigwedge\nolimits_{k}^{\diamond}}
V=V$ and $%
{\textstyle\bigwedge\nolimits_{1}^{\diamond}}
V^{\ast}=\ldots=%
{\textstyle\bigwedge\nolimits_{l}^{\diamond}}
V^{\ast}=V^{\ast},$ we have two types of elementary extensors over a real
vector space $V,$ namely when $%
{\textstyle\bigwedge\nolimits^{\diamond}}
V=V$ and when $%
{\textstyle\bigwedge\nolimits^{\diamond}}
V^{\ast}=V^{\ast}.$

\subsection{Exterior Product of Extensors}

We define the exterior product of $\tau\in\left.  \overset{\left.  {}\right.
}{ext}\right.  _{k}^{l}(V)$ and $\sigma\in\left.  \overset{\left.  {}\right.
}{ext}\right.  _{r}^{s}(V)$ (or, $\tau\in\left.  \overset{\ast}{ext}\right.
_{k}^{l}(V)$ and $\sigma\in\left.  \overset{\ast}{ext}\right.  _{r}^{s}(V)$)
as $\tau\wedge\sigma\in\left.  \overset{\left.  {}\right.  }{ext}\right.
_{k+r}^{l+s}(V)$ (respectively, $\tau\wedge\sigma\in\left.  \overset{\ast
}{ext}\right.  _{k+r}^{l+s}(V)$) given by%
\begin{align}
&  \tau\wedge\sigma(X_{1},\ldots,X_{k},Y_{1},\ldots,Y_{r},\Phi^{1},\ldots
,\Phi^{l},\Psi^{1},\ldots,\Psi^{s})\nonumber\\
&  =\tau(X_{1},\ldots,X_{k},\Phi^{1},\ldots,\Phi^{l})\wedge\sigma(Y_{1}%
,\ldots,Y_{r},\Psi^{1},\ldots,\Psi^{s}). \label{E5}%
\end{align}
Note that on the right side appears an exterior product of multivectors
(respectively, an exterior product of multiforms).

The duality scalar product of a multiform extensor $\tau\in\left.
\overset{\ast}{ext}\right.  _{k}^{l}(V)$ with a multivector extensor
$\sigma\in\left.  \overset{\left.  {}\right.  }{ext}\right.  _{r}^{s}(V)$ is
the \textit{scalar} extensor $\left\langle \tau,\sigma\right\rangle \in\left.
\overset{\ast}{ext}\right.  _{k+r}^{l+s}(V)$ defined by%
\begin{align}
&  \left\langle \tau,\sigma\right\rangle (X_{1},\ldots,X_{k},Y_{1}%
,\ldots,Y_{r},\Phi^{1},\ldots,\Phi^{l},\Psi^{1},\ldots,\Psi^{s})\nonumber\\
&  =\left\langle \tau(X_{1},\ldots,X_{k},\Phi^{1},\ldots,\Phi^{l}%
),\sigma(Y_{1},\ldots,Y_{r},\Psi^{1},\ldots,\Psi^{s})\right\rangle .
\label{E6}%
\end{align}

The duality left contracted product of a multiform extensor $\tau\in\left.
\overset{\ast}{ext}\right.  _{k}^{l}(V)$ with a multivector extensor
$\sigma\in\left.  \overset{\left.  {}\right.  }{ext}\right.  _{r}^{s}(V)$ (or,
a multivector extensor $\tau\in\left.  \overset{\left.  {}\right.  }%
{ext}\right.  _{k}^{l}(V)$ with a multiform extensor $\sigma\in\left.
\overset{\ast}{ext}\right.  _{r}^{s}(V)$) is the multivector extensor
$\left\langle \tau,\sigma\right\vert \in\left.  \overset{\left.  {}\right.
}{ext}\right.  _{k+r}^{l+s}(V)$ (respectively, the multiform extensor
$\left\langle \tau,\sigma\right\vert \in\left.  \overset{\ast}{ext}\right.
_{k+r}^{l+s}(V)$) defined by%
\begin{align}
&  \left\langle \tau,\sigma\right\vert (X_{1},\ldots,X_{k},Y_{1},\ldots
,Y_{r},\Phi^{1},\ldots,\Phi^{l},\Psi^{1},\ldots,\Psi^{s})\nonumber\\
&  =\left\langle \tau(X_{1},\ldots,X_{k},\Phi^{1},\ldots,\Phi^{l}%
),\sigma(Y_{1},\ldots,Y_{r},\Psi^{1},\ldots,\Psi^{s})\right\vert . \label{E7}%
\end{align}

The duality right contracted product of a multiform extensor $\tau\in\left.
\overset{\ast}{ext}\right.  _{k}^{l}(V)$ with a multivector extensor
$\sigma\in\left.  \overset{\left.  {}\right.  }{ext}\right.  _{r}^{s}(V)$ (or,
a multivector extensor $\tau\in\left.  \overset{\left.  {}\right.  }%
{ext}\right.  _{k}^{l}(V)$ with a multiform extensor $\sigma\in\left.
\overset{\ast}{ext}\right.  _{r}^{s}(V)$) is the multiform extensor
$\left\vert \tau,\sigma\right\rangle \in\left.  \overset{\ast}{ext}\right.
_{k+r}^{l+s}(V)$ (respectively, the multivector extensor $\left\vert
\tau,\sigma\right\rangle \in\left.  \overset{\left.  {}\right.  }{ext}\right.
_{k+r}^{l+s}(V)$) defined by%
\begin{align}
&  \left\vert \tau,\sigma\right\rangle (X_{1},\ldots,X_{k},Y_{1},\ldots
,Y_{r},\Phi^{1},\ldots,\Phi^{l},\Psi^{1},\ldots,\Psi^{s})\nonumber\\
&  =\left\vert \tau(X_{1},\ldots,X_{k},\Phi^{1},\ldots,\Phi^{l}),\sigma
(Y_{1},\ldots,Y_{r},\Psi^{1},\ldots,\Psi^{s})\right\rangle . \label{E8}%
\end{align}

\section{Jacobian Fields}

Let $\left\{  b_{\mu},\beta^{\mu}\right\}  $ and $\left\{  b_{\mu}^{\prime
},\beta^{\mu\prime}\right\}  $ be any two pairs of dual frame fields on the
open sets $U\subseteq M$ and $U^{\prime}\subseteq M,$ respectively.

If the parallelism structure $\left\langle U,B\right\rangle $ is compatible
with the parallelism structure $\left\langle U^{\prime},B^{\prime
}\right\rangle $ (i.e., define the same connection on $U\cap U^{\prime}%
\neq\emptyset$), then we can define a smooth \textit{extensor operator
field}\emph{\ }on $U\cap U^{\prime}$, namely $J$, by%
\[
\mathcal{V}(U\cap U^{\prime})\ni v\mapsto J(v)\in\mathcal{V}(U\cap U^{\prime
}),
\]
such that
\begin{equation}
J(v)=\beta^{\sigma}(v)b_{\sigma}^{\prime}. \label{JF1}%
\end{equation}
It will be called the \emph{Jacobian field} associated with the pairs of frame
fields $\left\{  b_{\mu},\beta^{\mu}\right\}  $ and $\left\{  b_{\mu}^{\prime
},\beta^{\mu\prime}\right\}  $ (in this order!).

Note that in accordance with the above definition the Jacobian field
associated with $\left\{  b_{\mu}^{\prime},\beta^{\mu\prime}\right\}  $ and
$\left\{  b_{\mu},\beta^{\mu}\right\}  $ is $J^{\prime}$, given by%
\[
\mathcal{V}(U\cap U^{\prime})\ni v\mapsto J^{\prime}(v)\in\mathcal{V}(U\cap
U^{\prime}),
\]
such that
\begin{equation}
J^{\prime}(v)=\beta^{\sigma\prime}(v)b_{\sigma}. \label{JF2}%
\end{equation}
It is the \emph{inverse extensor operator }of $J,$ i.e., $J\circ J^{\prime
}(v)=v$ and $J^{\prime}\circ J(v)=v$ for each $v\in\mathcal{V}(U\cap
U^{\prime}).$

We note that
\begin{equation}
J(b_{\mu})=b_{\mu}^{\prime}\text{ and }J^{-1}(b_{\mu}^{\prime})=b_{\mu}.
\label{JF3}%
\end{equation}

Take $a\in\mathcal{V}(U\cap U^{\prime}),$ the $a$-\emph{DCDO} associated with
$\left\langle U,B\right\rangle $ and $\left\langle U^{\prime},B^{\prime
}\right\rangle ,$ namely $\partial_{a}$ and $\partial_{a}^{\prime},$ are
related by
\begin{equation}
\partial_{a}^{\prime}v=J(\partial_{a}J^{-1}(v)). \label{JF4}%
\end{equation}

\begin{proof}
A straightforward calculation, yields
\[
\partial_{a}^{\prime}v=(a\beta^{\mu\prime}(v))b_{\mu}^{\prime}=(a\beta
^{\mu\prime}(v))J(b_{\mu}).
\]
Using the identity $\beta^{\mu}(J^{-1}(v))=\beta^{\mu\prime}(v)$ (valid for
smooth extensor fields) we get
\[
\partial_{a}^{\prime}v=J((a\beta^{\mu}(J^{-1}(v))b_{\mu}),
\]
from where the expected result follows.
\end{proof}

We see that from the definition of deformed parallelism structure,
$\partial_{a}^{\prime}$ is a $J$\emph{-deformation} of $\partial_{a}.$

Then, from Eq.(\ref{DPS4}), we have%

\begin{equation}
\partial_{a}^{\prime}\omega=J^{-\bigtriangleup}(\partial_{a}J^{\bigtriangleup
}(\omega)). \label{JF5}%
\end{equation}

Finally, we note that
\begin{equation}
J^{-\bigtriangleup}(\beta^{\mu})=\beta^{\mu\prime}\text{ and }%
J^{\bigtriangleup}(\beta^{\mu\prime})=\beta^{\mu}. \label{JF6}%
\end{equation}

\end{document}